\documentclass[11pt,a4paper]{article}
\usepackage[utf8]{inputenc}
\usepackage{hyperref}
\usepackage{fullpage}
\usepackage{amssymb}
\usepackage{mathtools}
\usepackage{amsthm}
\usepackage{eucal}
\usepackage{dsfont}
\usepackage[T1]{fontenc}
\usepackage{lmodern}
\usepackage[stretch=10,shrink=10]{microtype}
\usepackage[capitalise]{cleveref}
\usepackage{color,soul}
\usepackage[framemethod=tikz]{mdframed}
\usepackage{lipsum}
\usepackage{tikz-cd}

\allowdisplaybreaks[1]

\theoremstyle{plain}
\newtheorem{theorem}{Theorem}[section]
\newtheorem{lemma}[theorem]{Lemma}

\newtheorem{cor}[theorem]{Corollary}

\theoremstyle{definition}
\newtheorem{definition}[theorem]{Definition}

\newtheorem{claim}[theorem]{Claim}

\newtheorem{remark}[theorem]{Remark}
\newtheorem{fact}[theorem]{Fact}

\newtheorem{task}[theorem]{Task}

\newcommand {\eps} {\varepsilon}
\newcommand {\br} [1] {\ensuremath{ \left( #1 \right) }}
\newcommand {\Br} [1] {\ensuremath{ \left[ #1 \right] }}

\newcommand {\minusspace} {\: \! \!}
\newcommand {\smallspace} {\: \!}
\newcommand {\fn} [2] {\ensuremath{ #1 \minusspace \br{ #2 } }}
\newcommand {\Fn} [2] {\ensuremath{ #1 \minusspace \Br{ #2 } }}
\newcommand {\ball} [2] {\fn{\mathcal{B}^{#1}}{#2}}
\newcommand {\defeq} {\ensuremath{ \stackrel{\mathrm{def}}{=} }}
\newcommand {\mutinf} [2] {\fn{\mathrm{I}}{#1 \smallspace : \smallspace #2}}
\newcommand {\imax} [2] {\fn{\mathrm{I}_{\max}}{#1 \smallspace : \smallspace #2}}
\newcommand {\imaxeps} [2] {\fn{\mathrm{I}^{\varepsilon}_{\max}}{#1 \smallspace : \smallspace #2}}
\newcommand {\imaxdelta} [2] {\fn{\mathrm{I}^{\delta}_{\max}}{#1 \smallspace : \smallspace #2}}
\newcommand {\condmutinf} [3] {\mutinf{#1}{#2 \smallspace \middle\vert \smallspace #3}}
\newcommand {\prob} [1] {\Fn{\Pr}{#1}}
\newcommand {\abs} [1] {\ensuremath{ \left| #1 \right| }}
\newcommand {\norm} [1] {\ensuremath{ \left\| #1 \right\| }}
\newcommand {\normsub} [2] {\ensuremath{ \norm{#1}_{#2} }}
\newcommand {\onenorm} [1] {\normsub{#1}{1}}
\newcommand {\ent} [1] {\fn{\mathrm{H}}{#1}}
\newcommand {\relent} [2] {\fn{\mathrm{D}}{#1 \middle\| #2}}
\newcommand {\dmax} [2] {\fn{\mathrm{D}_{\max}}{#1 \middle\| #2}}
\newcommand {\dmaxeps} [2] {\fn{\mathrm{D}^{\eta}_{\max}}{#1 \middle\| #2}}
\newcommand {\dmaxepss} [3] {\fn{\mathrm{D}^{#3}_{\max}}{#1 \middle\| #2}}

\DeclareMathOperator*{\bigE}{\mathbb{E}}
\newcommand {\expec} [2] {\Fn{\bigE_{\substack{#1}}}{#2}}
\newcommand {\email} [1] {\href{mailto:#1}{\texttt{#1}}}

\newcommand {\bra} [1] {\ensuremath{ \left\langle #1 \right| }}
\newcommand {\ket} [1] {\ensuremath{ \left| #1 \right\rangle }}
\newcommand {\ketbratwo} [2] {\ensuremath{ \left| #1 \middle\rangle \middle\langle #2 \right| }}
\newcommand {\ketbra} [1] {\ketbratwo{#1}{#1}}
\newcommand {\cspace} [1] {\ensuremath{\mathnormal{#1}}}

\newcommand {\Tr} {\ensuremath{ \mathrm{Tr} }}
\newcommand {\partrace} [2] {\fn{\Tr_{#1}}{#2}}
\newcommand {\id} {\ensuremath{\mathrm{I}}}

\newcommand {\suppress}[1]{}

\newcommand {\set} [1] {\ensuremath{ \left\lbrace #1 \right\rbrace }}
\newcommand {\reg} [1] {\ensuremath{ \mathnormal{#1}}}
\newcommand {\Typ}{\mathsf{Typ}}

\def\X{\CMcal{X}}
\def\Y{\CMcal{Y}}

\def\B{\CMcal{B}}

\def\cD{\mathcal{D}}
\def\N{\mathcal{N}}

\def\G{\CMcal{G}}
\def\P{\mathrm{P}}

\def\Q{\CMcal{Q}}

\def\T{\CMcal{T}}

\def\O{\CMcal{O}}

\def\H{\mathcal{H}}
\def\cL{\mathcal{L}}

\def\F{\mathrm{F}}
\def\cF{\mathcal{F}}

\def\inf{\mathrm{inf}}
\def\max{\mathrm{max}}
\def\cI{\mathcal{I}}

\def\EE{\mathbb{E}}
\def\E{\mathcal{E}}
\def\M{\CMcal{M}}

\newcommand{\etaa}{\nu}

\newcommand{\Pos}{\cL_{\geq 0}}

\newcommand{\braket}[2]{\langle#1|#2\rangle}

\newcommand {\mytitle} {Expected communication cost of distributed quantum tasks}
\newcommand {\Anurag} {Anurag Anshu}
\newcommand{\Ankit} {Ankit Garg}
\newcommand{\Aram} {Aram W.~Harrow}
\newcommand{\Penghui} {Penghui Yao}

\newcommand {\CQT} {Centre for Quantum Technologies}
\newcommand {\NUS} {National University of Singapore}
\newcommand {\PRI} {Princeton University}
\newcommand {\MIT} {Center for Theoretical Physics, Massachusetts Institute of Technology}
\newcommand {\QUICS} {Joint Center for Quantum Information and Computer Science, University of Maryland}

\newcommand {\authorblock} [3] {
	\begin{minipage}[t]{0.2\linewidth}
		\centering
		{#1}\\[0.8ex]
		{\footnotesize {#2}\\[-0.7ex]
		\email{#3}}
	\end{minipage}\vspace{1ex}
}

\hypersetup{
  colorlinks={true},
	citecolor={blue},
	pdfstartview={FitH},
	pdfdisplaydoctitle={true},
	breaklinks={true},
	bookmarksopen={true},
	bookmarksnumbered={false},
	pdftitle={\mytitle},
	pdfauthor={\Anurag,\Ankit,\Aram,\Penghui}
}

\begin{document}
\begin{titlepage}
\title{\textbf{\mytitle}\\[2ex]}

\suppress{
\author{
    \authorblock{\Anurag}{\CQT, \NUS}{a0109169@u.nus.edu}
	\authorblock{\Ankit}{\PRI}{garg@cs.princeton.edu}
	\authorblock{\Aram}{\MIT}{aram@mit.edu}
	\authorblock{\Penghui}{\QUICS}{phyao1985@gmail.com}
}
}
\author{
Anurag Anshu\footnote{Centre for Quantum Technologies, National University of Singapore, Singapore. \texttt{a0109169@u.nus.edu}} \qquad
Ankit Garg\footnote{Microsoft Research New England, USA. \texttt{garga@microsoft.com}} \qquad 
Aram W.~Harrow\footnote{Center for Theoretical Physics, Massachusetts Institute of Technology, USA. \texttt{aram@mit.edu}} \qquad Penghui Yao\footnote{Joint Center for Quantum Information and Computer Science, University of Maryland, USA \texttt{phyao1985@gmail.com}} 
\footnote{Part of this work appeared in proceedings of the 11th Conference on the Theory of Quantum Communication and Cryptography (TQC 2016)}}

\clearpage\maketitle
\thispagestyle{empty}

\begin{abstract}
A central question in classical information theory is that of source compression,  which is the task where Alice receives a sample from a known probability distribution and needs to transmit it to the receiver Bob with small error. This problem has a one-shot solution due to Huffman, in which the messages are of variable length and the expected length of the messages matches the asymptotic and i.i.d. compression rate of the Shannon entropy of the source. 
 
In this work, we consider a quantum extension of above task, where Alice receives a sample from a known probability distribution and needs to transmit a part of a pure quantum state (that is associated to the sample) to Bob. We allow entanglement assistance in the protocol, so that the communication is possible through classical messages, for example using quantum teleportation. The classical messages can have a variable length and the goal is to minimize their expected length. We provide a characterization of the expected communication cost of this task, by giving a lower bound that is near optimal up to some additive factors. 

A special case of above task, and the quantum analogue of the source compression problem, is when Alice needs to transmit the whole of her pure quantum state. Here we show that there is no one-shot interactive scheme which matches the asymptotic and i.i.d. compression rate of the von Neumann entropy of the average quantum state. This is a relatively rare case in quantum information theory where the cost of a quantum task is significantly different from its classical analogue. Further, we also exhibit similar results for the fully quantum task of quantum state redistribution, employing some different techniques. We show implications for the one-shot version of the problem of quantum channel simulation.

\end{abstract}
\end{titlepage}
\section{Introduction}\label{sec:introduction}

Shannon, in his celebrated work \cite{Shannon}, initiated the idea of source compression by showing that in the asymptotic and i.i.d. setting  (where i.i.d. refers to independent and identically distributed), compression could be achieved up to the \textit{Shannon entropy} of the message source. The framework introduced in his work has led to the field of `information theory', which encapsulates various revolutionary ideas such as error-correcting codes, cryptography and noisy transmission. In the past few decades, information theory has also permeated into physics, a well studied consequence of which is the quantum information theory. 

\subsection*{Some fundamental tasks in quantum information theory}

A first example of quantum source compression was given by Schumacher~\cite{Schumacher95}, who showed that quantum source could be compressed to the \textit{von Neumann} entropy of the source in the asymptotic and i.i.d. setting. Since then, a large family of quantum source compression tasks and their compression schemes have been discovered. The communication task considered by Schumacher in~\cite{Schumacher95}, which is fundamental to all the subsequent tasks, can be formulated as a problem of {\em classical-quantum state transfer}. We informally define it below in the one-shot setting (which drops the asymptotic and i.i.d. assumption) and point out that it is a \textit{classical-quantum} task, where the communicating parties receive a classical input from an external source.

\vspace{0.1in}

\noindent\textbf{Classical-quantum state transfer:}  Alice (henceforth the sender) receives an input $x$ with probability $p(x)$ associated with a pure quantum state $\ketbra{\Psi^x}$, where $p\br{\cdot}$ is a distribution over a finite set $\X$ and $x\in\X$. The goal is that Bob (henceforth the receiver) outputs a quantum state $\Phi^x$ such that the distance between $\Phi^x$ and $\ketbra{\Psi^x}$, averaged over $x$, is smaller than $\eta \in (0,1)$.

\vspace{0.1in}

The distance between the quantum states will be measured in terms of the purified distance, formally defined in Section \ref{sec:preliminaries}. As mentioned above, the \textit{worst case} quantum communication cost for this task in the asymptotic and i.i.d. setting is characterized by $S(\sum_xp(x)\Psi_x)$~\cite{Schumacher95}, where $S(.)$ denotes the von Neumann entropy. Another example of a classical-quantum task concerns the distribution of a pure bipartite quantum state between Alice and Bob. This task can be viewed as a classical-quantum analogue of the task of quantum state splitting~\cite{AbeyesingheDHW09} and is an extension of the task of classical-quantum state transfer.

\vspace{0.1in}
  
\noindent\textbf{Classical-quantum state splitting:} Alice receives an input $x$ with probability $p(x)$ associated with a bipartite pure quantum state $\ket{\Psi^x}_{A'C}$. The goal for Alice and Bob is to share a state $\Phi^x_{A'C}$ with Alice holding $A'$ and Bob holding $C$, such that the distance between $\Phi^x_{A'C}$ and $\ketbra{\Psi^x}_{A'C}$, averaged over $x$, is smaller than 
$\eta \in (0,1)$.

\vspace{0.1in}

For this task, the worst case quantum communication cost in the asymptotic and i.i.d. setting is characterized by the {\em Holevo information}: $S(\sum_xp(x)\Psi^x_C) - \sum_xS(\Psi^x_C)$~\cite{Holevo73}. 

A natural setting for the above two classical-quantum tasks is quantum communication complexity. Here, the communicating parties Alice and Bob receive inputs $x \in \X$ and $y \in \Y$ and their goal is to compute a joint function $f(x,y)$. They are allowed to communicate quantum messages (the amount of which they wish to minimize) and may or may not have pre-shared entanglement. The two tasks correspond to the instances of a quantum communication protocol where Alice sends her first quantum message to Bob conditioned on her input $x$. 


A limitation of the classical-quantum tasks is that they do not completely capture the information theoretic properties of the quantum systems. This is achieved by the \textit{fully quantum} or coherent quantum tasks, where Alice and Bob are required to maintain quantum correlation with the environment (which we shall refer to as the Reference system). A well known example of such tasks is quantum state merging \cite{horodecki07}, which provided the first operational interpretation to the negativity of quantum conditional information. A generalization of this is the task of \textit{quantum state redistribution},  which was first introduced in \cite{Devatakyard,YardD09} and applied to the setting of quantum communication complexity in \cite{Dave14}. 


\vspace{0.1in}

\noindent\textbf{Quantum state redistribution:} Alice (A,C), Bob(B) and Reference(R) share a joint pure quantum state $\ket{\Psi}_{RBCA}$. Alice needs to transfer the register $C$ to Bob such that the final state between Alice (A), Bob (B,C) and Reference (R) is $\Psi'_{RBCA}$. It is required that the distance between $\Psi'_{RBCA}$ and $\Psi_{RBCA}$ is smaller than $\eps\in (0,1)$.

\vspace{0.1in}

The worst case communication cost of this task in the asymptotic and i.i.d. setting is characterized by the {\em conditional quantum mutual information} $\condmutinf{C}{R}{B}_{\Psi}$~\cite{Devatakyard,YardD09}. All of the above tasks are illustrated in Figure~\ref{fig:stateredist}.

\begin{figure}[ht]
\centering
\begin{tikzpicture}[xscale=0.9,yscale=1.1]

\draw[ultra thick] (-5,11) rectangle (12,1);

\draw[thick] (3.5,11) -- (3.5,6);

\node at (-3,10) {\fbox{\small$\{p(x),\ket{\Psi^x}\}_x$}};
\node at (-3,9) {Alice};
\node at (2,9) {Bob};
\draw[->] (-2,9.7) -- (-2,8.8);
\node at (-1.8,9.2) {$x$};
\draw[->] (2,8.0) -- (2,7.3);
\node at (2,7) {\small\fbox{$\Phi^x\overset{\eta}\approx \Psi^x$}};
\draw[->] (-1.5,8.5) -- (0.7,8.5);
\draw[->] (0.7,8.3) -- (-1.5,8.3);
\draw[->] (-1.5,8.1) -- (0.7,8.1);

\node at (5.5,10) {\fbox{\small$\{p(x),\ket{\Psi^x}_{A'C}\}_x$}};
\node at (5.5,9) {Alice};
\node at (10.5,9) {Bob};
\draw[->] (6.5,9.7) -- (6.5,8.8);
\node at (6.7,9.2) {$x$};

\draw[->](5.8,7.7)--(5.8,7.0);
\draw[->] (10.4,7.7) -- (10.4,7.0);
\draw[thick] (5.5,6.7) rectangle (10.5,6.9);
\node at (5.2,6.8) {$A'$};
\node at (10.8,6.8) {$C$};
\node at (8.5,7.2) {\small$\Phi^x_{A'C}\overset{\eta}\approx \Psi^x_{A'C}$};
\draw[->] (7, 8.5) -- (9.2,8.5);
\draw[->] (9.2,8.3) -- (7,8.3);
\draw[->] (7,8.1) -- (9.2,8.1);

\draw[thick] (-5, 6) -- (12,6);

\draw[thick] (-0.5,5) -- (-2,3) -- (-1.7,2) -- (1,2.5) -- (-0.5,5);
\node at (-0.2,5) {$R$};
\node at (-2.2,3) {$A$};
\node at (-2,2) {$C$};
\node at (1.2,2.5) {$B$};
\node at (-0.5,3) {$\ket{\Psi}$};

\node at (-0.7,5.6) {Reference};
\node at (-2.8,1.6) {Alice};
\node at (1.4,1.6) {Bob};

\draw[->, thick] (2.3,3.5) -- (4.2,3.5);

\draw[thick] (7.5,5) -- (6,2.5) -- (8.8,2) -- (9,3) -- (7.5,5);
\node at (7.8,5) {$R$};
\node at (5.8,2.5) {$A$};
\node at (9.2,3) {$B$};
\node at (9.1,2) {$C$};
\node at (7.8,3) {$\Psi' \overset{\eps}\approx \Psi$};

\node at (7.3,5.6) {Reference};
\node at (5.4,1.6) {Alice};
\node at (9.6,1.6) {Bob};

\end{tikzpicture}
\caption{\small The top left task is classical-quantum state transfer and the top right task is classical-quantum state splitting. The bottom task is quantum state redistribution. In all the cases, Alice and Bob are allowed to pre-share arbitrary entanglement and perform interactive communication.}
 \label{fig:stateredist}
\end{figure}
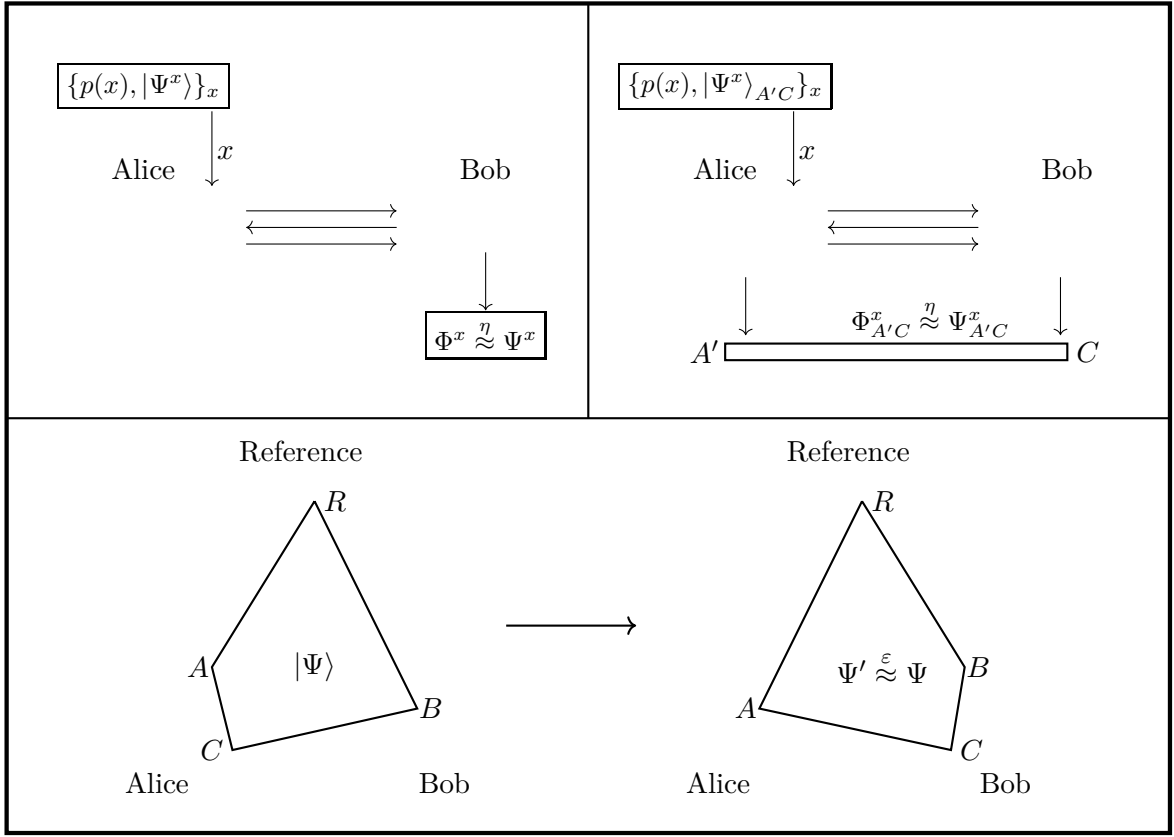

Coherent quantum tasks such as quantum state merging~\cite{Horodecki2005,horodecki07} have also found applications to the problems related to quantum channel coding, through the notion of quantum channel simulation. This task is informally defined as follows.

\vspace{0.1in}

\noindent\textbf{Entanglement-assisted quantum channel simulation:} Given a quantum channel $\E$ with input register $A$ and output register $B$, the goal of Alice and Bob is to simulate the action of $\E$. That is, if Alice is given an input quantum state $\rho_A$, Bob must output a quantum state $\sigma_B$ such that the distance between $\sigma_B$ and $\E(\rho_A)$ is smaller than $\eta\in (0,1)$.
\vspace{0.1in}

The Quantum Reverse Shannon Theorem~\cite{Bennett14, Renner11} states that using the pre-shared entanglement, Alice and Bob can simulate the action of $n$ copies of $\E$ using a number of bits equal to $n$ times the entanglement-assisted classical capacity of $\E$, as $n\rightarrow \infty$. 

\subsection*{Two communication costs of a protocol and their relevance}

There are two measures of the communication cost of a classical or quantum protocol. The worst case communication cost measures the total number of bits or qubits exchanged between the players. The \textit{expected communication cost} measures the expected number of bits that are exchanged between the players. For the case of classical protocols that involve public or private randomness, the expectation is over the distribution of the inputs and all the randomness involved in the protocol. It can also be similarly defined in the classical-quantum or the fully-quantum case, by using quantum teleportation \cite{BennettBCJPW:1993} to convert quantum messages into classical messages and taking expectation over the measurement outcomes. 

The first demonstration of the usefulness of the notion of expected communication cost was shown by Huffman \cite{Huffman52}. For the task of communicating a sample from a probability distribution $p$, he proved that by encoding each message into a codeword of different length, one could construct a code with expected length at most $H(p)+1$. Here, the Shannon entropy $H(\cdot)$ is the worst case communication cost for this task in the asymptotic and i.i.d. setting \cite{Shannon}. This, remarkably, led to an operational interpretation of the Shannon entropy of a source in the one-shot setting. In the subsequent works \cite{HJMR10, bravermanrao11}, a number of elegant one-shot communication protocols were discovered that achieved the near-optimal expected communication costs for their respective tasks. Not only are these results significant in information theory, they have been proved to be useful in classical communication complexity, as discussed below.  

In classical communication complexity, it is very natural to consider interactive protocols, since interaction gives more ability to solve a function~\cite{PONZIO2001323}. This comes at a price that the interactive protocols are hard to compress up to their overall information content (which is also known as the information complexity \cite{braverman12}). To understand this difficulty, observe that compression of any one round protocol incurs an error, which can accumulate over the rounds. If one needs to compress an $r$ round protocol in a round by round fashion, the error in each round must be below $O(1/r)$, which can lead to a large overhead in the worst case communication cost. The notion of expected communication cost serves as the right tool for overcoming this difficulty, as the dependence on error is much weaker and it composes well across several rounds. Indeed, using this notion, the aforementioned results \cite{HJMR10, bravermanrao11} have obtained important applications to direct sum problems in communication complexity. 

The notion of the expected communication cost has largely been left unexplored in the one-shot quantum domain. Protocols for various one-shot classical-quantum tasks \cite{Jain:2005, Jain:2008,AnshuJMSY2014} and fully quantum tasks \cite{Berta09, Renner11, Berta14, AJW17} have been investigated in the past two decades. However, all of these results only consider the worst case quantum communication cost. An exception is the work by Braunstein \textit{et.~al.~}\cite{Braunstein98}, which considered the classical-quantum state transfer task and noted several issues in generalizing directly the techniques of the classical Huffman coding scheme to the quantum setting. Still, recent progress in the field of quantum communication complexity, such as the direct sum result for bounded-round entanglement-assisted quantum communication complexity \cite{Dave14}, has raised the question of compressing quantum protocols in terms of the expected communication cost. 

\subsection*{Our results}

In this work, we investigate the possibility of having quantum protocols with small expected communication cost (that is, close to the worst case quantum communication cost in the asymptotic and i.i.d. setting) for the tasks of classical-quantum state transfer, classical-quantum state splitting and quantum state redistribution. Our main results are that no such compression scheme is possible for either of the three tasks. Using these results, we also give a no-go theorem for the one-shot version of the \textit{Quantum reverse shannon theorem}. 

We refer to a collection of pairs $\{\br{p(x),\Psi^x}\}_x$ as an \textit{ensemble} of states, where $x$ is drawn from a distribution $p\br{\cdot}$ over a domain $\X$. The results obtained in this manuscript are as follows, which hold in the model of entanglement-assisted quantum communication.

\begin{enumerate}
	\item In Section \ref{sec:oneway}, we give a near optimal characterization of the expected communication cost of the task of classical-quantum state splitting, for the protocols that achieve this task with a bounded number of rounds. More precisely, we show that the expected communication cost of any $r$-round interactive entanglement-assisted protocol is lower bounded by $Q(\eta, r)$ (a quantity defined in Definition \ref{def:tightcharac}, where $\eta$ is the average error). Furthermore, there exists a \textit{one-way} protocol that achieves the task of classical-quantum state splitting with the average error $O(\eta)$ and the expected communication cost $Q(\eta, r) + O(r\log Q(\eta, r))$. The additive factor of $O(r\log Q(\eta, r))$ arises due to the prefix-free encoding of the integers (as given in Fact \ref{prefixfree}).
	
	\item Specializing to the sub-case of classical-quantum state transfer, we show a large separation between the expected communication cost and the \textit{quantum information cost} (which is the worst case quantum communication cost in the asymptotic and i.i.d. setting). This shows there exists no coding scheme for quantum messages that performs as well as the Huffman coding scheme \cite{Huffman52} in the classical case. Our main result in Theorem \ref{thm:maingeneral} is as follows. 
	
	Fix some parameter $\delta \in (0,1)$ and an integer $d$. We construct an ensemble $\{\br{p(x),\ketbra{\Psi^x}}\}_x$ such that the quantum states $\ketbra{\Psi^x}$ belong to a $d$-dimensional Hilbert space and the index $x$ takes values over a set $\X$ of size $8d^7$. The important properties of this ensemble are as follows.
	\begin{enumerate}
		\item The quantum information cost, or the von Neumann entropy of the average state $\sum_xp(x)\ketbra{\Psi^x}$, is $\delta\log d + O(1)$. 
		\item Any one-way protocol achieving the classical-quantum state transfer of this ensemble with average error $\eta < \br{\frac{\delta}{8}}^2$ must communicate at least $(1-\delta)\log \br{d\delta} - O(1)$ bits. 
		\item Any interactive protocol achieving the classical-quantum state transfer of this ensemble with average error $\eta < \br{\frac{\delta}{8}}^4$ must communicate at least $\Omega\br{\frac{\log \br{d\delta}}{\log\log d}}$ bits. 
	\end{enumerate}
	Thus, arbitrarily large separations between the information cost and the expected communication cost can be obtained for small enough average error, by choosing a small enough $\delta$. For instance, setting $\delta = \frac{1}{\log d}$ and allowing an average error of at most $O(\frac{1}{\log^4 d})$, the information cost is a constant, whereas the expected communication cost is $\Omega\br{\frac{\log \br{d}}{\log\log d}}$. In contrast, the separation between the information cost and the expected communication cost for the Huffman coding scheme \cite{Huffman52} (which incurs no error) is at most by an additive factor of $1$.

Another property of our construction is that the number of bits of input given to Alice is $\approx 7\log d$. Thus, the lower bound on the expected communication cost is of the order of the input size. This may be contrasted with the well known exponential separations between information and communication \cite{GanorKR14,GanorKR15} and their recent quantum counterpart \cite{ATYY17}, where the lower bound on the communication cost is doubly exponentially smaller than the input size.
	
	\item We show the following result for the fully-quantum task of quantum state redistribution in Theorem \ref{thm:coherentmainagain}. We construct a pure quantum state $\ket{\Psi}_{RBCA}$ such that the expected communication cost of any interactive protocol (achieving the quantum state redistribution of above state with error $\eps$) is at least $\condmutinf{C}{R}{B}_{\Psi}\cdot \br{\frac{1}{\eps}}^{\Omega(1)}$. While the lower bounds on the classical-quantum tasks imply the same lower bounds on their fully-quantum counterparts (as it is harder to maintain coherence with a Reference system), Theorem \ref{thm:coherentmainagain} is stronger than Theorem \ref{thm:maingeneral}, as the former has no dependence on the number of rounds. Furthermore, the techniques used in the proof of Theorem \ref{thm:coherentmainagain} are different from those used in the proof of Theorem \ref{thm:maingeneral}.
	
	\item For the task of one-shot entanglement-assisted simulation of a quantum channel, we show in Theorem \ref{reverseshannonseparationgeneral} that the expected communication cost of simulating the quantum channel can be much larger than the entanglement-assisted classical capacity. This also implies a similar separation between the worst case communication cost of simulating the quantum channel and the entanglement-assisted capacity, as the expected communication cost is smaller than the worst case communication cost. We point out that this is in contrast with the corresponding classical result, since it was shown in \cite{HJMR10} that a classical channel can be simulated with an expected communication cost close to its channel capacity.  
\end{enumerate}

We point out that a part of this work has appeared in~\cite{AGHY16}, where one-way protocols for classical-quantum state transfer are studied (more specifically, Section \ref{sec:oneway}).  The present manuscript is substantially expanded including the treatment of interactive protocols for all the three tasks mentioned above, and a near optimal characterization of the expected communication cost of the task of classical-quantum state splitting. In addition, the present results have exponential improvements on the input size to Alice, which gives robustness to the results. Furthermore, an additional application to the one-shot quantum channel simulation is given.

\subsection*{Our techniques}

To illustrate the ideas in this manuscript, we consider a one-way protocol for the task of classical-quantum state transfer. All the quantum states appearing below are assumed to be in a Hilbert space of dimension $d$. We first show that for every message $i$ sent from Alice to Bob, there exists a quantum state $\sigma_i$ independent of $x$, such that the probability $p^x_i$ of this message conditioned on the input $x$ is upper bounded by $p^x_i\leq 2^{-\dmaxeps{\Psi^x}{\sigma_i}}$, where $\eta$ is the average error and $\dmaxeps{\cdot}{\cdot}$ is the {\em smooth max-relative entropy}. The quantity $Q(\eta, 1)$ essentially expresses the expected communication cost given this constraint. The protocol that achieves the expected communication cost of $Q(\eta,1)$ uses a variant of the convex split lemma \cite{ADJ14} (Lemma \ref{lem:convsplit}) and the classical-quantum rejection sampling approach of \cite{Jain:2005, Jain:2008}.

The upper bound on $p^x_i$ implies that the probability $p_i$ of a message $i$, averaged over $x$, is upper bounded as $p_i \leq \sum_x p(x) 2^{-\dmaxeps{\Psi^x}{\sigma_i}}$.  Next, we construct an ensemble $\{\br{p(x), \Psi^x}\}$ for which the quantity $\sum_x p(x) 2^{-\dmaxeps{\Psi^x}{\sigma_i}}$ is small, as a result of which the expected communication cost is large.

It can be verified that given any quantum state $\sigma_i$, and a random pure state $\ket{\Psi}$ chosen
according to the Haar measure, the smooth max-relative entropy ($=\dmaxeps{\Psi}{\sigma_i}$) attains a large value ($\approx
\log(d)$) with high probability. This suggests that the desired ensemble $\{(p(x), \Psi^x)\}_x$ might be constructed by choosing pure states according to the Haar measure, making the quantity $\sum_x p(x) 2^{-\dmaxeps{\Psi^x}{\sigma_i}}$ close to $\mathcal{O}(1)\cdot
2^{-\log(d)}$. This leads to the upper bound $p_i\leq \frac{\mathcal{O}(1)}{d}$ and hence the expected communication cost is at least $\log(d)-\mathcal{O}(1)$. However, this choice of the ensemble makes the von Neumann entropy of the average state $\sum_x p(x) \Psi^x$ equal to $\log d$, which is too large.

We remedy this issue by setting
$\ket{\Psi^x} = \sqrt{1-\delta}\ket{0}+\sqrt{\delta}\ket{x}$ for a parameter $\delta$, where
$\ket{0}$ is some fixed vector and $\ket{x}$ belongs to the $d-1$ dimensional subspace orthogonal to $\ket{0}$. We choose $\ket{x}$ according to the Haar measure in the $d-1$ dimensional subspace and show that the smooth max-relative entropy $\dmaxeps{\Psi^x}{\sigma}$ is still large ($\approx \log(d\delta)$ with high probability),  while the von Neumann entropy of the average state $\sum_x p(x) \Psi^x$ is $\approx \delta\log(d)$, which is much smaller than the expected communication cost.

The situation as above, where we can find an ensemble with small von Neumann entropy but large smooth max-relative entropy, does not arise when the quantum states in the ensemble are drawn from a fixed orthogonal basis (that is $\{\ket{\Psi^x}\}_x$ are mutually orthogonal). This explains why our lower bound technique does not apply to the classical case. 

For the task of quantum state redistribution, we start with an example of a quantum state $\Psi_{RABC}$ which has small quantum conditional mutual information $\condmutinf{R}{C}{B}_{\Psi}$. Given a communication protocol $\mathcal{P}$ with expected communication cost $C$ that achieves a quantum state redistribution for this quantum state, we apply a non-trace preserving operation on the register $R$ to obtain a new quantum state $\Phi_{RABC}$. This non-trace preserving operation serves to rescale the eigenvalues of $\Psi_R$. We show, exploiting the coherent quantum setting, that the protocol $\mathcal{P}$ still achieves the quantum state redistribution of $\Phi_{RABC}$, with a small increase in error and a bounded amount of the worst case quantum communication cost (upper bounded in terms of $C$). This is where we differ from the classical or classical-quantum settings, where applying such non-trace preserving maps may lead to a large worst case communication cost. To achieve the lower bound on $C$, it suffices to show that the state $\Phi_{RABC}$ requires a large worst case quantum communication cost for interactive protocols. For this, we use lower bounds for the task of quantum state redistribution as obtained in \cite{Berta14}. We find that the resulting lower bound is much larger than $\condmutinf{R}{C}{B}_{\Psi}$.

\section{Preliminaries}
\label{sec:preliminaries}

In this section we present some notations, definitions, facts and lemmas that we will use in our proofs.
\subsection*{Information theory}

For a natural number $n$, let $[n]$ represent the set $\{1,2, \dots,
n\}$. For a real $q>1$, let $[q]$ represent the set $[\lfloor q \rfloor]$. For a set $S$, let $|S|$ be the size of $S$. A \textit{tuple} is
a finite collection of positive integers, such as $(i_1,i_2\ldots
i_r)$ for some finite $r$. Let $\log$ represent the logarithm to the
base $2$ and $\ln$ represent the logarithm to the base $\mathrm{e}$. 

Consider a finite dimensional Hilbert space $\H$ endowed with an inner product $\langle \cdot, \cdot \rangle$ (in this paper, we only consider finite dimensional Hilbert spaces). The $\ell_1$ norm of an operator $X$ acting on $\H$ is
$\onenorm{X}\defeq\Tr\sqrt{X^{\dag}X}$; the $\ell_2$ norm is $\norm{X}_2\defeq\sqrt{\Tr XX^{\dag}}$ and the $\ell_{\infty}$ norm (spectral norm) $\norm{X}_{\infty}$ is the largest singular eigenvalue of $X$. A quantum state (or just a state) is a positive semi-definite matrix acting on $\H$ with trace equal to $1$. It is called {\em pure} if and only if the rank is $1$. Let $\ket{\psi}$ be a unit vector.  We use $\psi$ to represent the state
and also the density matrix  $\ketbra{\psi}$, associated with $\ket{\psi}$. 

A sub-normalized state is a positive semi-definite matrix with trace less than or equal to
$1$. A {\em quantum register} $A$ is associated with some Hilbert space $\H_A$. Define
$|A| \defeq \dim(\H_A)$. We denote by $\cD(\H_A)$, the set of quantum states in the
Hilbert space $\H_A$ and by $\cD_{\leq}(\H_A)$, the set of all sub-normalized states on
register $A$. Writing state $\rho$ with subscript $A$ indicates $\rho_A \in \cD(\H_A)$. The identity operator in Hilbert space $\H_A$ (and associated register $A$) is denoted by $\id_A$. The set of all linear operators associated to $\H_A$ is represented by $\cL(\H_A)$ and the set of all positive semi-definite operators is represented by $\Pos(\H_A)$.

 For two quantum states $\rho$ and $\sigma$, $\rho\otimes\sigma$ represents the tensor product (Kronecker product) of $\rho$ and $\sigma$.  The composition of two registers $A$ and $B$, denoted $AB$, is associated with the Hilbert space $\H_A \otimes \H_B$. If two registers $A,B$ are associated with the same Hilbert space, we shall denote it by $A\equiv B$. Let $\rho_{AB}$ be a bipartite quantum state in registers $AB$.  We define
\[ \rho_{\reg{B}} \defeq \partrace{\reg{A}}{\rho_{AB}}
\defeq \sum_i (\bra{i} \otimes \id_{\cspace{B}})
\rho_{AB} (\ket{i} \otimes \id_{\cspace{B}}) , \]
where $\set{\ket{i}}_i$ is an orthonormal basis for the Hilbert space $\cspace{A}$. The state $\rho_B$ is referred to as the marginal state of $\rho_{AB}$ in register $B$. Unless otherwise stated, a missing register from the subscript in a state will represent the partial trace over that register. A quantum map $\E: \cL(\H_A)\rightarrow \cL(\H_B)$ is a completely positive and trace preserving (CPTP) linear map (mapping states from $\mathcal{D}(\H_A)$ to states in $\mathcal{D}(\H_B)$). A completely positive and trace non-increasing linear map $\tilde{\E}:\cL(\H_A)\rightarrow \cL(\H_B)$ maps quantum states to sub-normalized states. A quantum measurement on a register $A$ is a collection of operators $\{M_1, M_2, \ldots\}$ such that $M_i \in \cL(\H_A)$ and $\sum_i M^{\dagger}_iM_i = \id_A$. If a quantum measurement is performed on a quantum state $\rho_A$, the post-measured state corresponding to the outcome $i$ is $\frac{M_i\rho_A M^{\dagger}_i}{\Tr(M_i\rho_A M^{\dagger}_i)}$. A {\em unitary} operator $U_A:\H_A \rightarrow \H_A$ is such that $U_A^{\dagger}U_A = U_A U_A^{\dagger} = \id_A$. An {\em isometry}  $V:\H_A \rightarrow \H_B$ is such that $V^{\dagger}V = \id_A$ and $VV^{\dagger} = \Pi_B$ where $\Pi_B$ is a projector on $\H_B$. The set of all unitary operations on register $A$ is  denoted by $\mathcal{U}(\H_A)$. 


\begin{definition}
We shall consider the following information theoretic quantities. Let $\varepsilon \in (0,1)$. 
\begin{enumerate}
\item {\bf Generalized fidelity:} (\cite{uhlmann76}, see also \cite{Tomamichel12}) For $\rho,\sigma \in \mathcal{D}_{\leq}(\H_A)$, $$\F(\rho,\sigma)\defeq \onenorm{\sqrt{\rho}\sqrt{\sigma}} + \sqrt{(1-\Tr(\rho))(1-\Tr(\sigma))}.$$ 
\item {\bf Purified distance:} (\cite{GilchristLN05}) For $\rho,\sigma \in \mathcal{D}_{\leq}(\H_A)$, $$\P(\rho,\sigma) = \sqrt{1-\F^2(\rho,\sigma)}.$$
\item {\bf $\varepsilon$-ball:} For $\rho_A\in \mathcal{D}(\H_A)$, $$\ball{\eps}{\rho_A} \defeq \{\rho'_A\in \mathcal{D}(\H_A)|~\F(\rho_A,\rho'_A) \geq 1-\varepsilon\}. $$ 
\item {\bf Von Neumann entropy:} (\cite{Neumann32}) For $\rho_A\in \mathcal{D}(\H_A)$, $$\ent{A}_{\rho} \defeq - \Tr(\rho_A\log\rho_A) .$$ 
\item {\bf Relative entropy:} (\cite{umegaki1954}) For $\rho_A,\sigma_A\in \mathcal{D}(\H_A)$, $$\relent{\rho_A}{\sigma_A} \defeq \Tr(\rho_A\log\rho_A) - \Tr(\rho_A\log\sigma_A) .$$ 
\item {\bf Max-relative entropy:} (\cite{Datta09}) For $\rho_A,\sigma_A\in \mathcal{D}(\H_A)$, $$ \dmax{\rho_A}{\sigma_A}  \defeq  \inf \{ \lambda \in \mathbb{R} : 2^{\lambda} \sigma_A \geq \rho_A \}  .$$ 
\item {\bf Smooth max-relative entropy:} (\cite{Datta09}, see also \cite{Jain:2009}) For $\rho_A,\sigma_A\in \mathcal{D}(\H_A)$, $$ \dmaxeps{\rho_A}{\sigma_A}  \defeq  \inf_{\rho'_A\in\ball{\eta}{\rho_A}}\dmax{\rho'_A}{\sigma_A} .$$
\item {\bf Mutual information:}  For $\rho_{AB}\in \mathcal{D}(\H_{AB})$,$$\mutinf{A}{B}_{\rho}  \defeq \relent{\rho_{AB}}{\rho_A\otimes\rho_B}= \ent{A}_{\rho} + \ent{B}_{\rho} - \ent{AB}_{\rho}.$$
\item {\bf Max-information:} (\cite{Renner11}) For $\rho_{AB}\in \mathcal{D}(\H_{AB})$, $$ \imax{A}{B}_{\rho} \defeq   \inf_{\sigma_{B}\in \mathcal{D}(B)}\dmax{\rho_{AB}}{\rho_{A}\otimes\sigma_{B}} .$$
\item {\bf Smooth max-information:} (\cite{Renner11}) For $\rho_{AB}\in \mathcal{D}(\H_{AB})$, $$\imaxeps{A}{B}_{\rho} \defeq \inf_{\rho'\in \ball{\eps}{\rho}} \imax{A}{B}_{\rho'} .$$	
\end{enumerate}
\label{def:infquant}
\end{definition}	
We will use the following facts. 

\begin{fact}[Triangle inequality for purified distance,~\cite{GilchristLN05}]
\label{fact:trianglepurified}
For states $\rho^1_A, \rho^2_A, \rho^3_A \in \mathcal{D}(\H_A)$,
$$\P(\rho^1_A,\rho^3_A) \leq \P(\rho^1_A,\rho^2_A)  + \P(\rho^2_A,\rho^3_A) . $$ 
\end{fact}

\suppress{\begin{fact}[Purified distance and trace distance, ~\cite{Tomamichel12}, Proposition 3.3]
\label{fact:purifiedtrace}
For subnormalized states $\rho^1_A,\rho^2_A \in \cD(\H_A)$
$$\frac{1}{2}\|\rho^1_A-\rho^2_A\|_1\leq \P(\rho^1_A,\rho^2_A) \leq \sqrt{\|\rho^1_A-\rho^2_A\|_1}.$$
\end{fact}
}
\begin{fact}[Uhlmann's theorem,~\cite{uhlmann76}]
\label{uhlmann}
Let $\rho_A,\sigma_A\in \mathcal{D}(\H_A)$. Let $\ket{\rho}_{AB}$ be a purification of $\rho_A$ and $\ket{\sigma}_{AC}$ be a purification of $\sigma_A$. There exists an isometry $V: \H_C \rightarrow \H_B$ such that,
 $$\F(\ketbra{\theta}_{AB}, \ketbra{\rho}_{AB}) = \F(\rho_A,\sigma_A) ,$$
 where $\ket{\theta}_{AB} = (I_A \otimes V) \ket{\sigma}_{AC}$.
\end{fact}

\begin{fact}[Monotonicity of quantum operations, \cite{lindblad75, barnum96}]
\label{fact:monotonequantumoperation}
For states $\rho_A$, $\sigma_A \in \mathcal{D}(\H_A)$, and quantum map $\E: \cL(\H_A)\rightarrow \cL(\H_B)$, 
$$\onenorm{\E(\rho_A) - \E(\sigma_A)} \leq \onenorm{\rho_A - \sigma_A} , \F(\rho_A,\sigma_A) \leq \F(\E(\rho_A),\E(\sigma_A))  , \dmax{\rho_A}{\sigma_A} \geq \dmax{\E(\rho_A)}{\E(\sigma_A)}.$$
\end{fact}

Following fact implies the Pinsker's inequality.

\begin{fact}[Lemma 5, \cite{Jain:2003a}]
\label{pinsker}
For quantum states $\rho_A,\sigma_A\in\cD(\H_A)$, 
$$\F(\rho_A,\sigma_A) \geq 2^{-\frac{1}{2}\relent{\rho_A}{\sigma_A}} \geq 2^{-\frac{1}{2}\dmax{\rho_A}{\sigma_A}}.$$
\end{fact}

\begin{fact}
\label{scalarpurified}
Let $\rho_A,\sigma_A\in \mathcal{D}(\H_A)$ be quantum states. Let $\alpha<1$ be a positive real number. If $\P(\alpha\rho_A,\alpha\sigma_A)\leq \eps$, then $$\P(\rho_A,\sigma_A)\leq \eps\sqrt{\frac{2}{\alpha}}.$$
\end{fact}

\begin{proof}
$\P(\alpha\rho_A,\alpha\sigma_A)\leq \eps$ implies $\F(\alpha\rho_A,\alpha\sigma_A)\geq \sqrt{1-\eps^2}\geq 1-\eps^2$. But, 
$\F(\alpha\rho_A,\alpha\sigma_A)= \alpha\|\sqrt{\rho_A}\sqrt{\sigma_A}\|_1+(1-\alpha)$. Thus, $$\F(\rho_A,\sigma_A)=\|\sqrt{\rho_A}\sqrt{\sigma_A}\|_1 \geq 1-\frac{\eps^2}{\alpha},$$ which leads to $\P(\rho_A,\sigma_A)\leq \sqrt{1-(1-\frac{\eps^2}{\alpha})^2}\leq \sqrt{\frac{2\eps^2}{\alpha}}$.
\end{proof}

\begin{fact}[Joint concavity of fidelity, \cite{Watrouslecturenote}, Proposition 4.7]\label{fact:fidelityconcave}
Given quantum states \\$\rho^1_A,\rho^2_A\ldots\rho^k_A,\sigma^1_A,\sigma^2_A\ldots\sigma^k_A \in \mathcal{D}(\H_A)$ and positive numbers $p_1,p_2\ldots p_k$ such that $\sum_ip_i=1$. Then $$\F(\sum_ip_i\rho^i_A,\sum_ip_i\sigma^i_A)\geq \sum_ip_i\F(\rho^i_A,\sigma^i_A).$$
\end{fact}

\begin{fact}[Alicki-Fannes inequality, \cite{fannes73}]
\label{fact:fannes}
Given quantum states $\rho^1_A,\rho^2_A\in \mathcal{D}(\H_A)$, such that $\P(\rho^1_A,\rho^2_A)= \eps \leq \frac{1}{2\mathrm{e}}$, it holds that $$|S(\rho^1_A)-S(\rho^2_A)|\leq \eps\log|A|+1.$$   
\end{fact}

\begin{fact}[Concavity of entropy, \cite{Watrouslecturenote}, Theorem 10.9]
\label{entropyconcave}
For quantum states $\rho^1_A,\rho^2_A\ldots \rho^n_A\in \cD(\H_A)$, and positive real numbers $\lambda_1,\lambda_2\ldots \lambda_n$ satisfying $\sum_i \lambda_i=1$, it holds that  $$S(\sum_i \lambda_i\rho^i_A)\geq \sum_i\lambda_iS(\rho^i_A).$$
\end{fact}

\begin{fact}[Subadditivity of entropy, \cite{LiebAraki70}]
\label{subadditive}
For a quantum state $\rho_{AB}\in \mathcal{D}(\H_{AB})$, it holds that  $|S(\rho_A)-S(\rho_B)|\leq S(\rho_{AB})\leq S(\rho_A)+S(\rho_B)$.
\end{fact}

\begin{fact}
\label{informationbound}
For a quantum state $\rho_{ABC}\in \cD(\H_{ABC})$, it holds that $$\mutinf{A}{C}_{\rho}\leq 2S(\rho_C),$$ $$\condmutinf{A}{C}{B}_{\rho}\leq \mutinf{AB}{C}_{\rho}\leq 2S(\rho_C).$$  
\end{fact}
\begin{proof}
From Fact \ref{subadditive}, $\mutinf{A}{C}_{\rho} = S(\rho_A)+S(\rho_C)-S(\rho_{AC}) \leq 2S(\rho_{C})$. 
\end{proof}

\begin{fact}
\label{fact:cqimax}
For a \textit{classical-quantum} state $\rho_{AB}\in \cD(\H_{AB})$ of the form $\rho_{AB}=\sum_j p(j)\ketbra{j}_A\otimes \sigma^j_B$, it holds that $\imax{A}{B}_{\rho}\leq \log(|B|)$.
\end{fact}
\begin{proof}
By definition, $\imax{A}{B}_{\rho}\leq \dmax{\rho_{AB}}{\rho_A\otimes\frac{\text{I}_B}{|B|}}$.  Also, $$\rho_{AB}=\sum_j p(j)\ketbra{j}_A\otimes \sigma^j_B \preceq |B|\sum_j p(j)\ketbra{j}_A\otimes \frac{\text{I}_B}{|B|} = |B| \rho_A\otimes \frac{\text{I}_B}{|B|}.$$ Thus, the fact follows.
\end{proof}

\begin{fact}
\label{cqmutinf}
For a \textit{classical-quantum} state $\rho_{ABC}\in \cD(\H_{ABC})$ of the form $\sum_j p(j)\ketbra{j}_A\otimes \rho^j_{BC}$, it holds that 
$$\mutinf{AB}{C}_{\rho}\geq \sum_j p(j)\mutinf{B}{C}_{\rho^j}.$$ 
\end{fact}
\begin{proof}
From the definition of mutual information, we have  
\begin{eqnarray*}
\mutinf{AB}{C}_{\rho}&=& S(\rho_{AB}) + S(\rho_C) - S(\rho_{ABC})\\ &=& S(\sum_j p(j)\ketbra{j}_A\otimes\rho^j_B) + S(\sum_j p(j)\rho^j_C) - S(\sum_j p(j)\ketbra{j}_A\otimes\rho^j_{BC}) \\&=& \sum_j p(j)S(\rho^j_B)+S(\sum_j p(j)\rho^j_C) - \sum_j p(j)S(\rho^j_{BC}) \\&\geq&  \sum_j p(j)S(\rho^j_B)+\sum_j p(j)S(\rho^j_C) - \sum_j p(j)S(\rho^j_{BC}) \quad (\text{Fact \ref{entropyconcave}})\\ &=& \sum_j p(j)\mutinf{B}{C}_{\rho^j}
\end{eqnarray*}
\end{proof}

\begin{fact}[Matrix Hoeffding Bound, \cite{Tropp2012}]
\label{matrixhoeff}
Let $Z_1,Z_2\ldots Z_r$ be independent and identically distributed random $d \times d$ Hermitian matrices with $\mathbb{E}(Z_i) = 0$ and $\norm{Z_i}_{\infty}\leq \lambda$. Then $$\text{Prob}(\|\frac{1}{r}\sum_i Z_i\|_{\infty}\geq \eps) \leq d\cdot e^{-\frac{n\eps^2}{8\lambda}}$$ and 
$$\text{Prob}(\|\frac{1}{r}\sum_i Z_i\|_1\geq \eps) \leq d\cdot  e^{-\frac{n\eps^2}{8d^2\lambda}}.$$

\end{fact}

\begin{fact}[Weyl's inequality, \cite{Bhatia96}]
\label{weylinequality}
Let $M,H,P$ be three matrices such that $M=H+P$. Let eigenvalues of $M,H,P$ arranged in descending order be $\{m_1,m_2,\ldots m_n\}$, $\{h_1,h_2,\ldots h_n\}$ and $\{p_1,p_2,\ldots p_n\}$ respectively. Then it holds that 
$$h_i+p_n \leq m_i \leq h_i + p_1.$$
\end{fact}

\begin{fact}[Prefix-free encoding, \cite{Elias75}]
\label{prefixfree}
For every integer $n\geq 0$, there exists an encoding of $n$ into a string of length $\lceil\log n\rceil + 2\lceil\log\log n\rceil +1$ in a prefix-free manner. 
\end{fact}

\begin{fact}[Hoeffding-Chernoff bound, \cite{Hoeff63}, see also \cite{Goemans15}]
\label{fact:chernoff}
Let $X_1 , \ldots, X_n$ be independent random variables, with each $X_i \in [0,1]$ always. Let $X \defeq X_1 + \cdots + X_n$ and  $\mu \defeq \expec{}{X}= \expec{}{X_1} + \cdots + \expec{}{X_n}$. Then for any $\varepsilon>0$,
\begin{align*}
\prob{X \geq (1+\varepsilon)\mu} &\leq \exp\left(-\frac{\varepsilon^2}{2+\varepsilon}\mu\right) \\
\end{align*}
\end{fact}

\section{Classical-quantum state splitting task}
\label{sec:oneway}

We begin with a formal definition of the classical-quantum state splitting task.
\begin{task}[\textbf{Classical-quantum state splitting}] \label{def:entsharing}
 Fix registers $A'C$, the associated Hilbert space $\H_{A'C}$ and an $\eta\in (0,1)$. Alice receives an input $x\sim p\br{\cdot}$ associated with a bipartite pure quantum state $\ket{\Psi^x}_{A'C}\in \cD(\H_{A'C})$, where $p\br{\cdot}$ is a distribution over a finite set $\X$ and $x\in\X$. The goal for Alice and Bob is to share a state $\Phi^x_{A'C}\in \cD(\H_{A'C})$, with Alice holding the register $A'$ and Bob holding the register $C$, such that $\sum_x p(x)\F^2(\Psi^x_{A'C},\Phi^x_{A'C})\geq 1-\eta^2$.
\end{task}

The parameter $\eta$ appearing above is referred to as the average error. In this section, we give a near optimal characterization of the expected communication cost of Task \ref{def:entsharing}. 

\begin{definition} 
\label{def:tightcharac}
Fix a finite set $\X$, $\eta\in (0,1)$ and integer $r\geq 1$. Consider an ensemble $\{(p(x), \ketbra{\Psi^x}_{A'C})\}_{x\in \X}$, where $\Psi^x_{A'C} \in \cD(\H_{A'C})$. Let $T_r$ be the set of all tuples $(i_r,i_{r-1}, \ldots, i_1)$ of positive integers. Let $Q(\eta,r)$ be defined as 
\begin{align}
Q(\eta,r) &\defeq \underset{\stackrel{\{\omega_C^{i_r,i_{r-1}, \ldots, i_1}\}_{(i_r, i_{r-1}, \ldots, i_1) \in T_r}}{\stackrel{\{\eps^x_{i_r,i_{r-1}, \ldots, i_1}\}_{(i_r, i_{r-1}, \ldots, i_1) \in T_r, x\in \X}}{ \{p^x_{i_r,i_{r-1}, \ldots, i_1}\}_{(i_r, i_{r-1}, \ldots, i_1) \in T_r, x\in \X}}}}{\inf} \sum_{x\in \X, (i_r, i_{r-1}, \ldots, i_1) \in T_r}p(x) p^x_{i_r,i_{r-1}, \ldots, i_1}\log\left(i_r\cdot i_{r-1}\cdots i_1\right) \nonumber\\
\text{s.t.} &\nonumber \\
& \sum_x p(x) \sum_{i_r, i_{r-1}, \ldots, i_1} p^x_{i_r,i_{r-1}, \ldots, i_1} \left(\eps^x_{i_r,i_{r-1}, \ldots, i_1}\right)^2 \leq \eta^2;  \nonumber\\
&\forall x \in \X\nonumber\\
  &\hspace{2cm} \sum_{i_r, i_{r-1}, \ldots, i_1}  p^x_{i_r,i_{r-1}, \ldots, i_1} = 1; \nonumber\\
&\forall (i_r, i_{r-1}, \ldots, i_1) \in T_r, x\in \X \nonumber\\
&\hspace{2cm}0\leq p^x_{i_r,i_{r-1}, \ldots, i_1} \leq 2^{- \dmaxepss{\Psi^x_C}{\omega_C^{i_r,i_{r-1}, \ldots, i_1}}{\eps^x_{i_r,i_{r-1}, \ldots, i_1}}} ,\nonumber\\
&\hspace{2cm}\omega_C^{i_r,i_{r-1}, \ldots, i_1} \in \cD(\H_C),  \quad \eps^x_{i_r,i_{r-1}, \ldots, i_1} \in (0,1). \nonumber
\end{align}

\end{definition}
The crucial aspect of the above definition is the upper bound on the probabilities $p^x_{i_r,i_{r-1}, \ldots, i_1}$. While this definition involves a complex optimization problem, it has two utilities. First is that it is a near optimal characterization of the expected communication cost, as shown below. Second is that it can be lower bounded by a simple quantity that will be crucial in the proof of Theorem \ref{thm:maingeneral}. 

Using this definition, we prove the following theorem which gives a near optimal characterization of the expected communication cost of the classical-quantum state splitting task, for a given ensemble. The statement of the theorem uses the notion of `entanglement-assisted interactive protocol', which will be defined in the following subsections.

\begin{theorem}
\label{theo:tightcharac}
Fix a finite set $\X$, $\eta\in (0,1)$ and an integer $r\geq 1$. Consider an ensemble $\{(p(x), \ketbra{\Psi^x}_{A'C})\}_{x\in \X}$, where $\Psi^x_{A'C} \in \cD(\H_{A'C})$. For any $r$-round entanglement-assisted interactive protocol for the classical-quantum state splitting task of above ensemble with average error $\eta$, the expected communication cost is lower bounded by $$Q(\eta, r).$$ Further, for every $\delta\in (0,1)$, there exists an entanglement-assisted one-way protocol that achieves the classical-quantum state splitting task of the above ensemble with average error $\eta+3\sqrt{\delta}$ and expected communication cost $$Q(\eta,r) + 2r\log Q(\eta,r) + 4r + 2\log\frac{4}{\delta}.$$
\end{theorem}
\begin{proof}
The first part of the theorem follows from Lemma \ref{lem:twoway}. The second part of the theorem follows from Lemma \ref{lem:expecachieve}. 
\end{proof}

The additive factor of $2r\log Q(\eta, r) + 4r$ arises due to the prefix-free encoding of the integers (Fact \ref{prefixfree}). To illustrate the key ideas involved in the first part of the theorem, we will first prove the lower bound for the one-way case ($r=1$).

\subsection{Lower bound on the one way protocols}
A general one-way quantum communication protocol $\mathcal{P}_1$ for the classical-quantum state splitting task with average error $\eta$ is described as follows.

\bigskip
\begin{mdframed}
\bigskip
Alice holds a register $A$ and Bob holds a register $B$, with pre-shared entanglement $\ketbra{\theta}_{AB} \in \cD(\H_{AB})$. Alice is given input $x\in \X$ with probability $p(x)$.

\begin{itemize}
\item Conditioned on input $x$, Alice applies the measurement defined by a set of POVM elements $\{M^x_1,M^x_2\ldots\}$ (where $M^x_i \in \cL(\H_A)$ and $\sum_i \br{M^x_i}^{\dagger}M^x_i=\id_A$)  and sends the outcome $i$ to Bob.  Let $p^x_i\defeq \Tr(\br{M^x_i}^{\dagger}M^x_i\theta_A)$ be the probability of outcome $i$. 
\item Alice applies a quantum operation $\E_i: \cL(\H_A)\rightarrow \cL(\H_{A'})$ on her registers to produce the register $A'$. Bob applies a quantum operation $\cF_i: \cL(\H_B)\rightarrow \cL(\H_C)$ on his registers to produce the register $C$. Let $\sigma^{x,i}_{A'C}\in \cD(\H_{A'C})$ be the resulting quantum state. 
\item The final state is $\sum_i p^x_i\sigma^{x,i}_{A'C} \in \cD(\H_{A'C}) $, which satisfies 
\begin{equation}
\label{eq:onewaycorrect}
\sum_xp(x)\bra{\Psi^x}_{A'C}\br{\sum_i p^x_i\sigma^{x,i}_{A'C}}\ket{\Psi^x}_{A'C}\geq 1-\eta^2
\end{equation}
 due to the correctness of the protocol.
\end{itemize}
\bigskip
\end{mdframed}
\centerline{Protocol $\mathcal{P}_1$}
\bigskip

The following lemma shows that for every entanglement-assisted one-way protocol, there exist quantities that are feasible for the optimization problem in the definition of $Q(\eta, 1)$. Recall that $T_1$ (setting $r=1$ in Definition \ref{def:tightcharac}) is the set of all positive integers.

\begin{lemma}
\label{lem:oneway}
Given protocol $\mathcal{P}_1$, there exist positive reals $\{\eps^x_i\}_{i\in T_1, x\in \X},\{p^x_i\}_{i\in T_1, x\in \X}$ and quantum states $\set{\omega^i_C}_{i\in T_1}$ such that 
$$\sum_x p(x) \sum_i p^x_i \left(\eps^x_i\right)^2 \leq \eta^2, \quad \eps^x_i \in (0,1)$$ and
$$\sum_i p^x_i =1, \quad p^x_i \leq 2^{-\dmaxepss{\Psi^x_C}{\omega^i_C}{\eps^x_i}}.$$
 Further, the expected communication cost of $\mathcal{P}_1$ is lower bounded by $\sum_x p(x)\sum_i p^x_i\log(i).$
\end{lemma}
\begin{proof}
We set $p^x_i$ to be the probabilities as given in the definition of the protocol $\mathcal{P}_1$. Let $\eps^x_i \defeq \P(\ketbra{\Psi^x}_{A'C}, \sigma^{x,i}_{A'C}) =  \sqrt{1-\left( \bra{\Psi^x}_{A'C}\sigma^{x,i}_{A'C}\ket{\Psi^x}_{A'C}\right)^2}$. Eq.~\eqref{eq:onewaycorrect} ensures that $$\sum_x p(x) \sum_i p^x_i \left(\eps^x_i\right)^2 \leq \eta^2.$$ Define $\omega_C^i \defeq \E_i(\theta_B)$. Consider,
$$\theta_B=\Tr_A(\br{M^x_i}^{\dagger}M^x_i\theta_{AB}) + \Tr_A((\text{I}-\br{M^x_i}^{\dagger}M^x_i)\theta_{AB})\succeq \Tr_A(\br{M^x_i}^{\dagger}M^x_i\theta_{AB})=p_i^x\rho^{x,i}_B,$$ where we have defined $\rho^{x,i}_B \defeq \frac{\Tr_A(\br{M^x_i}^{\dagger}M^x_i\theta_{AB})}{p_i^x}$. Thus, $\rho^{x,i}_B \preceq \frac{1}{p^x_i}\theta_B.$ By the definition of max-relative entropy, this implies that $2^{\dmax{\rho^{x,i}_B}{\theta_B}}\leq\frac{1}{p^x_i}$. Now we use the monotonicity of max-relative entropy under quantum operations (Fact \ref{fact:monotonequantumoperation}) to obtain 
\begin{equation}
\label{onewayprobupperbound}
p^x_i\leq 2^{-\dmax{\rho^{x,i}_B}{\theta_B}}\leq 2^{-\dmax{\sigma^{x,i}_B}{\E_i(\theta_B)}} \leq 2^{-\dmaxepss{\Psi^x_C}{\omega^i_C}{\eps^x_i}}.
\end{equation}
The last inequality follows since $\sigma^{x,i}_C \in \ball{\eps^x_i}{\Psi^x_C}$. The expected communication cost of $\mathcal{P}_1$ is $$\sum_x p(x)\sum_i p^x_i\lceil\log(i)\rceil \geq \sum_x p(x)\sum_i p^x_i\log(i).$$ This completes the proof. 
\end{proof} 

\subsection{Lower bound for the interactive protocols}
\label{sec:interactive}

Now we extend the arguments in the previous subsection to the interactive communication setting. 

Alice is given input $x$ with probability $p(x)$. A general $r$-round interactive protocol $\mathcal{P}_2$ (where $r$ is an odd number) with average error $\eta$ is described below. We assume that Alice and Bob only perform local quantum measurements and exchange classical bits. This is without loss of generality due to quantum teleportation, as they are allowed to pre-share arbitrary entanglement. It may be noted that Bob's operations do not depend on $x$. The protocol has been graphically represented in Figure \ref{fig:interactive}.
\bigskip
\begin{mdframed}
\bigskip
\label{multirounddescription}
Alice holds a register $A$ and Bob holds a register $B$, with pre-shared entanglement $\ketbra{\theta}_{AB} \in \cD(\H_{AB})$. Alice is given an input $x\in \X$ with probability $p(x)$.

\begin{itemize}
\item Alice performs a measurement $\M=\{M^{x,1}_{A},M^{x,2}_{A}\ldots \}$, where $M^{x,i}_A \in \cL(\H_A)$ and $\sum_i \br{M_A^{x,i}}^{\dagger}M_A^{x,i}=\id_{A}$. The probability of outcome $i_1$ is $p^x_{i_1}\defeq\Tr M^{x,i_1}_{A}\theta_{A}\br{M^{x,i_1}_{A}}^{\dagger}$. Let $\phi^{x,i_1}_{AB}\in \cD(\H_{AB})$ be the post-measured quantum state, conditioned on the outcome $i_1$. She sends the message $i_1$ to Bob.
 
\item Upon receiving the message $i_1$ from Alice, Bob performs a measurement $$\M^{i_1}=\{M^{1,i_1}_{B},M^{2,i_1}_{B},\ldots\},$$ 
where $M^{1,i_1}_{B}\in \cL(\H_B)$ and $\sum_i \br{M_B^{i,i_1}}^{\dagger}M_B^{i,i_1}=\id_B$. The probability of outcome $i_2$ is $$p^x_{i_2|i_1}\defeq \Tr M^{i_2,i_1}_{B}\phi^{x,i_1}_{B}\br{M_B^{i_2,i_1}}^{\dagger}.$$ Let $\phi^{x,i_2,i_1}_{AB}\in \cD(\H_{AB})$ be the post-measured state conditioned on the outcome $i_2$. Bob sends the message $i_2$ to Alice. 

\item For an odd round $k$, let the
post-measured state conditioned on the measurement outcomes $i_1,i_2\ldots i_{k-1}$ in the previous rounds be $\phi^{x,i_{k-1},i_{k-2},\ldots,i_1}_{AB}$. Alice performs the measurement $$\M^{x,i_{k-1},i_{k-2},\ldots,i_2,i_1}=\{M^{x,1,i_{k-1},i_{k-2},\ldots,i_2,i_1}_{A},M^{x,2,i_{k-1},i_{k-2},\ldots,i_2,i_1}_{A},\ldots\},$$ where $M^{x,i_k,i_{k-1},i_{k-2},\ldots,i_2,i_1}_{A} \in \cL(\H_A)$ and $$\sum_{i_k}\br{M^{x,i_k,i_{k-1},i_{k-2},\ldots,i_2,i_1}_{A}}^{\dagger}M^{x,i_k,i_{k-1},i_{k-2},\ldots,i_2,i_1}_{A} = \id_A.$$ She obtains the outcome $i_k$ with probability $$p^x_{i_k|i_{k-1},i_{k-2},\ldots,i_2,i_1}\defeq\Tr M^{x,i_k,i_{k-1},i_{k-2},\ldots,i_2,i_1}_{A}\phi^{x,i_{k-1},i_{k-2},\ldots,i_1}_{A}\br{M^{x,i_k,i_{k-1},i_{k-2},\ldots,i_2,i_1}_{A}}^{\dagger}.$$ Let the post-measured state with the outcome $i_k$ be $\phi^{x,i_k,i_{k-1},i_{k-2},\ldots,i_1}_{AB}\in \cD(\H_{AB})$. Alice sends the outcome $i_k$ to Bob. 

\item For an even round $k$, let the post-measured state with outcomes $i_1,i_2\ldots i_{k-1}$ in the previous rounds be $\phi^{x,i_{k-1},i_{k-2},\ldots,i_1}_{AB}$. Bob performs the measurement $$\M^{i_{k-1},i_{k-2},\ldots,i_2,i_1}=\{M^{1,i_{k-1},i_{k-2},\ldots,i_2,i_1}_{B},M^{2,i_{k-1},i_{k-2},\ldots,i_2,i_1}_{B},\ldots\},$$ where
$M^{i_k,i_{k-1},i_{k-2},\ldots,i_2,i_1}_{B} \in \cL(\H_B)$ and $$\sum_{i_k}\br{M^{i_k,i_{k-1},i_{k-2},\ldots,i_2,i_1}_{B}}^{\dagger}M^{i_k,i_{k-1},i_{k-2},\ldots,i_2,i_1}_{B} = \id_B.$$
He obtains the outcome $i_k$ with probability $$p^x_{i_k|i_{k-1},i_{k-2},\ldots,i_2,i_1}\defeq\Tr M^{i_k,i_{k-1},i_{k-2},\ldots,i_2,i_1}_{B}\phi^{x,i_{k-1},i_{k-2},\ldots,i_1}_{B}\br{M^{i_k,i_{k-1},i_{k-2},\ldots,i_2,i_1}_{B}}^{\dagger}.$$ Let the post-measured state with the outcome $i_k$ be $\phi^{x,i_k,i_{k-1},i_{k-2},\ldots,i_1}_{AB}\in \cD(\H_{AB})$. Bob sends the outcome $i_k$ to Alice. 

\item After receiving the message $i_r$ from Alice at the end of round $r$, Bob applies a unitary $U_{i_r,i_{r-1},\ldots,i_1}:\H_B\rightarrow \H_{B'C}$ such that $B\equiv B'C$. Define $$\ket{\tau^{x,i_r,i_{r-1},\ldots,i_1}}_{AB'C}\defeq U_{i_r,i_{r-1},\ldots,i_1}\ket{\phi^{x,i_r,i_{r-1},\ldots,i_1}}_{AB}.$$

\item For every $k\leq r$, define $$p^x_{i_1,i_2\ldots i_k}\defeq p^x_{i_1}\cdot p^x_{i_2|i_1}\cdot p^x_{i_3|i_2,i_1}\ldots p^x_{i_k|i_{k-1},i_{k-2},\ldots,i_1}.$$ 
\item Alice outputs the register $A'$ (where $A\equiv A'D$, for some register $D$) and Bob outputs register $C$. The final state in registers $A'C$, averaged over all the measurement outcomes, is $\Phi^x_{A'C}\defeq\sum_{i_r,i_{r-1},\ldots,i_1}p^x_{i_1,i_2,\ldots,i_r}\tau^{x,i_r,i_{r-1},\ldots,i_1}_{A'C}$, which satisfies  $\sum_x p(x)\F^2(\Phi^x_{A'C},\Psi^x_{A'C})\geq 1-\eta^2$.
\end{itemize}
\bigskip
\end{mdframed}
\centerline{Protocol $\mathcal{P}_2$}
\bigskip

\begin{remark}
\label{nontrivialmessage}
In the above protocol, we can assume, without loss of generality, that each measurement has at least $2$ outcomes. This is because a measurement with just one outcome is a local isometry. This local isometry can be merged with the next measurement by the same party, as the local operations of Alice and Bob commute. 
\end{remark}

\begin{remark}
\label{abortconvention}
It is possible that either of the parties abort the protocol after some number of rounds, conditioned on the input $x$ or the previous messages. Given an input $x$, let $(i_1,i_2\ldots i_k)$ be a sequence of messages after which Alice aborts (if $k$ is odd) or Bob aborts (if $k$ is even). We can extend this sequence to $(i_1, i_2, \ldots i_k, 1,1, \ldots 1)$, by appending $r-k$ `ones'. By Remark \ref{nontrivialmessage}, it is a unique sequence of messages, since any non-aborting round $k$ requires $i_k>1$. As a result, we can assume that all the sequences of the messages (including those with an abort) are of length $r$.
 
\end{remark}

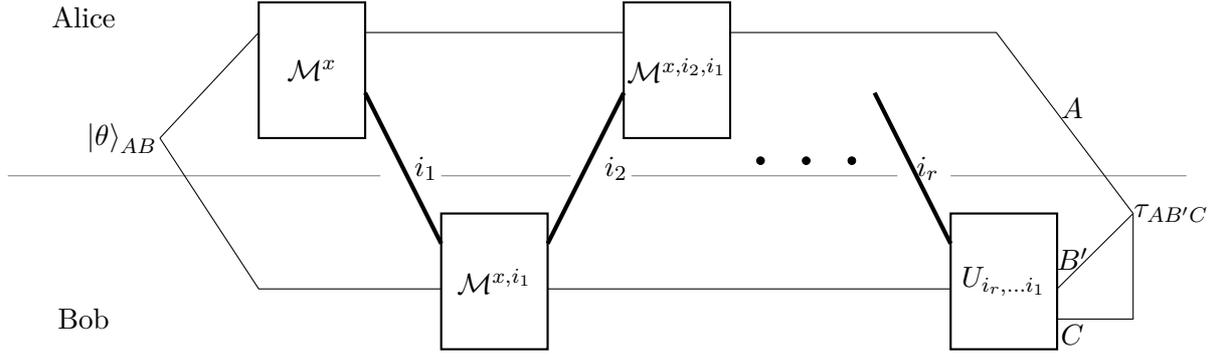
\begin{figure}
\centering
\begin{tikzpicture}

\node at (1.5,7) {Alice};
\node at (1.5,3) {Bob};
\draw [gray] (0.5,4.9) -- (5.4,4.9); 
\draw [gray] (6.2,4.9) -- (7.9,4.9); 
\draw [gray] (8.7,4.9) -- (12.2,4.9); 
\draw [gray] (12.9,4.9) -- (16,4.9);

\node at (2,5.4) {$\ket{\theta}_{AB}$};
\draw (2.5,5.4) -- (3.8,6.8);
\draw (5.2,6.8) -- (8.6,6.8);
\draw (10,6.8) -- (13.5,6.8) -- (15.3,4.4);

\draw (2.5,5.4) -- (3.8,3.4) -- (6.2,3.4);
\draw (7.6,3.4) -- (12.9,3.4);
\draw (14.3,3.4) -- (15.3,4.4);
\draw (14.3,3.0) -- (15.3,3.0) -- (15.3,4.4);

\node at (15.8,4.4) {$\tau_{AB'C}$};
\node at (14.5,2.8) {$C$};
\node at (14.5,3.8) {$B'$};
\node at (14.5,5.8) {$A$};

\draw [thick] (3.8,7.2) rectangle (5.2,5.4);
\node at (4.5,6.3) {$\M^x$};
\draw [ultra thick] (5.2,6) -- (6.2,4.0);
\node at (6,5) {$i_1$};
\draw [thick] (6.2,2.6) rectangle (7.6,4.4);
\node at (6.9,3.5) {$\M^{x,i_1}$};
\draw [ultra thick] (7.6,4.0) -- (8.6,6);
\node at (8.5,5) {$i_2$};
\draw [thick] (8.6,7.2) rectangle (10,5.4);
\node at (9.3,6.3) {$\M^{x,i_2,i_1}$};
\draw [ultra thick] (11.9,6) -- (12.9,4.0);
\node at (12.6,5) {$i_r$};
\draw [thick] (12.9,2.6) rectangle (14.3,4.4);
\node at (13.6,3.5) {$U_{i_r,\ldots i_1}$};

\draw [fill] (10.4,5.1) circle [radius = 0.05];
\draw [fill] (11,5.1) circle [radius = 0.05];
\draw [fill] (11.6,5.1) circle [radius = 0.05];

\end{tikzpicture}
\label{fig:interactive}
\caption{Graphical representation of the interactive protocol for classical-quantum state splitting. The output register is $A'C$.}
\end{figure}

\begin{definition}\label{def:expectedcc}
	The expected communication cost of the protocol $\mathcal{P}_2$ is the expected length of the messages over all probability outcomes, where the expectation is taken over the distribution of the input and the internal randomness of the protocol. Thus, it can be expressed as follows. 
	\begin{eqnarray}
	\label{eq:interactiveexpec}
	&&\sum_x p(x) \bigg(\sum_{i_1}p^x_{i_1}\lceil\log(i_1)\rceil + \sum_{i_1,i_2}p^x_{i_1}p^x_{i_2|i_1}\lceil\log(i_2)\rceil+\cdots+\sum_{i_1,i_2,\ldots,i_r}p^x_{i_1,i_2\ldots i_{r-1}}p^x_{i_r|i_{r-1},i_{r-2},\ldots, i_1}\lceil\log(i_r)\rceil\bigg)\nonumber\\ && = 
	\sum_x p(x)\sum_{i_1,i_2,\ldots, i_r}p^x_{i_1,i_2,\ldots,i_r}\left(\lceil\log(i_1)\rceil+\lceil\log(i_2)\rceil+\cdots+\lceil\log(i_r)\rceil\right)
	\end{eqnarray}
	
\end{definition}

Now we prove the following lemma, which is analogous to Lemma \ref{lem:oneway} in the one-way case. 
\begin{lemma}
\label{lem:twoway}
Fix a protocol $\mathcal{P}_2$ and average error $\eta \in (0,1)$. There exists a collection of quantum states $\{\omega_C^{i_r,i_{r-1}, \ldots, i_1}\}_{(i_r, i_{r-1}, \ldots, i_1) \in T_r}$ and non-negative reals $\{\eps^x_{i_r,i_{r-1}, \ldots, i_1}\}_{(i_r, i_{r-1}, \ldots, i_1) \in T_r, x\in \X}$, \\ $\{p^x_{i_r,i_{r-1}, \ldots, i_1}\}_{(i_r, i_{r-1}, \ldots, i_1) \in T_r, x\in \X}$ such that 
$$p^x_{i_r,i_{r-1}, \ldots, i_1} \leq 2^{- \dmaxepss{\Psi^x_C}{\omega_C^{i_r,i_{r-1}, \ldots, i_1}}{\eps^x_{i_r,i_{r-1}, \ldots, i_1}}},  \sum_{i_r, i_{r-1}, \ldots, i_1}  p^x_{i_r,i_{r-1}, \ldots, i_1} = 1 , \forall x \in \X$$ and 
$$ \sum_x p(x) \sum_{i_r, i_{r-1}, \ldots, i_1} p^x_{i_r,i_{r-1}, \ldots, i_1} \left(\eps^x_{i_r,i_{r-1}, \ldots, i_1}\right)^2 \leq \eta^2, \quad \eps^x_{i_r,i_{r-1}, \ldots, i_1} \leq 1.$$ Furthermore, the expected communication cost of $\mathcal{P}_2$ is lower bounded by 
$$\sum_xp(x)\sum_{i_1,i_2,\ldots,i_r}p^x_{i_r,i_{r-1},\ldots, i_1}\log(i_1\cdot i_2\cdots i_r).$$
\end{lemma}

\begin{proof}
Let $p^x_{i_r, i_{r-1},\ldots,i_1}$ be as given in the description of the protocol $\mathcal{P}_2$. Define 
$$\eps^x_{i_r, i_{r-1}, \ldots, i_1} \defeq \P(\Psi^x_{A'C},\tau^{x,i_r,i_{r-1},\ldots,i_1}_{A'C}) = \sqrt{1- \bra{\Psi^x}_{A'C}\tau^{x,i_r,i_{r-1},\ldots,i_1}_{A'C}\ket{\Psi^x}_{A'C}}.$$ It follows from the definition of the average error $\eta$ that 
$$ \sum_x p(x) \sum_{i_r, i_{r-1}, \ldots, i_1} p^x_{i_r,i_{r-1}, \ldots, i_1} \left(\eps^x_{i_r,i_{r-1}, \ldots, i_1}\right)^2 \leq \eta^2.$$ In order to define $\omega_C^{i_r, i_{r-1} \ldots i_1}$ and prove the upper bound on $p^x_{i_r, i_{r-1},\ldots,i_1}$, we consider the following chain of L\"owener inequalities. It involves recursively removing the operators on Alice's registers (which depend on $x$) by adding positive semidefinite operators. 
\begin{eqnarray*}
&&\Tr_A\br{\phi_{AB}^{x,i_r,i_{r-1},\ldots,i_1}}\\
&=&\frac{\Tr_A\br{M^{x,i_r,i_{r-1},\ldots,i_1}_AM^{i_{r-1},\ldots, i_1}_B\ldots M^{x,i_1}_A\theta_{AB}\br{M^{x,i_1}_A}^{\dagger}\ldots \br{M^{i_{r-1},\ldots,i_1}_B}^{\dagger}\br{M^{x,i_r,i_{r-1},\ldots,i_1}_A}^{\dagger}}}{p^x_{i_r,i_{r-1},\ldots,i_1}} \\
&\preceq& \frac{\sum_{i_r'}\Tr_A\br{M^{x,i_r',i_{r-1},\ldots,i_1}_AM^{i_{r-1},\ldots,i_1}_B\ldots M^{x,i_1}_A\theta_{AB}\br{M^{x,i_1}_A}^{\dagger}\ldots \br{M^{i_{r-1},\ldots, i_1}_B}^{\dagger}\br{M^{x,i_r',i_{r-1},\ldots,i_1}_A}^{\dagger}}}{p^x_{i_r,i_{r-1},\ldots,i_1}} \\ && (\text{Adding positive matrices}) \\ 
&=& \frac{\Tr_A\br{M^{i_{r-1},\ldots, i_1}_BM^{x,i_{r-2},\ldots, i_1}_A\ldots M^{x,i_1}_A\theta_{AB}\br{M^{x,i_1}_A}^{\dagger}\ldots M^{x,i_{r-2},\ldots, i_1\dagger}_A\br{M^{i_{r-1},\ldots, i_1}_B}^{\dagger}}}{p^x_{i_r,i_{r-1},\ldots,i_1}} \\ && (\text{Using the cyclicity of trace and the identity } \sum_{i'_r} \br{M^{x,i_r',i_{r-1},\ldots,i_1}_A}^{\dagger}M^{x,i_r',i_{r-1},\ldots,i_1}_A = \id_A)\\
&\preceq& \frac{\sum_{i_{r-2}'}\Tr_A\br{M^{i_{r-1},\ldots,i_1}_BM^{i_{r-2}',\ldots,i_1}_A\ldots M^{i_1}_A\theta_{AB}\br{M^{i_1}_A}^{\dagger}\ldots M^{i_{r-2}',\ldots,i_1\dagger}_A\br{M^{i_{r-1},\ldots,i_1}_B}^{\dagger}}}{p^x_{i_r,i_{r-1},\ldots,i_1}} \\ && (\text{Index }i'_{r-2}\text{ is different from the index }i_{r-2})\\
&=& \frac{\sum_{i_{r-2}'}\Tr_A\br{M^{i_{r-1},\ldots,i_1}_BM^{i_{r-2}',\ldots,i_1\dagger}_AM^{i_{r-2}',\ldots,i_1}_A\ldots M^{i_1}_A\theta_{AB}\br{M^{i_1}_A}^{\dagger}\ldots \br{M^{i_{r-1},\ldots,i_1}_B}^{\dagger}}}{p^x_{i_r,i_{r-1},\ldots,i_1}} \\ &&(\text{Alice's operations and Bob's operations commute, and using the cyclicity of trace}) \\&=& \frac{\Tr_A\br{M^{i_{r-1},\ldots,i_1}_BM^{i_{r-3},\ldots,i_1}_B\ldots M^{i_1}_A\theta_{AB}\br{M^{i_1}_A}^{\dagger}\ldots M^{i_{r-3},\ldots,i_1\dagger}_B\br{M^{i_{r-1},\ldots,i_1}_B}^{\dagger}}}{p^x_{i_r,i_{r-1},\ldots,i_1}} \\ &&(\text{Continuing the same way for all of the Alice's operations}) \\ 
&\preceq& \frac{\Tr_A\br{M^{i_{r-1},\ldots,i_1}_BM^{i_{r-3},\ldots,i_1}_B\ldots M^{i_2,i_1}_B\theta_{AB}\br{M^{i_2,i_1}_B}^{\dagger}\ldots \br{M^{i_{r-3},\ldots,i_1}_B}^{\dagger}\br{M^{i_{r-1},\ldots,i_1}_B}^{\dagger}}}{p^x_{i_r,i_{r-1},\ldots,i_1}}\\ &=& \frac{M^{i_{r-1},\ldots,i_1}_BM^{i_{r-3},\ldots,i_1}_B\ldots M^{i_2,i_1}_B\theta_{B}\br{M^{i_2,i_1}_B}^{\dagger}\ldots \br{M^{i_{r-3},\ldots,i_1}_B}^{\dagger}\br{M^{i_{r-1},\ldots,i_1}_B}^{\dagger}}{p^x_{i_r,i_{r-1},\ldots,i_1}} \\ 
&\preceq& \frac{\sum_{i_{r-1}'} M^{i_{r-1}',i_{r-2},\ldots, i_1}_B\br{\ldots\br{\sum_{i_2'} M^{i_2',i_1}_B\theta_{B}\br{M^{i_2',i_1}_B}^{\dagger}}\ldots}\br{M^{i_{r-1}',i_{r-2},\ldots, i_1}_B}^{\dagger}}{p^x_{i_r,i_{r-1},\ldots,i_1}} \\ &&(\text{Adding positive operators to make the numerator a quantum state})
\end{eqnarray*}

Now define $$\sigma^{i_{r-2},\ldots, i_1}_B \defeq \sum_{i_{r-1}'} M^{i_{r-1}',i_{r-2},\ldots, i_1}_B\br{\ldots\br{\sum_{i_2'} M^{i_2',i_1}_B\theta_{B}\br{M^{i_2',i_1}_B}^{\dagger}}\ldots}\br{M^{i_{r-1}',i_{r-2},\ldots, i_1}_B}^{\dagger}$$ 
which is independent of $x$. Then we have 
$$p^x_{i_r,i_{r-1},\ldots,i_1} \leq 2^{-\dmax{\Tr_A\br{\phi^{x,i_r,i_{r-1},\ldots,i_1}_{AB}}}{\sigma_B^{i_{r-2},\ldots, i_1}}},$$ by the definition of max-relative entropy. Applying Bob's final unitary and tracing out the register $B'$, we obtain  $$\omega_C^{i_r,i_{r-1},\ldots,i_1}\defeq \Tr_{B'}\br{U_{i_r,i_{r-1},\ldots,i_1}\sigma^{i_{r-2},\ldots, i_1}_BU^{\dagger}_{i_r,i_{r-1},\ldots,i_1}}$$ which is independent of $x$. Now using the monotonicity of max-relative entropy under quantum operations (Fact \ref{fact:monotonequantumoperation}), we find that
$$p^x_{i_r,i_{r-1},\ldots,i_1} \leq 2^{-\dmax{\tau_C^{x,i_r,i_{r-1},\ldots,i_1}}{\omega_C^{i_r,i_{r-1},\ldots,i_1}}} \leq 2^{-\dmaxepss{\Psi^x_C}{\omega_C^{i_r,i_{r-1},\ldots,i_1}}{\eps^x_{i_r, i_{r-1}, \ldots, i_1}}}.$$ Here, the last inequality uses the definition of $\eps^x_{i_r, i_{r-1}, \ldots, i_1}$. 

From Definition~\ref{def:expectedcc}, it follows that the expected communication cost of the protocol $\mathcal{P}_2$ is lower bounded by $$\sum_x p(x)\sum_{i_1,i_2,\ldots,i_r}p^x_{i_1,i_2,\ldots,i_r}\log\left(i_1\cdot i_2 \cdots i_r\right).$$ This completes the proof of the lemma.
\end{proof}

\subsection{Achievability proof}

In this subsection, we construct an entanglement-assisted protocol for the classical-quantum state splitting task of any given ensemble with the expected communication cost close to $Q(\eta,r)$. In order to do so, we will require the following lemma. Its proof is inspired by the proof of convex split lemma as shown in \cite[Lemma 3.1, arXiv version 1]{ADJ14}. 

\begin{lemma}
\label{lem:convsplit}
Let $\delta \in (0,1)$ and $\cI$ be a finite set. Let $\{\Psi^i_C\}_{i\in \cI}$, $\{\omega^i_C\}_{i\in \cI}$ be two collections of quantum states belonging to $\cD(\H_C)$ and $\{p_i\}_{i\in \cI}$ be a collection of non-negative reals satisfying $$p_i \in (0,1-\delta), \quad \sum_i p_i=1, \quad p_i \leq 2^{-\dmax{\Psi^i_C}{\omega^i_C}}.$$ Introduce a set of registers $\set{C_{i}}_{i\in \cI}$ such that $C_{i}\equiv C$ for all $i$. There exist a collection of quantum states $\set{\tau^{(-i)}}_{i \in \cI}$ (that depend on $\delta$) such that $\tau^{(-i)} \in \cD\left(\bigotimes_{j\in \cI, j\neq i}\H_{C_j}\right)$ and a quantum state $\tau \in \cD\left(\bigotimes_{i\in \cI}\H_{C_i}\right)$ defined as
$$\tau \defeq \sum_{i\in \cI} p_i \Psi^i_{C_i}\otimes \tau^{(-i)}$$ such that 
$$\dmax{\tau}{\bigotimes_{i \in \cI}\omega^i_{C_i}}\leq 2\log\frac{1}{\delta}.$$
\end{lemma} 
\begin{proof}
From the fact that $p_i\leq 2^{-\dmax{\Psi^i_C}{\omega^i_C}}$ and the definition of max-relative entropy, there exists a quantum state $\Psi'^i_C$ such that $\omega^i_C = p_i \Psi^i_C + (1-p_i)\Psi'^i_C$. Without loss of generality, we may assume that $\cI=\set{1,\ldots,\abs{\cI}}$. Let $\Psi^i_0 \defeq \Psi^i_C$, $\Psi^i_1\defeq \Psi'^i_C$, $p_i^0\defeq p_i$ and $p_i^1\defeq 1-p_i$. For a string $\vec{s} \in \{0,1\}^{|\cI|}$, we define $$\Psi_{\vec{s}} \defeq \otimes_{i\in \cI}\Psi^i_{s_i}, \quad p^{\vec{s}} \defeq \Pi_{i\in \cI}p_i^{s_i},$$ where $s_i$ is the $i$-th bit of $\vec{s}$.

Let $\Typ$ be the set of all strings in $\{0,1\}^{|\cI|-1}$ for which the number of $0$'s is at most $\frac{1}{\delta}\log\frac{1}{\delta}$. The quantum state $\tau^{(-i)}$, as promised in the statement of the lemma, is defined as 
$$\tau^{(-i)} \defeq \frac{1}{N_i}\sum_{\vec{t}\in \Typ}p_1^{t_1}\cdots p_{i-1}^{t_{i-1}}p_{i+1}^{t_i}\cdots p_{|\cI|}^{t_{|\cI|-1}}\Psi^{1}_{t_1}\otimes\cdots\otimes \Psi^{i-1}_{t_{i-1}}\otimes \Psi^{i+1}_{t_i}\otimes\cdots\otimes \Psi^{|\cI|}_{t_{|\cI|-1}},$$ where $N_i = \sum_{\vec{t}\in \Typ}p_1^{t_1}\cdots p_{i-1}^{t_{i-1}}p_{i+1}^{t_i}\cdots p_{|\cI|}^{t_{|\cI|-1}}$ is the normalization factor. From the relations 
$$p_1^{0}+ \cdots + p_{i-1}^{0}+ p_{i+1}^{0}+\cdots+p_{|\cI|}^{0}\leq \sum_{i\in \cI} p_i =1,$$
$$p_1^{0}+ \cdots + p_{i-1}^{0}+ p_{i+1}^{0}+\cdots+p_{|\cI|}^{0} = 1- p_i^{0} \geq \delta$$
and the Hoeffding-Chernoff bound (Fact \ref{fact:chernoff}), it holds that 
$$N_i = \Pr[\Typ] \geq 1- e^{- \log\frac{1}{\delta}}= 1-\delta.$$
Observe that 
$$\bigotimes_{i \in \cI}\omega^i_{C_i} = \sum_{\vec{s}\in \{0,1\}^{|\cI|}}p^{\vec{s}}\Psi_{\vec{s}}.$$ Let $q(\vec{s})$ be such that 
$$\sum_i p^0_i \Psi^i_0\otimes \tau^{(-i)} = \sum_{\vec{s}\in \{0,1\}^{|\cI|}}q(\vec{s})\Psi_{\vec{s}}.$$
We claim that $q(\vec{s}) \leq \frac{1}{\delta^2}p^{\vec{s}}$ for all $\vec{s}$, which implies that 
$$\sum_i p^0_i \Psi^i_0\otimes \tau^{(-i)} \preceq \frac{1}{\delta^2}\bigotimes_{i \in \cI}\omega^i_{C_i},$$
proving the lemma. To prove the claim, observe that if $\vec{s}$ is such that the number of $0$'s in it is more than $\frac{1}{\delta}\log\frac{1}{\delta}+1$, then $q(\vec{s})=0$, by definition of the set $\Typ$. Otherwise, we have 
\begin{eqnarray*}
q(\vec{s}) &=& \sum_{i: s_i=0} \frac{p^0_i}{N_i}p_1^{s_1}\cdots p_{i-1}^{s_{i-1}}p_{i+1}^{s_{i+1}}\cdots p_{|\cI|}^{s_{|\cI|}}\\ &=&
\sum_{i: s_i=0} \frac{1}{N_i}p_1^{s_1}\cdots p_{i-1}^{s_{i-1}}p^0_{i}p_{i+1}^{s_{i+1}}\ldots p_{|\cI|}^{s_{|\cI|}}\\ &=& \br{\sum_{i: s_i=0} \frac{1}{N_i}} p^{\vec{s}} \leq \frac{1}{\delta(1-\delta)}\log\frac{1}{\delta}\cdot p^{\vec{s}} \leq \frac{1}{\delta^2}p^{\vec{s}}.
\end{eqnarray*}
This completes the proof.
\end{proof}

The above lemma allows us to construct a protocol for a large class of feasible solutions (which also includes the optimal solutions) to the optimization problem in Definition \ref{def:tightcharac}. The intuition behind the protocol is that Alice performs a measurement with outcomes over the set $T_r$, such that the probability of outcome $(i_r, i_{r-1}, \ldots, i_1)$ is equal to $p^x_{i_r, i_{r-1}, \ldots, i_1}$. Alice communicates this outcome to Bob, and the task of Bob is to simply pick up a register associated to the outcome $(i_r, i_{r-1}, \ldots, i_1)$. This leads to the desired expected communication cost and the success of the protocol is guaranteed by Lemma \ref{lem:convsplit}. We shall also use a rejection sampling step, along the lines similar to \cite{Jain:2005, Jain:2008}.

\begin{lemma}
\label{lem:expecachieve}
Given an ensemble $\set{p(x), \ketbra{\Psi^x}_{A'C}}$, real numbers $\eta, \delta \in (0,1)$ and integer $r\geq 1$, there exists an entanglement-assisted one-way protocol that achieves the classical-quantum state splitting task for this ensemble with average error $\eta+3\sqrt{\delta}$ and expected communication cost at most 
$$Q(\eta,r) + 2r\log Q(\eta,r) + 4r + 2\log\frac{4}{\delta}.$$
\end{lemma}
\begin{remark}
	Note that $Q(\eta,r)$ is a lower bound on $r$-round entanglement-assisted protocols for classical-quantum state splitting. Lemma~\ref{lem:expecachieve} implies that 
\begin{equation}
\label{one_way_r_way}
Q(\eta+ 3\sqrt{\delta}, 1) \leq Q(\eta,r) + 2r\log Q(\eta,r) + 4r + 2\log\frac{4}{\delta}.
\end{equation}
This suggests that interactions do not significantly save the expected communication cost for classical-quantum state splitting, which is intuitive to understand as Bob does not initially have any information about the input $x$. We conjecture that interactions do not save the expected communication cost at all.
\end{remark}
\begin{proof}[Proof of Lemma \ref{lem:expecachieve}]
 Abbreviate the tuple $(i_r, i_{r-1}, \ldots, i_1)$ with $\vec{i}$. Given a feasible solution $\set{\omega_C^{\vec{i}}}_{\vec{i} \in T_r}$, $\set{\eps^x_{\vec{i}}}_{\vec{i} \in T_r, x\in \X}$ and $\set{p^x_{\vec{i}}}_{\vec{i} \in T_r, x\in \X}$ for $Q(\eta,r)$, let $T'_r$ be  
 the union of the supports of $\set{p^x_{\vec{i}}}_x$. It can be verified that for optimal solutions, the associated $T'_r$ is a finite subset of $T_r$. 

Let $A$ be a register large enough such that $\ket{\omega^{\vec{i}}}_{AC}$ is a purification of $\omega_C^{\vec{i}}$. Introduce the registers $A_{\vec{i}}, C_{\vec{i}}$ respectively (with $A_{\vec{i}}\equiv A, C_{\vec{i}}\equiv C$) for every $\vec{i}\in T'_r$. The following quantum state $\ketbra{\omega}\in \cD(\bigotimes_{\vec{i}\in T'_r}\H_{A_{\vec{j}}C_{\vec{j}}})$ is a purification of $\bigotimes_{\vec{i}\in T'_r}\omega^{\vec{i}}_{C_{\vec{i}}}$:
$$\ket{\omega}\defeq \bigotimes_{\vec{i}\in T'_r}\ket{\omega^{\vec{i}}}_{A_{\vec{i}}C_{\vec{i}}}.$$

Let $\B$ be the set of all $x$ for which there exists an $\vec{i}$ such that $p^x_{\vec{i}}\geq 1-\delta$. Let $\G$ be the set of the rest of the $x$. Let $\Psi'^{x,\vec{i}}_C\in \cD(\H_C)$ be the quantum state achieving the optimum in the definition of $\dmaxepss{\Psi^x_C}{\omega_C^{\vec{i}}}{\eps^x_{\vec{i}}}$. By definition of $Q(\eta,r)$, we have $\P(\Psi'^{x,\vec{i}}_C, \Psi^x_C)\leq \eps^x_{\vec{i}}$. Let $\ketbra{\Psi'^{x,\vec{i}}}_{AC}\in \cD(\H_{AC})$ be a purification of $\Psi'^{x, \vec{i}}_C$ such that $\P(\ketbra{\Psi'^{x,\vec{i}}}_{AC}, \ketbra{\Psi^{x}}_{AC}) = \P(\Psi'^{x,\vec{i}}_C, \Psi^x_C)$. Applying Lemma \ref{lem:convsplit} for every $x\in \G$ (which uses the finiteness of $T'_r$), there exist quantum states $\tau^{x, \vec{i}}\in \cD(\bigotimes_{\vec{j}\in T'_r, \vec{j}\neq \vec{i}}\H_{C_{\vec{j}}})$ such that 
$$\dmax{\tau^x}{\bigotimes_{\vec{i}\in T'_r}\omega^{\vec{i}}_{C_{\vec{i}}}} \leq 2\log\frac{1}{\delta},$$ 
where
$$\tau^x \defeq \sum_{\vec{i}\in T'_r} p^x_{\vec{i}} \cdot\Psi'^{x, \vec{i}}_{C_{\vec{i}}}\otimes \tau^{x,\vec{i}}.$$This implies that there exists a quantum state $\theta^x$ such that 
\begin{equation}
\label{eq:omegasplit}
\bigotimes_{\vec{i}\in T'_r}\omega^{\vec{i}}_{C_{\vec{i}}} = \delta^2\tau^x + (1-\delta^2)\theta^x = \delta^2\cdot\sum_{\vec{i}\in T'_r} p^x_{\vec{i}} \cdot\Psi'^{x, \vec{i}}_{C_{\vec{i}}}\otimes \tau^{x,\vec{i}} + (1-\delta^2)\theta^x.
\end{equation}
Let $\ketbra{\tau^{x,\vec{i}}}\in \cD(\bigotimes_{\vec{j}\in T'_r, \vec{j}\neq \vec{i}}\H_{A_{\vec{j}}C_{\vec{j}}})$ be a purification of $\tau^{x,\vec{i}}$ and $\ketbra{\theta^x}\in \cD(\bigotimes_{\vec{i}\in T'_r}\H_{A_{\vec{j}}C_{\vec{j}}})$ be a purification of $\theta^x$.  By Eq.\eqref{eq:omegasplit}, the following quantum state is a purification of $\bigotimes_{\vec{i}\in T'_r}\omega^{\vec{i}}_{C_{\vec{i}}}$:
$$\ket{\bar{\omega}^x}\defeq \delta\sum_{\vec{i}\in T'_r} \sqrt{p^x_{\vec{i}}}\ket{\Psi'^{x, \vec{i}}}_{A_{\vec{i}}C_{\vec{i}}}\ket{\tau^{x,\vec{i}}}\ket{1}_P\ket{\vec{i}}_J + \sqrt{1-\delta^2}\ket{\theta^x}\ket{0}_P\ket{\vec{0}}_J,$$ where we have introduced registers $P,J$ of sufficiently large dimensions. From Uhlmann's theorem (Fact \ref{uhlmann}), there exists an isometry $$V^x: \bigotimes_{\vec{i}\in T_r}\H_{A_{\vec{i},k}}\rightarrow \H_{PJ}\otimes\bigotimes_{\vec{i}\in T_r}\H_{A_{\vec{i}}}$$ such that 
$$\ket{\bar{\omega}^x} = V^x\ket{\omega}.$$
The protocol is now as follows. 
\begin{enumerate}
\item {\bf Shared entanglement:} Alice and Bob share infinitely many copies of the quantum state $\ket{\omega}$ and assign a unique positive integer in $\{1,2,3, \ldots\}$, as an index, to each copy. \footnote{We remark that the shared entanglement can be made finite by introducing some additional errors in the protocol. }
\item {\bf Alice's encoding:}
\begin{enumerate}
\item If $x\in \G$, Alice communicates to Bob the bit $0$.
\begin{enumerate}  
\item {\bf Rejection sampling step:} Alice applies the isometry $V^x$ and then measures the register $P$ in each copy, according to ascending order of the index. Let $k$ be the first index where she obtains the outcome `$1$'. She communicates the index $k$ to Bob using the Huffman coding scheme for the probability distribution $\{(1-\delta^2)^{i-1}\delta^2\}_{i\in \{1,2,3,\ldots\}}$.
\item Alice measures the register $J$ in the $k$-th copy, to obtain an outcome $\vec{i}\in T'_r$. She then communicates $\vec{i} = (i_r, i_{r-1}, \ldots, i_1)$ by encoding each $i_\ell$ (for $\ell \in [r]$) in a prefix-free manner (as guaranteed by Fact \ref{prefixfree}).
\item Alice introduces the register $A'$ in the quantum state $\ket{0}_{A'}$ and swaps $A_{\vec{i}}$ (from the $k$-th copy) with $A'$.
\end{enumerate}
\item If $x\in \B$, Alice communicates to Bob the bit $1$.
\begin{enumerate}
\item Let $\vec{i}_x\in T'_r$ be such that $p^x_{\vec{i}_x}\geq 1-\delta$. Alice communicates $\vec{i}=\vec{i}_x$ to Bob with probability $p^x_{\vec{i}}$ and $\vec{i}=(1,1,\ldots 1)$ with probability $1-p^x_{\vec{i}}$, using the prefix-free encoding.
\item Let $V^{x,\vec{i}_x}: \H_{A}\rightarrow \H_A$ be the isometry such that $$\P\br{\Psi'^{x, \vec{i}_x}_C, \omega^{\vec{i}_x}_C} = \P\br{\ketbra{\Psi'^{x, \vec{i}_x}}_{AC}, V^{x,\vec{i}_x}\ketbra{\omega^{\vec{i}_x}}_C\br{V^{x,\vec{i}_x}}^{\dagger}},$$ as guaranteed by Uhlmann's theorem (Fact \ref{uhlmann}). Alice applies the isometry $V^{x,\vec{i}_x}$ on the register $A_{\vec{i}_x}$ (from the $1$-st copy of the shared entanglement $\ket{\omega}$). 
\item Alice introduces the register $A'$ in the quantum state $\ket{0}_{A'}$ and swaps $A_{\vec{i}}$ (from the $1$-st copy of the shared entanglement $\ket{\omega}$) with $A'$.
\end{enumerate}

\end{enumerate}
\item {\bf Bob's decoding:}
\begin{enumerate}
\item Receiving the message from Alice, Bob uses the first bit to check whether $x\in \G$ or $x\in \B$. If $x\in \G$, Bob first decodes index $k$, which is possible due to the prefix-free encoding of the Huffman coding scheme. If $x\in \B$, Bob sets $k=1$. He then decodes $\vec{i}$, which is possible due to the prefix-free encoding (Fact \ref{prefixfree}). 
\item Bob introduces the register $C$ in the quantum state $\ket{0}_{C}$ and swaps $C_{\vec{i}}$ (from the $k$-th copy) with $C$. 
\end{enumerate}
\item Final output $\Phi^x_{A'C}$ is obtained in the registers $A'C$.
\end{enumerate}

\vspace{0.1in}

\noindent {\bf Error analysis:} There always exists an index $k$ where Alice obtains the measurement outcome `$1$'. If $x\in \G$, Alice and Bob output the quantum state $\ket{\Psi'^{x,\vec{i}}}_{A'C}$ with probability $p^x_{\vec{i}}$. Thus, the average error for any $x\in \G$ is at most 
$$\sum_{\vec{i}\in T'_r}p^x_{\vec{i}}\P^2\br{\ketbra{\Psi'^{x,\vec{i}}}_{A'C}, \ketbra{\Psi^x}_{A'C}} = \sum_{\vec{i}\in T'_r}p^x_{\vec{i}}\br{\eps^x_{\vec{i}}}^2.$$
If $x\in \B$, then $\dmax{\Psi^{x, \vec{i}_x}_C}{\omega^{\vec{i}_x}_C}\leq \log\frac{1}{p^x_{\vec{i}_x}}\leq \log\frac{1}{1-\delta}$. This implies, using Fact \ref{pinsker}, that $$ \P\br{\ketbra{\Psi'^{x, \vec{i}_x}}_{AC}, V^{x,\vec{i}_x}\ketbra{\omega^{\vec{i}_x}}_C\br{V^{x,\vec{i}_x}}^{\dagger}} = \P\br{\Psi'^{x, \vec{i}_x}_C, \omega^{\vec{i}_x}_C} \leq \sqrt{\delta}.$$ Since Alice and Bob output $V^{x,\vec{i}_x}\ketbra{\omega^{\vec{i}_x}}_C\br{V^{x,\vec{i}_x}}^{\dagger}$ with probability at least $1-\delta$, the average error is at most 
\begin{eqnarray*}
&&(1-\delta)\P^2(\ketbra{\Psi^{x}}_{AC}, V^{x,\vec{i}_x}\ketbra{\omega^{\vec{i}_x}}_C\br{V^{x,\vec{i}_x}}^{\dagger}) + \delta\\
&& \leq (1-\delta)\P^2\br{\ketbra{\Psi^{x}}_{AC}, \ketbra{\Psi'^{x,\vec{i}_x}}_{A'C}} + 2\sqrt{\delta} + \delta\\ && \hspace{0.3cm}(\text{Using the triangle inequality for purified distance, Fact \ref{fact:trianglepurified}})\\ && \leq p^x_{\vec{i}_x}\P^2\br{\ketbra{\Psi^{x}}_{AC}, \ketbra{\Psi'^{x, \vec{i}_x}}_{A'C}} + 2\sqrt{\delta} + \delta\\ && \leq \sum_{\vec{i}\in T'_r}p^x_{\vec{i}}\P^2\br{\ketbra{\Psi^{x}}_{AC}, \ketbra{\Psi'^{x, \vec{i}}}_{A'C}} + 2\sqrt{\delta}+ \delta\\
&& = \sum_{\vec{i}\in T'_r}p^x_{\vec{i}}\br{\eps^x_{\vec{i}}}^2 + 2\sqrt{\delta}+ \delta.
\end{eqnarray*}
This implies that the overall average error of the protocol is at most
$$\sum_x p(x)\sum_{\vec{i}\in T'_r}p^x_{\vec{i}}\br{\eps^x_{\vec{i}}}^2 + 2\sqrt{\delta}+ \delta \leq \eta^2 + 3\sqrt{\delta}.$$

\vspace{0.1in}

\noindent {\bf Expected communication cost:}  The expected communication cost of the protocol can be upper bounded by the addition of the expected communication cost of the rejection sampling step, the expected communication cost for communicating $\vec{i}$ and the worst case communication cost for the first bit. 

Using Fact \ref{prefixfree}, the number of bits used in the communication of $\vec{i} = (i_r, i_{r-1}, \ldots, i_1)$ is at most 
\begin{eqnarray*}
&&\log\frac{1}{\delta} + \log \br{i_r\cdot i_{r-1}\cdots i_1} + 2\log(\log i_r \cdot \log i_{r-1}\ldots \log i_1) + 4r \\&& \leq \log (i_r\cdot i_{r-1}\cdots i_1) + 2r\log\log (i_r \cdot i_{r-1}\cdots i_1) + 4r,
\end{eqnarray*}
where we have used the identity $a_1\cdot a_2\cdot \ldots a_r \leq (a_1+a_2+\cdots+a_r)^r$, for positive reals $a_1, a_2, \ldots a_r$. If $x\in \B$, then the tuple $\vec{i}_x$ is communicated with probability $p^x_{\vec{i}_x}$ and the tuple $(1,1,\ldots 1)$ is communicated with the remaining probability. Since the tuple $(1,1,\ldots 1)$ has the smallest encoding length in the prefix-free encoding in Fact \ref{prefixfree}, the expected communication cost is smaller than the expected communication cost for transmitting $\vec{i}$ with probability $p^x_{\vec{i}}$. Thus, the  expected communication cost for communicating $\vec{i}$ is upper bounded by 
\begin{eqnarray*}
&&\sum_x p(x)\sum_{i_r, i_{r-1}, \ldots, i_1} p^x_{i_r,i_{r-1}, \ldots, i_1} \left(\log (i_r\cdot i_{r-1}\cdots i_1) + 2r\log\log (i_r \cdot i_{r-1}\cdots i_1) + 4r\right) \\ && \leq  \left(\sum_x p(x)\sum_{i_r, i_{r-1}, \ldots, i_1} p^x_{i_r,i_{r-1}, \ldots, i_1} \log (i_r\cdot i_{r-1}\cdots i_1)\right) + \\ &&  2r\log\left(\sum_x p(x)\sum_{i_r, i_{r-1}, \ldots, i_1} p^x_{i_r,i_{r-1}, \ldots, i_1}\log (i_r \cdot i_{r-1}\cdots i_1)\right) + 4r.
\end{eqnarray*}
The expected communication cost for the Huffman coding of the probability distribution $\{(1-\delta^2)^{i-1}\delta^2\}_{i\in \{1,2,3\ldots\}}$ is at most the Shannon entropy of the probability distribution up to additional one bit, which is 
\begin{eqnarray*}
&&\sum_{i=1}^{\infty}(1-\delta^2)^{i-1}\delta^2\log\frac{1}{(1-\delta^2)^{i-1}\delta^2}+1\\ 
&&=\log\frac{1-\delta^2}{\delta^2} + \sum_{i=1}^{\infty}i(1-\delta^2)^{i-1}\cdot\delta^2\log\frac{1}{1-\delta^2} + 1\\
&&=\log\frac{1-\delta^2}{\delta^2} + \frac{1}{\delta^4}\cdot\delta^2\log\frac{1}{1-\delta^2} + 1\\ 
&&\leq 2\log\frac{1}{\delta}+\frac{1}{\delta^2}\log\frac{1}{1-\delta^2}+1 \leq 2\log\frac{4}{\delta}, 
\end{eqnarray*}
where we have used the inequality $\log\frac{1}{1-\delta^2} \leq 2\delta^2$. The proof now follows from the definition of $Q(\eta,r)$.
\end{proof}

\subsection{A simple lower bound on $Q(\eta, r)$}

For our application, we simplify the bound given in Definition \ref{def:tightcharac}. This is achieved in the following lemma. 

\begin{lemma}
\label{lem:simplerlow}
Fix a finite set $\X$, $\eta, \gamma\in (0,1)$ and an integer $r\geq 1$. For any ensemble $\{(p(x), \ketbra{\Psi^x}_{A'C})\}_{x\in \X}$, where $\Psi^x_{A'C} \in \cD(\H_{A'C})$, the associated $Q(\eta, r)$ satisfies that 
$$Q(\eta, r) \geq (1-\gamma)^2\cdot\frac{-\log \left(\max_{\omega_C \in \cD(\H_C)} \sum_xp(x)2^{-\dmaxepss{\Psi^x_C}{\omega_C}{\frac{\eta}{\gamma}}}\right) + 2\log(1-\gamma)}{2\log r + 8},$$ if $r>1$ and  
$$Q(\eta, 1) \geq (1-\gamma)^2\cdot\br{-\log \left(\max_{\omega_C \in \cD(\H_C)} \sum_xp(x)2^{-\dmaxepss{\Psi^x_C}{\omega_C}{\frac{\eta}{\gamma}}}\right) + 2\log(1-\gamma) - 1},$$ for $r=1$.
\end{lemma}
\begin{proof}
Define $$Q^*\defeq -\log \left(\max_{\omega_C \in \cD(\H_C)} \sum_xp(x)2^{-\dmaxepss{\Psi^x_C}{\omega_C}{\frac{\eta}{\gamma}}}\right).$$ 
Our proof proceeds in the following steps.
\begin{enumerate}
\item \textbf{Pruning out $x$ with large error:}

Let $\G$ be the set of all indices $x$ such that $$\sum_{i_r, i_{r-1}, \ldots, i_1} p^x_{i_r,i_{r-1}, \ldots, i_1} \left(\eps^x_{i_r,i_{r-1}, \ldots, i_1}\right)^2\leq \frac{\eta^2}{\gamma}.$$ Let $\B$ be the rest of the indices. Then by Markov's inequality, we have $\sum_{x\in\G}p(x)\geq 1-\gamma$. Let $q(x)$ be defined as $q(x) \defeq \frac{p(x)}{\sum_{x'\in\G}p(x')}$ if $x\in \G$ and $0$, otherwise.

\item \textbf{Removing $(i_r,i_{r-1},\ldots,i_1)$ with large error:}

Now for each $x\in\G$, we define $\B_x$ to be the set of tuples $(i_r,i_{r-1},\ldots,i_1)$ for which $\eps^x_{i_r,i_{r-1}, \ldots, i_1} \geq \frac{\eta}{\gamma}$. Let $\G_x$ be rest of the indices. Then we have $\sum_{(i_r,i_{r-1},\ldots,i_1)\in \B_x}p^x_{i_r,i_{r-1},\ldots,i_1} \leq \gamma$. And hence for all $(i_r,i_{r-1},\ldots,i_1)\notin \B_x$, we obtain $$p^x_{i_r,i_{r-1},\ldots,i_1} \leq 2^{-\dmaxepss{\Psi^x_C}{\omega_C^{i_r,i_{r-1},\ldots,i_1}}{\frac{\eta}{\gamma}}}.$$

\item \textbf{Upper bound on average probability of a message:}

We define a new probability distribution $q^x_{i_r,i_{r-1},\ldots,i_1}$ which is $0$ whenever $(i_r,i_{r-1},\ldots,i_1)\in \B_x$ and equal to $p^x_{i_r,i_{r-1},\ldots,i_1} / p^x(\G_x)$ otherwise, where we define $$p^x(\G_x)\defeq\sum_{(i_r,i_{r-1},\ldots, i_1)\in \G_x}p^x_{i_r,i_{r-1},\ldots, i_1}.$$ It follows that 
$$q^x_{i_r,i_{r-1},\ldots,i_1} \leq \frac{1}{1-\gamma}2^{-\dmaxepss{\Psi^x_C}{\omega_C^{i_r,i_{r-1},\ldots,i_1}}{\frac{\eta}{\gamma}}}.$$

Define $s_{i_ri_{r-1},\ldots,i_1}\defeq \sum_x q(x)q^x_{i_r,i_{r-1},\ldots,i_1}$. We have $\sum_{i_r,i_{r-1},\ldots,i_1}s_{i_r,i_{r-1},\ldots,i_1}=1$. Furthermore,
 
\begin{eqnarray}
\label{simultiupperbound}
s_{i_r,i_{r-1},\ldots,i_1}&\leq& \frac{1}{1-\gamma}\sum_xq(x)2^{-\dmaxepss{\Psi^x_C}{\omega_C^{i_r,i_{r-1},\ldots,i_1}}{\frac{\eta}{\gamma}}}\nonumber\\ &\leq& \frac{1}{(1-\gamma)^2}\sum_xp(x)2^{-\dmaxepss{\Psi^x_C}{\omega_C^{i_r,i_{r-1},\ldots,i_1}}{\frac{\eta}{\gamma}}}\nonumber\\&\leq & \frac{2^{-Q^*}}{(1-\gamma)^2}.
\end{eqnarray}
\item \textbf{Lower bound on $Q(\eta, r)$:}

Consider
\begin{eqnarray*}
Q(\eta,r)&\geq& \sum_{x\in\G}p(x) \sum_{i_r,i_{r-1},\ldots,i_1 \in \G_x}p^x_{i_r,i_{r-1},\ldots,i_1} \log(i_ri_{r-1},\ldots,i_1)\\&\geq& (1-\gamma)^2 \sum_xq(x) \sum_{i_r,i_{r-1},\ldots,i_1}q^x_{i_r,i_{r-1},\ldots,i_1} \log(i_ri_{r-1},\ldots,i_1) \\ &=& (1-\gamma)^2 \sum_{i_r,i_{r-1},\ldots,i_1}s_{i_r,i_{r-1},\ldots,i_1} \log(i_ri_{r-1},\ldots,i_1)
\end{eqnarray*}

where the last equality follows from the definition of $s_{i_r,i_{r-1},\ldots,i_1}$.  Invoking Equation~\eqref{simultiupperbound}  and Claim~\ref{combinatorics} with $b= Q^*+2\log(1-\gamma)$, the lemma follows.

\end{enumerate}
\end{proof}

\begin{claim}
\label{combinatorics}
Suppose $s_{i_r,i_{r-1},\ldots, i_1}\leq 2^{-b}$ for all $(i_1,i_2,\ldots,i_r)$. Then we have $$\sum_{i_1,i_2\ldots i_1} s_{i_r,i_{r-1},\ldots, i_1}\log(i_ri_{r-1},\ldots,i_1) \geq \frac{b}{2(\log r + 4)},$$ for $r>1$ and 
$$\sum_{i_1} s_{i_1}\log(i_1) \geq b-1,$$ for $r=1$.
\end{claim}

\begin{proof}
We first consider the case $r=1$. $\sum_{i_1} s_{i_1}\log(i_1)$ is minimized when $s_{i_1} = 2^{-b}$ for all $i_1\in [2^b]$. For this, we have $$\sum_{i_1} s_{i_1}\log(i_1) = 2^{-b} \sum_{i=1}^{2^b}\log i \geq \log 2^b -1 \geq b-1.$$

Now we consider the case $r>1$. For an integer $k$ let $N(k)$ be the number of {ordered} tuples $(i_1,i_2,\ldots, i_r)$ such that $k=i_1\cdot i_2\cdots i_r$. Let $M(k)=\sum_{k'=1}^k N(k')$. The quantity $\sum_{i_1,i_2\ldots i_1} s_{i_r,i_{r-1},\ldots, i_1}\log(i_r\cdot i_{r-1}\cdots i_1)$ is minimized when all $s_{i_r,i_{r-1},\ldots, i_1}$ with the {smallest possible values} of the product $i_1\cdot i_2\cdot\ldots i_r$ have taken the value $2^{-b}$. Let $k^{*}$ be the largest integer such that $M(k^{*}) < 2^b$. Let $N'(k^{*}+1)\defeq 2^b-M(k^{*})$.  Then $$\sum_{i_1,i_2\ldots i_1} s_{i_ri_{r-1},\ldots,i_1}\log(i_r\cdot i_{r-1}\cdots i_1) = 2^{-b}(\sum_{k=1}^{k^{*}} N(k)\log(k)+N'(k^{*}+1)\log(k^{*}+1)).$$ 
 
Our lower bound proceeds by evaluating $N(k)$. Let $k= 2^{a_1}3^{a_2}\ldots p_t^{a_t}$ be the prime decomposition of $k$. 
Each of the $r$ integers that multiply to give $k$ can be written as $n_f=2^{a^f_1}3^{a^f_2}\ldots p_t^{a^f_t}$ where $f\in [r]$. Since $n_1\cdot n_2\ldots n_r = k$, we have that $$a^1_1+a^2_1+\cdots+a^r_1 = a_1,\quad a^1_2+a^2_2+\cdots+a^r_2 = a_2, \quad a^1_t+a^2_t+\cdots+a^r_t = a_t.$$

We need to compute the number of ways of selecting the ordered tuple \\$(a^1_1,a^2_1,\ldots, a^r_1, a^1_2,\ldots, a^r_2, \ldots a^1_t,a^2_t,\ldots, a^r_t)$ that satisfy the above constraints. We note that the order matters. For the first constraint, the number of ways is equal to ${a_1+r-1 \choose r-1}$. Similar argument holds for rest of the constraints, and each being independent, we obtain that the number of ways is : $${a_1+r-1 \choose r-1}{a_2+r-1 \choose r-1}\cdots {a_t+r-1 \choose r-1}.$$ Since 
$${a_1+r-1 \choose r-1} = (a_1+1)\br{\frac{a_1}{2}+1}\cdots\br{\frac{a_1}{r-1}+1}<2^{a_1+\frac{a_1}{2}+\cdots+\frac{a_1}{r-1}} < 2^{a_1(\log r +\gamma)},$$ where $\gamma$ is Euler-Mascheroni constant, we have that 
$$N(k)<2^{(a_1+a_2+\cdots+a_t)(\log r + \gamma)} < k^{(\log r + \gamma)}.$$ Thus, $M(k)<k^{(\log r + \gamma+1)}$. 

Now, $k^{*}$ is the integer such that $M(k^{*})<2^b <M(k^{*}+1)$. This means, $k^{*}> 2^{\frac{b}{\log r+\gamma+1}}-1$. 
Now we are in a position to obtain the final lower bound. Consider,
\begin{eqnarray*}
&&\sum_{i_1,i_2\ldots i_1} s_{i_ri_{r-1},\ldots,i_1}\log(i_r\cdot i_{r-1}\cdots i_1) \\&=&2^{-b}\br{\sum_{k=1}^{k^{*}} N(k)\log(k)+N'(k^{*}+1)\log\br{k^{*}+1}}\\&>& 2^{-b}\br{\sum_{k=\sqrt{k^{*}}}^{k^{*}} N(k)\log(k)+N'(k^{*}+1)\log\br{k^{*}+1}}\\ &>& \frac{\log(k^{*})}{2}\cdot 2^{-b}\br{\sum_{k=\sqrt{k^{*}}}^{k^{*}} N(k)+N'\br{k^{*}+1}}\\ &=& \frac{\log(k^{*})}{2}\cdot \br{1-2^{-b}\sum_{k=1}^{\sqrt{k^{*}}-1} N(k)} \\&>& \frac{\log(k^{*})}{2}\cdot \br{1-2^{-b}\br{\sqrt{k^{*}}}^{\log r+\gamma+1}} > \frac{b(1-2^{-b/2})}{2(\log r+\gamma+1)}.
\end{eqnarray*}

This proves the claim.
\end{proof}

\section{Separating expected communication and information for classical-quantum state transfer}
\label{sec:example}

We formally define the task of classical-quantum state transfer. 

\begin{task}[\textbf{Classical-quantum state transfer}]
\label{def:quantumhuffman}
Fix a Hilbert space $\H$ and a $\eta\in (0,1)$. Alice receives an input $x\sim p\br{\cdot}$ associated with a quantum pure state $\ketbra{\Psi^x}\in \cD(\H)$, where $p\br{\cdot}$ is a distribution over a finite set $\X$ and $x\in\X$. The goal is that Bob outputs a quantum state $\Phi^x\in \cD(\H)$ satisfying $\sum_x p(x)\F^2(\Psi^x,\Phi^x)\geq 1-\eta^2$.
\end{task}

The parameter $\eta$ will be referred to as the average error of a protocol achieving the task. This section is devoted to the proof of the following theorem. For the ease of presentation and noting that only one register is involved, we will drop the register labels on the quantum states. 

\begin{theorem}
	\label{thm:maingeneral}
	Fix positive integer $d>4$ and $\delta\in(0,1/4)$. There exists a collection of $N\defeq
	8d^7$ states $\{\ket{\Psi^x}\}_{x=1}^N$ belonging to a $d$-dimensional Hilbert space $\H$, and a probability distribution $\{p(x)\}_{x=1}^N$, such that following holds for the ensemble $\{(p(x),\Psi^x)\}_{x=1}^N$.
	
	\begin{itemize}
		\item  The von Neumann entropy of the average state satisfies $\mathrm{S}(\sum_x p(x) \Psi^x) \leq \delta\log(d)+H(\delta)+2$
		\item For any one-way protocol achieving the classical-quantum state transfer of the above ensemble with average error $\eta \in (0,(\frac{\delta}{8})^2)$, the expected communication cost is lower bounded by $(1-\sqrt{\eta})^2 \log(\frac{d\delta}{128})$.
		
		\item For any $r$-round interactive protocol achieving the classical-quantum state transfer of the above ensemble with average error $\eta \in (0,(\frac{\delta}{8})^2)$, the expected communication cost is lower bounded by $$\frac{1}{20}\cdot\frac{\log(\frac{d\delta}{128})}{(\log r)}.$$
		\item For any interactive protocol (with arbitrary many rounds) achieving the classical-quantum state transfer of the above ensemble with average error $\eta \in (0,(\frac{\delta}{10})^4)$, the expected communication cost is lower bounded by 
		$$\frac{1}{30}\cdot\frac{\log(\frac{d\delta}{128})}{(\log\log(d)-2\log\eta)}.$$
	\end{itemize}
\end{theorem}

The proof of this theorem, given in Subsection \ref{subsec:maingen}, shall follow from the construction given below. 

\subsection{Construction}

Let $\H$ be a $d$-dimensional Hilbert space; $\ket{0}$ be an arbitrary pure state in $\H$, and $V$ be the $\br{d-1}$-subspace orthogonal to $\ket{0}$.  Let $P$ be the projector onto $V$ and $S$ be the unit ball in $V$. Let $\mu$ be a Haar measure on $S$. We continue to use the ket notation to represent the vectors in $S$. We sample $m$ pure states $\{\ket{x_1},\ket{x_2}\ldots \ket{x_m}\}$,  independently from $\mu$, where $m$ is to be chosen later. Fix a $\delta\in (0,\frac{1}{4})$. For each $i$, define the following random hermitian matrices 
\begin{equation}\label{eqn:Zistatetransfer}
	Z^1_i \defeq \ketbra{x_i}, \quad Z^2_i \defeq \sqrt{\delta-\delta^2}(\ket{x_i}\bra{0} + \ket{0}\bra{x_i}) + \delta\ketbra{x_i}, \quad Z^3_i \defeq \ketbra{x_i}\otimes \ketbra{x_i}.
\end{equation}
 We have that $\|Z^1_i\|_{\infty}\leq 1, \|Z^2_i\|_{\infty}\leq 2\sqrt{\delta}< 1, \|Z^3_i\|_{\infty} \leq 1$. Furthermore, it holds that 

\begin{eqnarray}
&& \mathbb{E}(Z^1_i) = \frac{P}{d} ,\nonumber \\ 
&& \mathbb{E}(Z^2_i) = \delta\frac{P}{d},\nonumber \\ && \mathbb{E}(Z^3_i) = \frac{P\otimes P  + F}{d(d+1)}\label{eqn:xx}  
\end{eqnarray}
where $F$ is the swap operator on $V\times V$. Namely,
\[F\ket{x,y}\defeq\begin{cases}
\ket{y,x}&~\mbox{if $\ket{x}\in V,\ket{y}\in V$}\\
$0$&~\mbox{otherwise.}
\end{cases}\]
From the Matrix Hoeffding bound (Fact \ref{matrixhoeff}), the following three inequalities hold for any $\eps>0$: 
\begin{eqnarray*}
&&\prob{\onenorm{\frac{\sum_i Z^1_i}{m} - \frac{P}{d}}\geq\epsilon}\leq d\cdot e^{-\frac{m\eps^2}{8\cdot d^2}}, \nonumber\\
&&
\prob{\onenorm{\frac{\sum_i Z^2_i}{m} - \delta\frac{P}{d}} \geq \eps}\leq d\cdot e^{-\frac{m\eps^2}{8\cdot d^2}}, \nonumber\\ &&
\prob{\onenorm{\frac{\sum_i Z^3_i}{m} - \frac{P\otimes P  + F}{d(d+1)}}}\leq d\cdot e^{-\frac{m\eps^2}{8\cdot d^4}}.
\end{eqnarray*}
Setting $m = \frac{8d^5}{\eps^2}$, we find that for $d> 4$, all the upper bounds are less than $1/3$. Applying the union bound, there exists a set of pure states $\N_{\eps}\defeq\set{\ket{x_1},\ket{x_2}\ldots \ket{x_m}}\subseteq S$ satisfying  

\begin{eqnarray}
\label{eq:niceset}
&& \onenorm{\frac{\sum_i Z^1_i}{m} - \frac{P}{d}} \leq \eps, \nonumber \\ &&
\onenorm{\frac{\sum_i Z^2_i}{m} - \delta\frac{P}{d}} \leq \eps, \nonumber\\ &&
\onenorm{\frac{\sum_i Z^3_i}{m} - \frac{P\otimes P  + F}{d(d+1)}} \leq \eps,
\end{eqnarray}
where $\set{Z^1_i, Z^2_i, Z^3_i}_{1\leq i\leq m}$ are defined in Eq.~\eqref{eqn:Zistatetransfer}. The ensemble is now constructed as follows, with $\X$ being the set $[m]$.  
\begin{equation}\label{eqn:lambda}
	p(i)\defeq \frac{1}{m}, \quad \ket{\Psi_i}\defeq\sqrt{1-\delta}\ket{0}+\sqrt{\delta}\ket{x_i} \quad \text{for }i\in [m],
\end{equation}
 where $\ket{x_i}$ is the state given in $\N_{\eps}$. We will use the notation $\EE_i$ to represent expectation according to the distribution $p(i)$.

\subsection{Upper bound on the von-Neumann entropy}

We have the following lemma.
\begin{lemma}
\label{mutinfbound}
The von-Neumann entropy of the average state $\expec{i}{\Psi_i} = \sum_i p(i) \Psi_i$ satisfies $\mathrm
{S}\br{\expec{i}{\Psi_i}}\leq \br{\delta+\eps}\log d+H\br{\delta}+1$.
\end{lemma}
\begin{proof}
Consider,
 $$ \expec{i}{\ketbra{\Psi_i}}=\br{1-\delta}\ketbra{0}+\expec{i}{\sqrt{\delta(1-\delta)}(\ket{0}\bra{x_i}+\ket{x_i}\bra{0})+\delta\ketbra{x_i}}.$$
 From Eq.~\eqref{eq:niceset}, it follows that $$\onenorm{\expec{i}{\Psi_i}- \br{1-\delta}\ketbra{0}-\delta\frac{P}{d}} \leq \eps.$$ Now we use Alicki-Fannes inequality (Fact \ref{fact:fannes}) to conclude that 
 $$S\br{\expec{i}{\Psi_i}}\leq S\br{\br{1-\delta}\ketbra{0}+\delta\frac{P}{d}}+\eps\log d+1=\br{\delta+\eps}\log d+ H\br{\delta}+1.$$
\end{proof}


\subsection{Bound on the average smooth max-relative entropy}
\label{subsec:finalbound}

Given our construction and Lemma \ref{lem:simplerlow}, it is sufficient to provide an upper bound on $\EE_i 2^{-\dmaxepss{\Psi_i}{\omega}{\etaa}}$, for an arbitrary quantum state $\omega$ and small enough $\etaa$. We will prove the following lemma.

\begin{lemma}
\label{lem:uniformdmax}
Fix $\delta$ as defined above and let $\etaa\in (0,\frac{\delta}{8})$. Let $\eps = \frac{1}{d}$. Then it holds that 
$$\max_{\omega\in \cD(\H)}\expec{i}{2^{-\dmaxepss{\Psi_i}{\omega}{\etaa}}} \leq 2^{-\log(d\delta) + 6}.$$
\end{lemma}

\begin{proof}[Proof of Lemma \ref{lem:uniformdmax}]
Fix a quantum state $\omega \in \cD(\H)$. Let $k<d$ be an integer and $Q^{-}$ ($Q^{+}$) be the projector onto the subspace where the eigenvalues of $\omega$ are less than (greater than or equal to) $\frac{1}{k}$. Since $\mathrm{rk}~Q^{+}\leq k$, it holds that $\mathrm{rk}~Q^{-}\geq d+1-k$.  Let $W$ be the subspace corresponding to $Q^-$. We have $\dim W=d-\mathrm{rk}~Q^+\geq d+1-k$. Let $W'$ be an arbitrary $\br{d-k}$-dimensional subspace of $W$ orthogonal to $\ket{0}$, the existence of which is easy to verify. Let $Q$ be a projector onto $W'$. 

We apply Lemma~\ref{lem:smalloverlapbound} by setting $k=\frac{d}{4}$, $\eps = \frac{1}{d}$ $\alpha=\frac{\delta}{4}$ and obtain $$\Pr_i\Br{\bra{\Psi_i}Q\ket{\Psi_i} < \delta/2+\delta/d} \leq \frac{96}{d}.$$ Let $\B$ be the set of all $i$ such that $\bra{\Psi_i}Q\ket{\Psi_i}< \frac{\delta}{2}$. Let $\G$ be the set of rest of the $i$. Above inequality implies that $\Pr_i\Br{\B} \leq \frac{96}{d}$. Now, consider
\begin{eqnarray*}
&&\expec{i}{2^{-\dmaxepss{\Psi_i}{\omega}{\etaa}}} \\
&&= \sum_{i\in \B}\frac{1}{m}\cdot 2^{-\dmaxepss{\Psi_i}{\omega}{\etaa}} + \sum_{i\in \G}\frac{1}{m}\cdot 2^{-\dmaxepss{\Psi_i}{\omega}{\etaa}}\\
&&\leq \Pr_i\Br{\B} + \sum_{i\in \G}\frac{1}{m}\cdot 2^{-\dmaxepss{\Psi_i}{\omega}{\etaa}}\\
&& \leq \frac{96}{d}+ \max_{i\in \G}2^{-\dmaxepss{\Psi_i}{\omega}{\etaa}}.
\end{eqnarray*}	
For an $i\in \G$, we have $\bra{\Psi_i} Q^{-}\ket{\Psi_i} \geq \bra{\Psi_i} Q\ket{\Psi_i}\geq \frac{\delta}{2}> 2\etaa$. Thus, we use Lemma \ref{maxentropybound} to conclude that for all $i\in \G$, 
	\begin{eqnarray*}
	&& 2^{-\dmaxepss{\Psi_i}{\omega}{\etaa}} \\ &&\leq \frac{4}{d(1-\etaa)\br{\sqrt{(1-2\etaa)(\frac{\delta}{2})}-\sqrt{2(1-\frac{\delta}{2})\etaa}}^2 }\\ && \leq \frac{4}{d(1-\frac{\delta}{8})\br{\sqrt{(1-\frac{\delta}{4})(\frac{\delta}{2})}-\sqrt{(1-\frac{\delta}{2})\frac{\delta}{4}}}^2 }\\ && \leq \frac{40}{d\delta}.
	\end{eqnarray*}
This leads to the upper bound
	\begin{eqnarray*}
	\expec{i}{2^{-\dmaxepss{\Psi_i}{\omega}{\etaa}}} \leq  \frac{96}{d} +  \frac{40}{d\delta} \leq 2^{-\log(d\delta) + 6},
	\end{eqnarray*}
which proves the lemma.  
\end{proof}

For the discussion below, we fix the quantum state $\omega \in \cD(\H)$ as appearing in the above proof. The following lemma provides an explicit lower bound on the smooth max-relative entropy between $\ketbra{\Psi_i}$ and $\omega$.
\begin{lemma}
\label{maxentropybound}
For any $i$ and $\etaa \in (0,1)$ satisfying  $\bra{\Psi_i} Q^{-}\ket{\Psi_i}>2\etaa$, it holds that $$2^{-\dmaxepss{\Psi_i}{\omega}{\etaa}}\leq  \frac{1}{k(1-\etaa)(\sqrt{(1-\etaa)\bra{\Psi_i}Q^{-}\ket{\Psi_i}}-\sqrt{\bra{\Psi_i}Q^{+}\ket{\Psi_i}\etaa})^2}.$$
\end{lemma}

\begin{proof}
 Define the quantity $$S^{\etaa}(\Psi_i||Q^{-})\defeq \text{inf}_{\ketbra{\lambda} \in \cD(\H): |\braket{\lambda}{\Psi_i}|^2>1-\etaa}\bra{\lambda}Q^{-}\ket{\lambda}.$$
The lemma follows from Claim~\ref{upperboundsmoothdmax} and Claim~\ref{exactoverlap}.

\end{proof}

\begin{claim}
\label{upperboundsmoothdmax}
For any $i$, it holds that 
$$2^{-\dmaxepss{\Psi_i}{\omega}{\etaa}} \leq \frac{1}{k(1-\etaa)S^{2\etaa}(\Psi_i||Q^{-})}.$$  
\end{claim}

\begin{proof} 
	For a fixed $i$, let $\rho_i \in \cD(\H)$ be the state that achieves the infimum in the definition of $\dmaxepss{\Psi_i}{\omega}{\etaa}$. It satisfies $\bra{\Psi_i}\rho_i\ket{\Psi_i}\geq 1-\etaa$. This means that the largest eigenvalue of $\rho_i$ is at least $1-\etaa$. Thus, consider the eigen-decomposition $\rho_i=\lambda_1\ketbra{\lambda_1}+\sum_{j>1}\lambda_j\ketbra{\lambda_j}$. We have $\lambda_1\geq 1-\etaa$ or equivalently $\sum_{j>1}\lambda_j\leq \etaa$. Thus, $$1-\etaa\leq\bra{\Psi_i}\rho_i\ket{\Psi_i}=\lambda_1|\braket{\Psi_i}{\lambda_1}|^2+\sum_{j>1}\lambda_j|\braket{\Psi_i}{\lambda_j}|^2 \leq |\braket{\Psi_i}{\lambda_1}|^2+\sum_{j>1}\lambda_j \leq  |\braket{\Psi_i}{\lambda_1}|^2+\etaa.$$
Hence, $|\braket{\Psi_i}{\lambda_1}|^2\geq 1-2\etaa$. Moreover, $$2^{\dmax{\rho_i}{\omega}}= \|\omega^{-\frac{1}{2}}\rho\omega^{-\frac{1}{2}}\|_{\infty}\geq (1-\etaa)\|\omega^{-\frac{1}{2}}\ketbra{\lambda_1}\omega^{-\frac{1}{2}}\|_{\infty}=(1-\etaa)\bra{\lambda_1}\omega^{-1}\ket{\lambda_1},$$ 
where $\omega^{-1}$ is the pseudo-inverse of $\omega$. From the definition of the projector $Q^{-}$, the following inequality holds: $$\bra{\lambda_1}\omega^{-1}\ket{\lambda_1} \geq  k\bra{\lambda_1}Q^{-}\ket{\lambda_1}.$$ Thus we get 
	$$2^{\dmax{\rho_i}{\omega}} \geq k(1-\etaa)\bra{\lambda_1}Q^{-}\ket{\lambda_1}.$$ Inverting and using $|\braket{\Psi_i}{\lambda_1}|^2\geq 1-2\etaa$, we have 
	$$2^{-\dmax{\rho_i}{\omega}} \leq \frac{1}{k(1-\etaa)\bra{\lambda_1}Q^{-}\ket{\lambda_1}} \leq \frac{1}{k(1-\etaa)S^{2\etaa}(\Psi_i||Q^{-})}.$$  
This proves the claim.
	
\end{proof}

\begin{claim}
\label{exactoverlap}
 If $\bra{\Psi_i}Q^{-}\ket{\Psi_i}>\etaa$, then we have $$S^{\etaa}(\Psi_i||Q^{-}) = (\sqrt{(1-\etaa)\bra{\Psi_i}Q^{-}\ket{\Psi_i}}-\sqrt{\bra{\Psi_i}Q^{+}\ket{\Psi_i}\etaa})^2.$$ Else $S^{\etaa}(\Psi_i||Q^{-})=0$.  
\end{claim}

The proof of Claim~\ref{exactoverlap} involves direct but tedious calculations, which is deferred to Appendix~\ref{proof_claims}. 

\begin{lemma}
\label{lem:smalloverlapbound}
Let $Q$ be the projector onto a $\br{d-k}$-dimensional subspace of $\H$ such that $Q\ket{0}=0$. For every $\alpha\in (0,1)$, it holds that 
$$\Pr_i\Br{\bra{\Psi_i}Q\ket{\Psi_i} < \delta\frac{d-k}{d}+\delta\eps - \alpha} \leq \frac{\delta^2}{\alpha^2}\br{3\eps + \frac{3}{d}}.$$
\end{lemma}

The proof of this lemma uses the following two claims.

\begin{claim}
\label{averageofprojector}
It holds that
$$ \frac{\delta (d-k)}{d} + \delta\eps \geq  \expec{i}{\bra{\Psi_i}Q\ket{\Psi_i}}\geq \frac{\delta (d-k)}{d} - \delta\eps.$$
\end{claim}

\begin{proof}
Since $Q$ is orthogonal to $\ket{0}$, we have that $$\expec{i}{\bra{\Psi_i}Q\ket{\Psi_i}}  = \delta \expec{i}{\bra{x_i}Q\ket{x_i}}.$$ 
Using Eq.\eqref{eq:niceset}, we find that 
$$\frac{\mathrm{rk}~Q}{d}+\eps\geq\Tr\br{Q\frac{P}{d}} + \eps\geq \mathbb{E}_{i} \bra{x_i}Q\ket{x_i} \geq \Tr\br{Q\frac{P}{d}} - \eps = \frac{\mathrm{rk}~Q}{d}-\eps.$$
The claim follows from the fact that $\mathrm{rk}~Q=d-k$.

\end{proof}

\begin{claim}
\label{squareofprojector}
It holds that 
$$ \delta^2\frac{(d-k)(d-k+1)}{d(d+1)}+\delta^2\eps \geq\expec{i}{\br{\bra{\Psi_i}Q\ket{\Psi_i}}^2}  \geq \delta^2\frac{(d-k)(d-k+1)}{d(d+1)}-\delta^2\eps.$$
\end{claim}
\begin{proof}
Since $Q$ is orthogonal to $\ket{0}$, we have that $$\expec{i}{\br{\bra{\Psi_i}Q\ket{\Psi_i}}^2} = \delta^2 \expec{i}{\br{\bra{x_i}Q\ket{x_i}}^2}  = \delta^2 \expec{i}{ \Tr\br{\br{Q\otimes Q}\br{\ketbra{x_i}\otimes \ketbra{x_i}}}}.$$ Using Eq~\eqref{eq:niceset}, we find that 

\begin{eqnarray*}
	&&\frac{(\mathrm{rk}~Q)^2+\mathrm{rk}~Q}{d(d+1)}+\eps\\
	&&=\Tr\br{\br{Q\otimes Q}\frac{P\otimes P + F}{d\br{d+1}}}+ \eps\\
	&&\geq\expec{i}{\Tr\br{\br{Q\otimes Q}\br{\ketbra{x_i}\otimes \ketbra{x_i}}}}\\
	&&\geq \Tr\br{\br{Q\otimes Q}\frac{P\otimes P + F}{d(d+1)}} - \eps \\
	&&= \frac{(\mathrm{rk}~Q)^2+\mathrm{rk}~Q}{d(d+1)}-\eps,
\end{eqnarray*}
where the both inequalities are from Eq.~\eqref{eqn:xx}. Using the value of $\mathrm{rk}~Q$, the claim follows.

\end{proof}

Using these claims, we now proceed to the proof of Lemma \ref{lem:smalloverlapbound}.

\begin{proof}[Proof of Lemma \ref{lem:smalloverlapbound}]
The variance of $\bra{\Psi_i}Q\ket{\Psi_i}$ can be upper bounded using Claims \ref{averageofprojector} and \ref{squareofprojector} as
\begin{eqnarray*}
	&&\expec{i}{\br{\bra{\Psi_i}Q\ket{\Psi_i}}^2}  - \expec{i}{\bra{\Psi_i}Q\ket{\Psi_i}}^2 \\
	&\leq& \delta^2\frac{(d-k)(d-k+1)}{d(d+1)}+\delta^2\eps - \br{\frac{\delta (d-k)}{d} - \delta\eps}^2\\
	&\leq&\delta^2\br{\eps- \eps^2 + 2\eps + \frac{(d-k)(d-k+1)}{d(d+1)} - \frac{(d-k)^2}{d^2}}\\
	&\leq&\delta^2\br{3\eps + \frac{3}{d}}.
\end{eqnarray*}
Now, using Chebyshev's inequality, we find that 
$$\Pr_i\Br{\bra{\Psi_i}Q\ket{\Psi_i} \leq \expec{i}{\br{\bra{\Psi_i}Q\ket{\Psi_i}} - \alpha} } \leq \frac{\delta^2}{\alpha^2}\br{3\eps + \frac{3}{d}}.$$
Using Claim \ref{averageofprojector}, the lemma follows.

\end{proof}

\subsection{Proof of Theorem \ref{thm:maingeneral}}
\label{subsec:maingen}

\begin{proof} 
	We use the construction above with $\eps = \frac{1}{d}$ as chosen in Lemma \ref{lem:uniformdmax}. Thus the set $\N_{\frac{1}{d}}$ has $8d^7$ elements. We prove each item as follows.
\begin{itemize}
\item From Lemma \ref{mutinfbound}, the von Neumann entropy of the average state is upper bounded by $$(\delta+\frac{1}{d})\log(d) + H(\delta) + 1 < \delta\log(d) + H(\delta)+2.$$
	
\item For the `one-way' part of the theorem, we set $\gamma = \sqrt{\eta}$ in Lemma \ref{lem:simplerlow} (with $r=1$) and apply Lemma \ref{lem:uniformdmax} with $\frac{\eta}{\gamma}\leftarrow \etaa$. We note that Lemma \ref{lem:uniformdmax} applies since $\frac{\eta}{\gamma} = \sqrt{\eta} \leq \frac{\delta}{8}$ by the choice of $\eta$. The resulting lower bound takes the form $$(1-\sqrt{\eta})^2 \log(\frac{d\delta}{64})-1 > (1-\sqrt{\eta})^2 \log(\frac{d\delta}{128}).$$
	
\item  The `round-dependent' part of the theorem follows in a similar manner, where we apply Lemma \ref{lem:simplerlow} with $r>1$. The lower bound we obtain takes the form $$(1-\sqrt{\eta})^2\frac{\log(\frac{d\delta}{128})}{2(\log r + 4)}> \frac{1}{20}\cdot\frac{\log(\frac{d\delta}{128})}{(\log r)}.$$
\begin{remark}
One can also use Equation \ref{one_way_r_way} to obtain a lower bound
$$\Omega\br{\log d\delta - r\log\log d\delta},$$ for small enough values of $r$. But it does not lead to a better lower bound in the interactive (round independent) part of the theorem.
\end{remark}
\item For the `round independent' part of the theorem, we proceed as follows. 
	Fix a communication protocol $\mathcal{P}$ with $r$ rounds and average error $\eta$. Define an odd number $\ell\defeq 2\lceil\frac{\log(d)}{\eta^2}\rceil + 1 > \frac{\log(d)}{\eta^2}$, which is assumed to be smaller than $r$. Let $\B$ denote the set of all instances $(x,i_1,i_2,\ldots, i_r)$ (input $x$ and messages exchanged) in which the protocol terminates before the round $\ell$. From Remark \ref{abortconvention}, such instances are of the form $(x,i_1,i_2,\ldots,1,1,\ldots, 1)$. Let $\G$ be the remaining set of instances. It is easy to infer that if $(x,i_r,i_{r-1},\ldots, i_1)\in \G$, then the number of bits exchanged in this instance is at least $\ell$ (as at least one bit must be exchanged in each round till round $\ell$). 
	
	Now we consider two cases. The first case is that $\sum_{(x,i_1,i_2,\ldots, i_r)\in \G}p(x)p^x_{i_r,i_{r-1},\ldots,i_1}> \eta^2$. Then the expected communication cost is lower bounded by $\eta^2\frac{\log(d)}{\eta^2} = \log(d)$. 
	
	The second case is that $\sum_{(i_1,i_2,\ldots, i_r)\in \G}p(x)p^x_{i_r,i_{r-1},\ldots,i_1}\leq \eta^2$. Let $\mathcal{P}'$ be the protocol that simulates $\mathcal{P}$ up to $\ell$ rounds.
	Namely, if round $\ell$ is reached then parties abort and Bob considers $\ket{0}_C$ as his output. If the protocol $\mathcal{P}'$ terminates before or in round $\ell$, then Bob outputs the same as in $\mathcal{P}$.
	
\suppress{If the protocol $\mathcal{P}'$ terminates after exchanging messages $(i_1,i_2\ldots i_k)$, for $k<\ell$, then following Remark \ref{abortconvention}, we define $p^x_{i_r,i_{r-1},\ldots,i_1}\defeq p^x_{i_k,i_{k-1},\ldots,i_1}$ iff $(i_r,i_{r-1},\ldots,i_1) = (1,1,\ldots,1,i_k,i_{k-1},\ldots,i_1)$ and 0 otherwise. We further define $\phi^{x,i_r,i_{r-1},\ldots,i_1}_{AB}\defeq \phi^{x,i_k,i_{k-1},\ldots,i_1}_{AB}$ iff $(i_r,i_{r-1},\ldots,i_1) = (1,1,\ldots,1,i_k,i_{k-1},\ldots,i_1)$ and $\ket{0}_{AB}$ otherwise. }
	Let $\tilde{\Phi}^x_C$ be the  state output from Bob, conditioned on input $x$. We have that 
	$$\tilde{\Phi}^x_C = \sum_{(x,i_1,i_2,\ldots, i_r)\in \B}p^x_{i_r,i_{r-1},\ldots,i_1}\tau^{x,i_r,i_{r-1},\ldots,i_1}_C + \beta\ketbra{0}_C$$ with $\beta\leq \eta^2$ by assumption. On the other hand, the final state $\Phi^x_C$ of the original protocol is 
	$$\Phi^x_C = \sum_{(x,i_1,i_2,\ldots, i_r)\in \B}p^x_{i_r,i_{r-1},\ldots,i_1}\tau^{x,i_r,i_{r-1},\ldots,i_1}_C + \sum_{(x,i_1,i_2,\ldots, i_r)\in \B}p^x_{i_r,i_{r-1},\ldots,i_1}\tau^{x,i_r,i_{r-1},\ldots,i_1}_C.$$
From the joint concavity of fidelity (Fact \ref{fact:fidelityconcave}), we obtain $\F(\tilde{\Phi}^x_C,\Psi^x_C)\geq 1-\eta^2$. This implies 
	$$\sum_x p(x)\F^2(\Psi^x,\tilde{\Phi}^x_C) \geq \sum_x p(x)\F^2(\Psi^x,\Phi^x_C)-\sum_xp(x)\|\tilde{\Phi}^x_C-\Psi^x_C\|_1\geq 1-\eta^2-2\eta\geq 1-3\eta.$$
Thus, $\mathcal{P}'$ is a protocol with $\ell$ rounds and average error $\sqrt{3\eta}<(\frac{\delta}{8})^2$. The expected communication cost of $\mathcal{P}'$ is lower bounded by (using the `round-dependent' part established above): 
	$$\frac{1}{20}\cdot\frac{\log(\frac{d\delta}{128})}{(\log \ell)}\geq \frac{1}{30}\cdot\frac{\log(\frac{d\delta}{128})}{(\log\log(d)-2\log\eta)}.$$ 
This is also the lower bound on the expected communication cost of $\mathcal{P}$, which proves the item.
\end{itemize}
\end{proof}

\section{Expected communication cost of quantum state redistribution}

The task of quantum state redistribution is formally defined as follows.

\begin{task}[\textbf{Quantum state redistribution}]\label{def:qstateredist}
Fix the registers $RBCA$ associated to a Hilbert space $\H_{RBCA}$ and an $\eps\in (0,1)$. Let a pure quantum state $\ketbra{\Psi}_{RBCA}\in \cD(\H_{RBCA})$ be shared among Alice (A,C), Bob (B) and Reference (R). Alice needs to transfer the register $C$ to Bob such that the final state between Alice (A), Bob (B,C) and Reference (R) is $\Psi'_{RBCA}\in \cD(\H_{RBCA})$. It is required that $\P(\Psi'_{RBCA},\Psi_{RBCA})\leq \eps^2$. 
\end{task}

The parameter $\eps$ will be referred to as the error of the protocol. We describe a  general structure of an interactive protocol for quantum state redistribution (Task \ref{def:qstateredist}) and its \textit{expected communication cost}. We assume that Alice and Bob only exchange classical messages via quantum teleportation. An $r$-round interactive protocol $\mathcal{P}$ (where $r$ is an odd number) with error $\eps$ and expected communication cost $C$ is as follows. It is also graphically depicted in Figure \ref{fig:interactivecoherent}. 

\bigskip

\begin{mdframed}
\bigskip
\label{stateredistdesc}
Let quantum state $\ket{\Psi}_{RBCA} \in \cD(\H_{RBCA})$ be shared among Alice $(A,C)$, Bob $(B)$ and Referee $(R)$. Alice and Bob possess the shared entanglement $\ketbra{\theta}_{E_AE_B}\in \cD(\H_{E_AE_B})$ in registers $E_A$ (with Alice) and $E_B$ (with Bob).

\begin{itemize}
\item Alice performs a measurement $\M=\set{M^1_{ACE_A},M^2_{ACE_A}\ldots },$ where $M^i_{ACE_A} \in \cL(\H_{ACE_A})$ and $\sum_i\br{M^i_{ACE_A}}^{\dagger}M^i_{ACE_A} = \id_{ACE_A}$. The probability of outcome $i_1$ is defined as  $p_{i_1}\defeq\Tr\br{M^{i_1}_{ACE_A}\Psi_{CA}\otimes\theta_{E_A}\br{M^{i_1}_{ACE_A}}^{\dagger}}$. Let $\phi^{i_1}_{RBACE_AE_B}$ be the global normalized quantum state, conditioned on this outcome. She sends the message $i_1$ to Bob.
 
\item Upon receiving the message $i_1$ from Alice, Bob performs a measurement  $$\M^{i_1}=\set{M^{1,i_1}_{BE_B},M^{2,i_1}_{BE_B},\ldots},$$ 
where $M^{i,i_1}_{BE_B} \in \cL(\H_{BE_B})$ and $\sum_i\br{M^i_{BE_B}}^{\dagger}M^i_{BE_B} = \id_{BE_B}$. The probability of outcome $i_2$ is $p_{i_2|i_1}\defeq \Tr\br{M^{i_2,i_1}_{BE_B}\phi^{i_1}_{BE_B}\br{M^i_{BE_B}}^{\dagger}}$. Let $\phi^{i_2,i_1}_{RBACE_AE_B}$ be the global normalized quantum state conditioned on this outcome $i_2$ and previous outcome $i_1$. Bob sends the message $i_2$ to Alice. 

\item Consider any odd round $1<k\leq r$. Let the measurement outcomes in all the previous rounds be $i_1,i_2,\ldots, i_{k-1}$ and the corresponding global normalized state be $\phi^{i_{k-1},i_{k-2},\ldots, i_1}_{RBACE_AE_B}$. Alice performs a measurement $$\M^{i_{k-1},i_{k-2},\ldots, i_2,i_1}=\set{M^{1,i_{k-1},i_{k-2},\ldots, i_2,i_1}_{ACE_A},M^{2,i_{k-1},i_{k-2},\ldots, i_2,i_1}_{ACE_A},\ldots},$$ where $M^{i,i_{k-1},i_{k-2},\ldots, i_2,i_1}_{ACE_A}\in \cL(\H_{ACE_A})$ and $\sum_i \br{M^{i,i_{k-1},i_{k-2},\ldots, i_2,i_1}_{ACE_A}}^{\dagger}M^{i,i_{k-1},i_{k-2},\ldots, i_2,i_1}_{ACE_A} = \id_{ACE_A}$. She obtains an outcome $i_k$ with probability $$p_{i_k|i_{k-1},i_{k-2},\ldots, i_2,i_1}\defeq\Tr\br{M^{i_k,i_{k-1},i_{k-2},\ldots, i_2,i_1}_{ACE_A}\phi^{i_{k-1},i_{k-2},\ldots, i_1}_{ACE_A}}.$$ Let the global normalized state conditioning on the outcome $i_k$ be $\phi^{i_k,i_{k-1},i_{k-2},\ldots, i_1}_{RBACE_BE_A}$. Alice sends the outcome $i_k$ to Bob. 

\item Consider an even round $2<k\leq r$. Let the measurement outcomes in previous rounds be $i_1,i_2\ldots i_{k-1}$ and the corresponding global normalized state be $\phi^{i_{k-1},i_{k-2},\ldots, i_1}_{RBACE_AE_B}$. Bob performs a measurement $$\M^{i_{k-1},i_{k-2},\ldots, i_2,i_1}=\set{M^{1,i_{k-1},i_{k-2},\ldots, i_2,i_1}_{BE_B},M^{2,i_{k-1},i_{k-2},\ldots, i_2,i_1}_{BE_B},\ldots},$$ where $M^{i,i_{k-1},i_{k-2},\ldots, i_2,i_1}_{BE_B}\in \cL(\H_{BE_B})$ and $\sum_i \br{M^{i,i_{k-1},i_{k-2},\ldots, i_2,i_1}_{BE_B}}^{\dagger}M^{i,i_{k-1},i_{k-2},\ldots, i_2,i_1}_{BE_B} = \id_{BE_B}$. He obtains an outcome $i_k$ with probability $$p_{i_k|i_{k-1},i_{k-2},\ldots, i_2,i_1}\defeq\Tr\br{M^{i_k,i_{k-1},i_{k-2},\ldots, i_2,i_1}_{BE_B}\phi^{i_{k-1},i_{k-2},\ldots, i_1}_{BE_B}}.$$ Let the global normalized state conditioning on the outcome $i_k$ be $\phi^{i_k,i_{k-1},i_{k-2},\ldots, i_1}_{RBACE_BE_A}$. Bob sends the outcome $i_k$ to Alice. 

\item After receiving the message $i_r$ from Alice at the end of round $r$, Bob applies a unitary $U^b_{i_r,i_{r-1},\ldots,i_1}:\H_{BE_B}\rightarrow \H_{BC_0T_B}$ such that $E_B\equiv C_0T_B$ and $C_0\equiv C$. Alice applies a unitary $U^a_{i_r,i_{r-1},\ldots,i_1}:\H_{ACE_A}\rightarrow \H_{ACE_A}$. Let $U_{i_r,i_{r-1},\ldots,i_1}\defeq U^a_{i_r,i_{r-1},\ldots,i_1}\otimes U^b_{i_r,i_{r-1},\ldots,i_1}$. Define $$\ket{\tau^{i_r,i_{r-1},\ldots,i_1}}_{RBACC_0T_BE_A}\defeq U_{i_r,i_{r-1},\ldots,i_1}\ket{\phi^{i_r,i_{r-1},\ldots,i_1}}_{RBACE_BE_A}.$$

\item For every $k\leq r$, define $$p_{i_1,i_2,\ldots,i_k}\defeq p_{i_1}\cdot p_{i_2|i_1}\cdot p_{i_3|i_2,i_1}\cdots p_{i_k|i_{k-1},i_{k-2},\ldots, i_1}.$$ The joint state in registers $RBC_0A$, after Alice and Bob's final unitaries, is $\Psi'_{RBC_0A}\defeq\sum_{i_r,i_{r-1},\ldots,i_1}p_{i_1,i_2,\ldots,i_r}\tau^{i_r,i_{r-1},\ldots,i_1}_{RBC_0A}$. It satisfies  $\P(\Psi'_{RBC_0A},\Psi_{RBC_0A})\leq \eps$ due to the correctness of the protocol. 
\end{itemize}
\vspace{0.1in}
\end{mdframed}
\centerline{Protocol $\mathcal{P}_3$}
\bigskip

The following fact is easily shown.
\begin{fact}
\label{expcost}
The expected communication cost of $\mathcal{P}_3$ is lower bounded by $$\sum_{i_1,i_2,\ldots,i_r}p_{i_1,i_2,\ldots,i_r}\log(i_1\cdot i_2\cdots i_r)$$
\end{fact}

Our main result of this section is the following theorem.
\begin{theorem}
	\label{thm:coherentmainagain}
	Fix a $p\in (0,1)$ and an $\eps \in [0,(\frac{1}{70})^{\frac{4}{1-p}}]$. There exists a pure quantum state $\Psi_{RBCA}\in \cD(\H_{RBCA})$ (that depends on $\eps$) such that any interactive entanglement-assisted communication protocol for its quantum state redistribution with error $\eps$ requires expected communication cost at least $\condmutinf{R}{C}{B}_{\Psi}\cdot(\frac{1}{\eps})^{p}$.
\end{theorem}



The proof of this theorem is given towards the end of this section. In order to facilitate the proof, we will introduce a coherent representation of the above protocol in the following lemma. 

\begin{lemma}
\label{cohlemma}
For every $k\leq r$, let $\O_k$ represent the set of all tuples $(i_1,i_2,\ldots,i_k)$ such that $\{i_1,i_2,\ldots,i_k\}$ is a sequence of measurement outcomes that occurs with non-zero probability up to $k$-th round of $\mathcal{P}_3$. 

There exist registers $M_1,M_2,\ldots,M_r$ and isometries $$\{U_{i_{k-1},i_{k-2},\ldots, i_2,i_1}: \H_{ACE_A}\rightarrow \H_{ACE_AM_k}| k >1, k \text{ odd }, (i_1,i_2\ldots i_{k-1})\in \O_{k-1}\},$$ $$\{U_{i_{k-1},i_{k-2},\ldots, i_2,i_1}: \H_{BE_B}\rightarrow \H_{BE_BM_k}|  k \text{ even }, (i_1,i_2\ldots i_{k-1})\in \O_{k-1}\}$$ and $U: \H_{ACE_A}\rightarrow \H_{ACE_AM_1}$, such that 
\begin{align*}
	&\ket{\Psi}_{RBCA}\ket{\theta}_{E_AE_B} \\
	&= U^{\dagger}\sum_{i_1,i_2,\ldots,i_r}\sqrt{p_{i_1,i_2,\ldots,i_r}}U^{\dagger}_{ i_1}U^{\dagger}_{ i_2,i_1}\cdots U^{\dagger}_{i_r,i_{r-1},\ldots,i_1}\ket{\tau^{i_r,i_{r-1},\ldots,i_1}}_{RBCAC_0T_BE_A}\ket{i_r}_{M_r}\ldots\ket{i_1}_{M_1},
\end{align*}
for some pure states $\ketbra{\tau^{i_r,i_{r-1},\ldots,i_1}}_{RBCAC_0T_BE_A}\in \cD(\H_{RBCAC_0T_BE_A})$. 
\end{lemma} 

The proof of this lemma is deferred to Appendix \ref{appen:cohlemma}. The proof is an adaptation of the arguments related to the convex-split lemma introduced in~\cite{ADJ14}. The statement of the lemma is simplified by introducing a family of controlled unitaries.

\begin{figure}
	\centering
	\begin{tikzpicture}[rotate=90, scale = 0.9]
	
	\draw [gray] (-1,16.5) -- (-1,-2);
	\draw [gray] (6,16.5) -- (6,13.6);
	\draw [gray] (6,12.2) -- (6,10.5);
	\draw [gray] (6,6.5) -- (6,4.5);
	\draw [gray] (6,-1) -- (6,-2);
	\node at (-2,16) {Referee};
	\node at (2,16) {Alice};
	\node at (8,16) {Bob};
	
	\node at (5,15.5) {$\Psi_{RBAC}$};
	\node at (-1.7,13) {$R$};
	\node at (2.2,13) {$A$};
	\node at (3.7,13) {$C$};
	\node at (8.8,13) {$B$};
	\draw [thick] (5, 14.7) -- (-1.5,13.5) -- (-1.5 , 0.5) -- (5,-1);
	\draw [thick] (5, 14.7) -- (2.5,13.5) -- (2.5,12);
	\draw [thick] (5, 14.7) -- (9,13.5) -- (9,9.5);
	\draw [thick] (5, 14.7) -- (3.5,13.5) -- (3.5,12);

	\node at (6,12.8) {$\theta_{E_AE_B}$};
	\node at (4.8,12.3) {$E_A$};
	\node at (7.3,12.3) {$E_B$};
	\draw [thick] (6,13.5) -- (4.5,13) -- (4.5,12);
	\draw [thick] (6,13.5) -- (7,13) -- (7,9.5);
	
	\draw [ultra thick] (2,12) rectangle (5.5,10.5);
	\node at (4,11.2) {$\M_{ACE_A}$};
	\node at (6,10.1) {$i_1$};
	\draw [thick] (2.5,10.5) -- (2.5, 3.5);
	\draw [thick] (3.5,10.5) -- (3.5, 3.5);
	\draw [thick] (4.5,10.5) -- (4.5, 3.5);
	\draw [thick] (5,10.5) -- (6.5, 9.5);
	
	\draw [ultra thick] (6,9.5) rectangle (9.5,8);
	\node at (8,8.8) {$\M^{i_1}_{BE_B}$};
	\draw [thick] (6.5,8) -- (5,6.5);
	\draw [thick] (9,8) -- (9,3);
	\draw [thick] (7,8) -- (7,3);
	\node at (6,7.1) {$i_2$};
	
	\draw [fill] (5.5,6) circle [radius=0.08];
	\draw [fill] (5.5,5.4) circle [radius=0.08];
	\draw [fill] (5.5,4.8) circle [radius=0.08];
	\node at (6,3.8) {$i_r$};

	\draw [ultra thick] (2,3.5) rectangle (5.5,2);
	\node at (4,2.7) {$U^a_{i_r\ldots i_1}$};
	\node at (2.2,1.7) {$A$};
	\node at (3.2,1.7) {$C$};
	\node at (4.2,1.7) {$E_A$};
	\draw [thick] (5,4.5) -- (5,3.5);
	\draw [thick] (2.5,2) -- (2.5, 0.5) -- (5, -1);
	\draw [thick] (3.5,2) -- (3.5, 0.5) -- (6, 0);
	\draw [thick] (4.5,2) -- (4.5, 0.5) -- (6,0);

	\draw [ultra thick] (6,1.5) rectangle (9.5,3);
	\node at (8,2.2) {$U^b_{i_r\ldots i_1}$};
	\draw [thick] (5,4.5) -- (6.5,3);
	\draw [thick] (7,1.5) -- (7,0.5) -- (6,0);
	\draw [thick] (8,1.5) -- (8,0.5) -- (5,-1);
	\draw [thick] (9,1.5) -- (9,0.5) -- (5,-1);
	\node at (6.7, 1.2) {$T_B$};
	\node at (7.7, 1.2) {$C_0$};
	\node at (8.7, 1.2) {$B$};
	
	\node at (5,-1.8) {$\Psi'_{RAC_0B}$};
	
	\end{tikzpicture}
	\caption{Graphical representation of an interactive protocol for Quantum state redistribution. The input state is $\Psi_{RBAC}$, the shared entanglement and the local registers are included in $E_A,E_B$ and the final state $\Psi'_{RAC_0B}$ satisfies $\F^2(\Psi'_{RAC_0B},\Psi_{RAC_0B})\geq 1-\eps^2$. The messages $i_1,i_2 \ldots $ are exchanged by Alice and Bob till the round $r$. The first measurement $\M \defeq \{M^{1}_{ACE_A},M^{2}_{ACE_A}\ldots\}$ is performed by Alice. Measurement $\M^{i_k,i_{k-1},\ldots i_1} = \{M^{1,i_{k},i_{k-1}\ldots i_2,i_1}_{BE_B},M^{2,i_{k},i_{k-1}\ldots i_2,i_1}_{BE_B}\ldots\}$ is performed by Alice (with registers $ACE_A$ if $k$ is even) and by Bob (with registers $BE_B$ if $k$ is odd). The final unitaries $U^a_{i_r \ldots i_1}$ and $U^b_{i_r \ldots i_1}$ are applied by Alice and Bob, respectively.}
	\label{fig:interactivecoherent}
\end{figure}
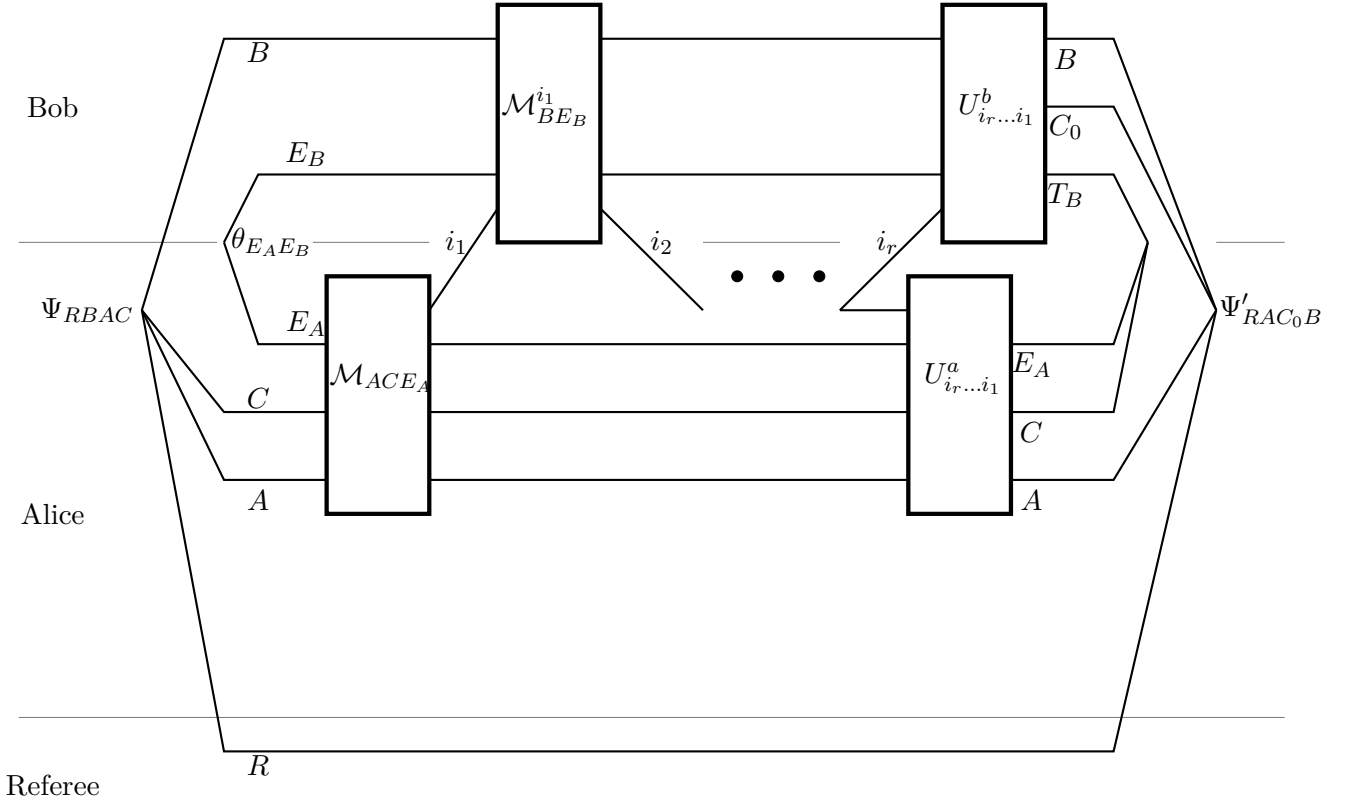

\begin{definition}
\label{shortunitaries}
We introduce the following definitions.

\begin{itemize}
\item Let $k>1$ be odd. Isometry $U_k: \H_{ACE_AM_1M_2\ldots M_{k-1}}\rightarrow \H_{ACE_AM_1M_2\ldots M_{k-1}M_k}$,  $$U_k \defeq \sum_{i_1,i_2\ldots i_{k-1}} \ketbra{i_1}_{M_1}\otimes \ketbra{i_2}_{M_2}\otimes\ldots\otimes\ketbra{i_{k-1}}_{M_{k-1}}\otimes U_{i_{k-1},i_{k-2},\ldots, i_2,i_1}.$$
\item  For $k$ even, Isometry $U_k: \H_{BE_BM_1M_2\ldots M_{k-1}}\rightarrow \H_{BE_BM_1M_2\ldots M_{k-1}M_k}$, $$U_k \defeq \sum_{i_1,i_2\ldots i_{k-1}} \ketbra{i_1}_{M_1}\otimes \ketbra{i_2}_{M_2}\otimes\ldots\otimes\ketbra{i_{k-1}}_{M_{k-1}}\otimes U_{i_{k-1},i_{k-2},\ldots, i_2,i_1}.$$
\item Unitary $U^a_{r+1}: \H_{ACE_AM_1M_2\ldots M_r} \rightarrow \H_{ACE_AM_1M_2\ldots M_r}$, $$U^a_{r+1} \defeq \sum_{i_1,i_2,\ldots,i_r} \ketbra{i_1}_{M_1}\otimes \ketbra{i_2}_{M_2}\otimes\ldots\otimes\ketbra{i_r}_{M_r}\otimes U^a_{i_r,i_{r-1},\ldots,i_1}.$$
\item Unitary $U^b_{r+1}: \H_{BE_BM_1M_2\ldots M_r} \rightarrow \H_{BC_0T_BM_1M_2\ldots M_r}$, $$U^b_{r+1} \defeq \sum_{i_1,i_2,\ldots,i_r} \ketbra{i_1}_{M_1}\otimes \ketbra{i_2}_{M_2}\otimes\ldots\otimes\ketbra{i_r}_{M_r}\otimes U^b_{i_r,i_{r-1},\ldots,i_1}.$$
\item Unitary $U_{r+1}: \H_{ACE_ABE_BM_1M_2\ldots M_r} \rightarrow \H_{ACE_ABC_0T_BM_1M_2\ldots M_r}$, $$U_{r+1} \defeq \sum_{i_1,i_2,\ldots,i_r} \ketbra{i_1}_{M_1}\otimes \ketbra{i_2}_{M_2}\otimes\ldots\otimes\ketbra{i_r}_{M_r}\otimes U_{i_r,i_{r-1},\ldots,i_1}.$$
\end{itemize}

\end{definition}

This leads to a more convenient representation of Lemma \ref{cohlemma}.  
\begin{cor}
\label{cohequation}
It holds that 
$$\ket{\Psi}_{RBCA}\ket{\theta}_{E_AE_B}=U^{\dagger}U_2^{\dagger}\cdots U_{r+1}^{\dagger} \sum_{i_1,i_2,\ldots,i_r}\sqrt{p_{i_1,i_2,\ldots,i_r}}\ket{\tau^{i_r,i_{r-1},\ldots,i_1}}_{RBCAC_0T_BE_A}\ket{i_r}_{M_r}\ldots\ket{i_1}_{M_1}.$$ and 
$$\P(\Psi_{RBC_0A},\sum_{i_1,i_2,\ldots,i_r}p_{i_1,i_2,\ldots,i_r}\tau^{i_r,i_{r-1},\ldots,i_1}_{RBC_0A})\leq \eps.$$
\end{cor}

\begin{proof}

The corollary follows immediately from Definition \ref{shortunitaries} and Lemma \ref{cohlemma}.

\end{proof}

\subsection{Construction}
\label{sec:lowerbound}

In this subsection, we obtain a lower bound on expected communication cost of quantum state redistribution by considering a class of states defined below.

Let the register $R$ be composed of two registers $R_A,R'$, such that $R\equiv R_AR'$. Let $|R_A| = d_a$ and $|R'|=|C|=|B|=d$. \begin{definition}
\label{staterediststate}
Define the quantum state $\ketbra{\Psi}_{RBCA}\in \cD(\H_{RBCA})$ as $$\ket{\Psi}_{RBCA}\defeq\frac{1}{\sqrt{d_a}}\sum_{a=1}^{d_a}\ket{a}_{R_A}\ket{a}_A\ket{\psi^a}_{R'BC},$$ where $$\ket{\psi^a}_{R'BC}=\sum_{j=1}^d\sqrt{e_j}\ket{u_j}_{R'}\ket{v_j(a)}_B\ket{w_j(a)}_C,$$ and $e_1\geq e_2\geq \ldots \geq e_d>0$ and $\sum_{i=1}^d e_i = 1$. Furthermore, each of the three sets  $\{\ket{u_1},\ldots,\ket{u_d}\}$, $\{\ket{v_1(a)},\ldots,\ket{v_d(a)}\}$, $\{\ket{w_1(a)},\ldots,\ket{w_d(a)}\}$ forms an orthonormal basis (second and third bases may depend arbitrarily on $a$) in $\H_{R'}, \H_{B}$ and $\H_{C}$ respectively. 
\end{definition}

The choice of the probability distribution $\{e_1,e_2,\ldots, e_d\}$ is given by the following lemma. 
\begin{lemma}
\label{lowentropy}
Fix a $\beta \geq 1$ and an integer $d>1$. There exists a probability distribution $\mu=\{e_1,e_2,\ldots, e_d\}$, with $e_1\geq e_2 \ldots \geq e_d$, such that $e_d = \frac{1}{d\beta}$ and Shannon entropy $\mathrm{H}(\mu)\leq 2\frac{\log(d)}{\beta}$.
\end{lemma}

\begin{proof}
Set $e_2=e_3=\ldots= e_d = \frac{1}{d\beta}$. Then $e_1=1-\frac{d-1}{d\beta}$. Using $x\log\br{\frac{1}{x}}\leq \frac{\log e}{e} < 1$ for all $x>0$, we can upper bound the entropy of the distribution as $$\sum_i e_i\log\br{\frac{1}{e_i}} = \br{1-\frac{d-1}{d\beta}}\log\br{\frac{1}{1-\frac{d-1}{d\beta}}} + \frac{d-1}{d\beta}\log\br{d\beta} < 2 + \frac{\log d}{\beta}\leq 2\frac{\log d}{\beta}.$$ 
\end{proof}

Given $\ket{\psi^a}_{R'BC}$ from Definition~\ref{staterediststate}, we define a `GHZ state' $\ketbra{\omega^a}_{R'BC}\in \cD(\H_{R'BC})$ corresponding to it as $$\ket{\omega^a}_{R'BC}\defeq \frac{1}{\sqrt{d}}\sum_{j=1}^d\ket{u_j}_{R'}\ket{v_j(a)}_B\ket{w_j(a)}_C.$$ Using this, we define the state $\ketbra{\omega_{RBCA}}\in \cD(\H_{RBCA})$ as $$\ket{\omega_{RBCA}}\defeq \frac{1}{\sqrt{d_a}}\sum_{a=1}^{d_a}\ket{a}_{R_A}\ket{a}_A\ket{\omega^a}_{R'BC}.$$
The following relation is easy to verify. 
\begin{equation}
\label{psiandomega}
\ket{\omega}_{RBCA} = \frac{1}{\sqrt{d_a\cdot d}}\Psi_R^{-\frac{1}{2}}\ket{\Psi}_{RBCA}
\end{equation}

\bigskip

The protocol $\mathcal{P}_3$ achieves quantum state redistribution of $\Psi_{RBCA}$ with error $\eps$ and expected communication cost $C$. The following lemma is a refined form of corollary \ref{cohequation}. Its proof is deferred to Appendix \ref{proof:goodcoh}.

\begin{lemma}
\label{goodcoh}
Given a $r$-round protocol $\mathcal{P}_3$ achieving the quantum state redistribution for $\ket{\Psi}_{RBCA}$ with $\set{U_i}_{1\leq i\leq r+1},\set{p_{i_1,\ldots,i_r}}, U_{r+1}^a,U_{r+1}^b$ as defined in Definition~\ref{shortunitaries}. There exists a probability distribution $\{p'_{i_1,i_2,\ldots,i_r}\}$ and pure states $\ketbra{\kappa^{i_r,i_{r-1},\ldots,i_1}}_{CE_AT_B}\in \cD(\H_{CE_AT_B})$ such that 
$$\P(\Psi_{RBCA}\otimes\theta_{E_AE_B}, \nu_{RBCAE_AE_B})\leq 2\sqrt{\eps},$$
where $$\ket{\nu}_{RBCAE_AE_B} \defeq U^{\dagger}U_2^{\dagger}\cdots U_{r+1}^{\dagger}\sum_{i_1,i_2,\ldots,i_r}\sqrt{p'_{i_1,i_2,\ldots,i_r}}\ket{\Psi}_{RBC_0A}\otimes\kappa^{i_r,i_{r-1},\ldots,i_1}_{CE_AT_B}\ket{i_r}_{M_r}\ldots\ket{i_1}_{M_1}.$$ Furthermore,
\[\sum_{i_1,i_2,\ldots,i_r}p'_{i_1,i_2,\ldots,i_r}\log\br{i_1\cdot i_2\cdots i_r}\leq\frac{C}{1-\eps}.\]  
\end{lemma}

We now use Lemma \ref{goodcoh} to prove the following lemma for the state $\omega_{RBCA}$. Recall that $e_d$ is the smallest eigenvalue of $\psi^a_{R'}$, independent of $a$.

\begin{lemma}
\label{convepr}
It holds that $$\P(\omega_{RBCA}\otimes\theta_{E_AE_B}, \omega_{RBC_0A}\otimes \bar{\nu}_{RBCAE_AE_B})\leq \sqrt{\frac{8\eps}{e_d\cdot d}},$$
where $$\ket{\bar{\nu}}_{RBCAE_AE_B} \defeq U^{\dagger}U_2^{\dagger}\cdots U_{r+1}^{\dagger}\sum_{i_1,i_2,\ldots,i_r}\sqrt{p'_{i_1,i_2,\ldots,i_r}}\ket{\omega}_{RBC_0A}\otimes\kappa^{i_r,i_{r-1},\ldots,i_1}_{CE_AT_B}\ket{i_r}_{M_r}\ldots\ket{i_1}_{M_1}.$$
Furthermore,
\[\sum_{i_1,i_2,\ldots,i_r}p'_{i_1,i_2,\ldots,i_r}\log\br{i_1\cdot i_2\cdots\cdot i_r}\leq\frac{C}{1-\eps}.\] 
\end{lemma}
\begin{proof}

Define a completely positive map $\tilde{\E}:\H_R\rightarrow \H_R$ as $ \tilde{\E}(\rho)\defeq \frac{e_d}{d_a}(\Psi^{-\frac{1}{2}}_R\rho\Psi^{-\frac{1}{2}}_R)$, which is trace non-increasing since $\Psi^{-1}_R \preceq \frac{d_a}{e_d}\text{I}_R$. From Equation~\eqref{psiandomega}, we have $$\tilde{\E}(\Psi_{RBCA}) = e_d\cdot d\cdot\omega_{RBCA}$$ and $$\tilde{\E}(\nu_{RBCAE_AE_B}) = e_d\cdot d\cdot\bar{\nu}_{RBCAE_AE_B}.$$

Consider,
\begin{eqnarray*}
2\sqrt{\eps} &\geq& \P(\Psi_{RBCA}\otimes\theta_{E_AE_B}, \nu_{RBCAE_AE_B})\\ &&\text{(Lemma \ref{goodcoh})}\\&\geq& \P(\tilde{\E}(\Psi_{RBCA})\otimes\theta_{E_AE_B},\tilde{\E}(\nu_{RBCAE_AE_B}))\\ && (\text{Fact \ref{fact:monotonequantumoperation}}) \\ &=& \P(d\cdot e_d\cdot\omega_{RBCA}\otimes\theta_{E_AE_B},d\cdot e_d\cdot \bar{\nu}_{RBCAE_AE_B}).
\end{eqnarray*}
Using Fact \ref{scalarpurified}, we thus obtain $$\P(\omega_{RBCA}\otimes\theta_{E_AE_B}, \bar{\nu}_{RBCAE_AE_B})\leq \sqrt{\frac{8\eps}{d\cdot e_d}}.$$
Furthermore, the probabilities $p'_{i_1,i_2,\ldots,i_r}$ are same as in Lemma \ref{goodcoh}. This completes the proof. 

\end{proof}

Now we exhibit an interactive entanglement-assisted communication protocol for the Quantum state redistribution of $\omega_{RBCA}$ with bounded worst case quantum communication cost. Its proof is deferred to Appendix \ref{append:exptoworst}.

\begin{lemma}
\label{exptoworst}
Fix a $\mu\in (0,1)$. There exists an entanglement-assisted $r$-round quantum communication protocol for the quantum state redistribution of $\omega_{RBCA}$ with worst case quantum communication cost at most $\frac{2C}{\mu(1-\eps)}$ and error at most $\sqrt{\frac{8\eps}{e_d\cdot d}}+\sqrt{\mu}$.
\end{lemma}

The next lemma obtains a lower bound on the worst case quantum communication cost of the Quantum state redistribution of $\omega_{RBCA}$.

\begin{lemma}
\label{redistworstcase}
Let $d>2^{18}$ be the dimension of the register $B$. Then the worst case quantum communication cost of any interactive entanglement-assisted quantum state redistribution protocol of the state $\omega_{RBCA}$, with error $\delta \in (0, \frac{1}{6})$, is at least $\frac{1}{6}\log(d)$.
\end{lemma}
\begin{proof}
	As shown in~\cite{Berta14}( Section $5$, Proposition $2$), the worst quantum communication cost for the quantum state redistribution of the state $\omega_{RBCA}$, with error $\delta$ is lower bounded by $$\frac{1}{2}\br{\imaxdelta{R}{BC}_{\omega}-\imax{R}{B}_{\omega}}.$$
	From Definition~\ref{staterediststate},  $\omega_{RBC}=\frac{1}{d_a}\sum_{a=1}^{d_a} \ketbra{a}_{R_A}\otimes\omega^a_{R'BC}$ is a \textit{classical-quantum} state. Consider, 
\begin{eqnarray*}
\imaxdelta{R}{BC}_{\omega} &\geq& \inf_{\rho_{RBC}\in \ball{\delta}{\omega_{RBC}}}\mutinf{R}{BC}_{\rho}\\ &\geq& \inf_{\rho_{R}\in \ball{\delta}{\omega_{R}}}S(\rho_R) + \inf_{\rho_{BC}\in \ball{\delta}{\omega_{BC}}}S(\rho'_{BC}) - \sup_{\rho_{RBC}\in \ball{\delta}{\omega_{RBC}}}S(\rho_{RBC}) \\ &\geq& \mutinf{R}{BC}_{\omega} - 3\delta\log(d) - 3 \quad (\text{Fact \ref{fact:fannes}})\\ &\geq& \frac{1}{d_a}\sum_{a}\mutinf{R'}{BC}_{\omega^a} - 3\delta\log(d)-3 \quad (\text{Fact \ref{cqmutinf}}) \\ &=& 2\log(d)-3\delta\log(d)-3.
\end{eqnarray*}
To bound $\imax{R}{B}_{\omega}$, note that $\omega_{RB}=\frac{1}{d\cdot d_a}\sum_{a=1}^{d_a}\sum_{j=1}^d\ketbra{a}_{R_A}\otimes\ketbra{u_j}_{R'}\otimes\ketbra{v_j(a)}_{B}$ is also a \textit{classical-quantum} state. Using Fact \ref{fact:cqimax}, we obtain $\imax{R}{B}_{\omega} \leq \log(|B|) = \log(d)$. Thus, the worst case quantum communication cost is lower bounded by $$\frac{1}{2}\br{\imaxdelta{R}{BC}_{\omega}-\imax{R}{B}_{\omega}}\geq \frac{\log(d)-3\delta\log(d)-3}{2}=\frac{1-3\delta}{2}\log(d) - 1.5 > \frac{1}{6}\log(d),$$ for $d>2^{18}$.
\end{proof}

Now, we are in a position to prove Theorem \ref{thm:coherentmainagain}.

\begin{proof}[Proof of Theorem~\ref{thm:coherentmainagain}]
	Suppose there exists an $r$-round entanglement-assisted communication protocol $\mathcal{P}_3$ for the Quantum state redistribution of the state $\ket{\Psi_{RBCA}}$ with error $\eps$ and expected communication cost at most $\condmutinf{R}{C}{B}_{\Psi}\cdot (\frac{1}{\eps})^p$. Then we show a contradiction for $p< 1$.
	
Let $d>2^{18}$, $\mu\defeq 32\cdot \epsilon^{\frac{1-p}{2}}$ and $\beta\defeq\frac{128}{\mu\epsilon^p}$. We choose $\{e_1,e_2\ldots e_d\}$ (Definition \ref{staterediststate}) as constructed in Lemma \ref{lowentropy}. Thus, $$\condmutinf{R}{C}{B}_{\Psi}\leq 2S(\Psi_C)\leq 4\frac{\log(d)}{\beta} \quad \text{(Fact \ref{informationbound})}.$$ 
	Fix a $\mu \in (0,1)$. From Lemma \ref{exptoworst}, there exists a communication protocol $\mathcal{P}'$ for the quantum state redistribution of $\omega_{RBCA}$, with error at most $\sqrt{\mu}+ \sqrt{8\beta\eps}=8\sqrt{2}\cdot\eps^{\frac{1-p}{4}}$ and the worst case quantum communication cost at most 
$$\frac{2\cdot\condmutinf{R}{C}{B}_{\Psi}}{\mu(1-\eps)}\cdot (\frac{1}{\eps})^p\leq 8\frac{\log(d)}{\beta\mu(1-\eps)}\cdot (\frac{1}{\eps})^p \leq 16\frac{\log(d)}{\beta\mu}\cdot (\frac{1}{\eps})^p,$$
where the last inequality holds since $\eps<1/2$. Note that $\eps \in [0, (\frac{1}{70})^{\frac{4}{1-p}}]$. Thus, we have a protocol for the quantum state redistribution of $\omega_{RBCA}$, with error at most $8\sqrt{2}\cdot\eps^{\frac{1-p}{4}} < \frac{1}{6}$ and worst case communication at most $\frac{1}{8}\log(d)$, in contradiction with Lemma~\ref{redistworstcase}.  
\end{proof}

\section{Example for one-shot quantum channel simulation}
\label{reverseshannon}

An entanglement-assisted protocol $\mathcal{P}_4$ for communicating a $m$-bit classical message over a quantum channel $\E: \cL(\H_A)\rightarrow \cL(\H_B)$ is as follows. Alice holds a register $A'$ and Bob holds a register $B'$ such that the quantum state in these registers is $\ket{\theta}_{A'B'}$. Based on an input $x\in [2^m]$, Alice applies a map $\mathcal{A}_x: \cL(\H_{A'})\rightarrow \cL(\H_A)$ and sends the register $A$ through the channel $\E$. Upon receiving the output from the channel $\E$, Bob applies a decoding map $\mathcal{B}: \cL(\H_{BB'}) \rightarrow \cL(\H_{X'})$, where $X'$ is the output register with $|X'|= 2^m$ and a fixed basis $\{\ket{x}\}_{x\in [2^m]}$. Let the output of the protocol be $\rho^x_{X'}$. The error of the protocol, given an input $x$, is $\P^2(\ketbra{x}_{X'}, \rho^x_{X'})$. The worst case error of protocol is defined as
$\eta_{\mathcal{P}_4}\defeq \text{max}_{x}\P^2(\ketbra{x}_{X'}, \rho^x_{X'})$. 

The one-shot $\eta$-error entanglement-assisted classical capacity of channel $\E$ is the largest $m$ such that there exists a protocol $\mathcal{P}_4$ that communicates a $m$-bit classical message over $\E$ with worst case error $\eta_{\mathcal{P}_4}\leq \eta$.

The simulation of a quantum channel $\E$ can be regarded as a converse to the transmission of a message using the channel. The task is defined as follows.

\begin{task}[\textbf{Entanglement-assisted quantum channel simulation}]
\label{def:qreverseshannon}
Fix a quantum channel $\E:\cL(\H_A)\rightarrow \cL(\H_B)$ and $\eta\in (0,1)$. Alice receives input $\rho_A\in \cD(\H_A)$. Bob needs to output a quantum state $\sigma_B \in \cD(\H_B)$ such that $\P(\E(\rho_A),\sigma_B)\leq \eta$.
\end{task}

The parameter $\eta$ is referred to as the error of simulation of quantum channel $\E$. A general (possibly interactive) protocol $\Q$ for simulating quantum channel $\E$ follows the description as given in Section \ref{sec:interactive} with the final condition of the correctness changed, accordingly. Its expected communication cost is defined in the similar manner. For a given protocol $\Q$, we define the \textit{simulation cost} of $\Q$ as the maximum expected communication cost of $\Q$ over all inputs $\rho_A$.

In this section, we construct an example of a quantum channel which requires large amount of expected communication for its simulation. It uses the construction presented in Section \ref{sec:example} and we carry over the notations from Section \ref{sec:example}. 

As assumed in Lemma \ref{lem:uniformdmax}, we set $\eps=\frac{1}{d}$ and $m=\N_{\eps} = 8d^7$. Let $A$ be a register such that $|A|=m$ and let $\H_A$ have the basis $\set{\ket{j}}_{1\leq j\leq m}$. We consider the following classical-quantum channel $\E:\cL(\H_A)\rightarrow \cL(\H_B)$, $$\E(\rho)=\sum_{j} \bra{j}\rho\ket{j}\ketbra{\Psi_j}.$$ The entanglement-assisted classical capacity of this channel is given in~\cite{Bennett02} as $$\mathcal{C}(\E)\defeq \text{max}_{\rho_A}\mutinf{R}{A}_{\E(\rho_{RA})},$$ where $\rho_{RA}$ is a purification of $\rho_A$ on register $R$.

\begin{lemma}
	The entanglement-assisted classical capacity $\mathcal{C}(\E)$ of channel $\E$ is upper bounded by $\delta\log(d)+H(\delta)+2$.
\end{lemma}
\begin{proof}
	Since $\E$ is a classical-quantum channel, the quantum state $\E(\rho_{RA})$ is a classical quantum state. More specifically, we have 
	$$\E(\rho_{RA})= \sum_{j}\bra{j}_A\ket{\rho}_{RA}\bra{\rho}_{RA}\ket{j}_A\otimes\ketbra{\Psi_j}_B =  \sum_{j} \rho^{\frac{1}{2}}\ketbra{j}_R\rho^{\frac{1}{2}}\otimes\ketbra{\Psi_j}_B.$$ 
	
	Note that $\mutinf{R}{A}_{\E(\rho_{RA})}\leq S(\E(\rho_A))$. This implies
	$$\mathcal{C}(\E)\leq \text{max}_{\rho_A} \mathrm{S}\br{\sum_j\bra{j}\rho\ket{j}\ketbra{\Psi_j}} = \text{max}_{\mu(j)}\mathrm{S}\br{\sum_j \mu(j)\ketbra{\Psi_j}}.$$
	The equality above holds since the quantity $\bra{j}\rho\ket{j}$ can be viewed as a probability distribution over indices in $\{1,2\ldots m\}$. Now, let $\mu(j),\mu'(j)$ be any two distributions over $\{1,2\ldots m\}$. Then it holds that 
	\begin{eqnarray*}
		&&\mathrm{S}\br{\sum_j (\alpha\mu(j)+(1-\alpha}\mu'(j))\ketbra{\Psi_j})\\&=&\mathrm{S}\br{\alpha\sum_j \mu(j)\ketbra{\Psi_j}+(1-\alpha)\sum_j \mu'(j)\ketbra{\Psi_j}}\\&\geq& \alpha\cdot \mathrm{S}\br{\sum_j\mu(j)\ketbra{\Psi_j}}+(1-\alpha)\cdot \mathrm{S}\br{\sum_j\mu'(j)\ketbra{\Psi_j}}
	\end{eqnarray*}
	
	This implies that the desired distribution achieving the maximum is the uniform distribution over $\{1,2\ldots K\}$. But for such a distribution, we have $$\text{max}_{\mu(j)}\mathrm{S}\br{\sum_j \mu(j)\ketbra{\Psi_j}}=\mathrm{S}\br{\sum_j\frac{1}{m}\ketbra{\Psi_j}}.$$ As computed in Lemma~\ref{mutinfbound}, this entropy is upper bounded by $\delta\log(d)+H(\delta)+2$. This proves the lemma.
\end{proof}

Now we are in a position to prove the main result of this section.

\begin{theorem}
	\label{reverseshannonseparationgeneral}
	Fix a positive integer $d>4$ and a $\delta \in (0,\frac{1}{4})$. There exists a register $A$ such that $|A| = 8d^7$ and a channel $\E:\cL(\H_A)\rightarrow \cL(\H_B)$ with one-shot $\eta$-error entanglement-assisted classical capacity upper bounded by $$ \frac{\delta\log(d)+H(\delta)+H(\eta)+2}{1-\eta},$$ such that for any protocol achieving the simulation of above channel with error at most $\eta$, following holds:
	\begin{itemize}
		\item If the protocol is one-way and $\eta\in (0,\br{\frac{\delta}{8}}^2)$, then simulation cost of the protocol is at least $(1-\sqrt{\eta})^2 \log\br{\frac{d\delta}{128}}$.
		\item If the protocol is interactive with $r$-rounds and $\eta\in (0,\br{\frac{\delta}{8}}^2)$, the simulation cost is lower bounded by $\frac{1}{20}\cdot\frac{\log\frac{d\delta}{128}}{\log r}$.
		\item If $\eta\in (0, \br{\frac{\delta}{8}}^4)$ and the protocol is interactive, the simulation cost of the protocol is lower bounded by 
		$$\frac{1}{30}\cdot\frac{\log\frac{d\delta}{128}}{\log\log d-2\log\eta}.$$
	\end{itemize}  
\end{theorem}

\begin{proof}
	From Lemma 30 in \cite{WehnerMatthews14}, the one-shot entanglement-assisted classical capacity of a channel $\E$ with error $\eta$ is upper bounded by $$\frac{\mathcal{C}(\E)+H(\eta)}{1-\eta} < \frac{\delta\log(d)+H(\delta)+H(\eta)+2}{1-\eta}.$$ 
	
	On the other hand, consider a protocol that simulates the action of the channel $\E$ with error $\eta$. Simulation cost of the protocol is maximum expected communication cost over all inputs to the channel. Thus, it is lower bounded by expected communication cost when inputs are given according to a fixed distribution. We consider a distribution over inputs as follows: Alice receives a $\ketbra{j}$ with probability $\frac{1}{m}$. The channel outputs the state $\ketbra{\Psi_j}$ and the protocol must simulate this output with error at most $\eta$. It is now easy to observe that the lower bound of Theorem \ref{thm:maingeneral} applies for respective choice of parameters, which proves the theorem.
\end{proof}

\section{Conclusion}\label{sec:conclusion}

In this work, we have studied the expected communication cost of three quantum tasks: Classical-quantum state transfer (Task \ref{def:quantumhuffman}), classical-quantum state splitting (Task \ref{def:entsharing}) and quantum state redistribution (Task \ref{def:qstateredist}). We have given a nearly optimal characterization of the expected communication cost of the classical-quantum state splitting task. For its special case of classical-quantum state transfer and the task of quantum state redistribution, we have shown large separations between the expected communication cost and the quantum information cost (which is the worst case quantum communication cost in the asymptotic and i.i.d. setting).  As an application of our main results, we show that in the one-shot setting, quantum channels cannot be simulated with an expected communication cost as small as their entanglement-assisted classical capacity. 

We have following questions for the future research.
\begin{itemize}
\item Theorem~\ref{thm:maingeneral} has a dependence on the number of rounds. We get rid of this dependence at the expense of weaker lower bound on the expected communication cost. But we conjecture that our techniques are not optimal and interactions cannot reduce the expected communication costs.

\item Is there an operational interpretation of the fundamental quantum information theoretic quantities in the one-shot settings?  Our result says that expected communication cost is not the right notion, but naturally we cannot rule out other notions.
\end{itemize}

\section*{Acknowledgment}

A.A., A.G. and P.Y. would like to thank the Institute for Mathematical Science, Singapore for their hospitality and their organized workshop "Semidefinite and Matrix Methods for Optimization and Communication". A.A would like to thank the Institute for Quantum Computing, University of Waterloo for their hospitality, where part of this work was done. We thank Dave Touchette for helpful comments on the manuscript. A.A. thanks Rahul Jain and Guo Yalei for helpful discussions. A.G. thanks Mohammed Bavarian and Henry Yuen for helpful discussions.

A.A. is supported by the National Research Foundation, Prime Minister’s Office, Singapore and the Ministry of Education, Singapore under the Research Centres of Excellence programme. Most of this research was done when A.G. was a graduate student at Princeton University and his research was partially supported by NSF grants CCF-1149888 and CCF-1525342, a
Simons fellowship for graduate students in theoretical computer science and a Siebel scholarship.  AWH is funded by NSF grants
CCF-1629809 and CCF-1452616. P.Y. is supported by the Department of Defense. Part of this work was done when P.Y.
was a postdoctoral fellow at IQC supported by NSERC and CIFAR.

\bibliographystyle{alpha}
\bibliography{references}
\appendix

\section{Proof of Claim~\ref{exactoverlap}}
\label{proof_claims}

\begin{proof}[Proof of Claim \ref{exactoverlap}]
Let $\ket{\lambda_i}$ be the state that achieves the infimum in the definition of $S^{\etaa}(\Psi_i||Q^{-})$. We know that $\ket{\lambda_i}$ satisfies $|\braket{\lambda}{\Psi_i}|^2>1-\etaa$ and also minimizes the overlap with the subspace $Q^{-}$. Intuitively, this state must lie in the span of two vectors $\{Q^{-}\ket{\Psi_i},Q^{+}\ket{\Psi_i}\}$, which is shown to be true below. 

Let us assume $$\ket{\lambda_i}=aQ^{-}\ket{\Psi_i}+bQ^{+}\ket{\Psi_i}+c\ket{\theta},$$ where $\ket{\theta}$ is a normalized vector orthogonal to $\{Q^{-}\ket{\Psi_i},Q^{+}\ket{\Psi_i}\}$. Then we have: 
\begin{equation}
\label{conditions}
|a|^2\bra{\Psi_i}Q^{-}\ket{\Psi_i}+|b|^2\bra{\Psi_i}Q^{+}\ket{\Psi_i}+|c|^2=1,\quad |a\bra{\Psi_i}Q^{-}\ket{\Psi_i}+b\bra{\Psi_i}Q^{+}\ket{\Psi_i}|>\sqrt{1-\etaa}
\end{equation}
 where the equality is a  normalization condition and the inequality indicates that the overlap between $\ket{\lambda_i}$ and $\ket{\Psi_i}$ is at least $\sqrt{1-\etaa}$. We need to minimize the function
\begin{equation}
\label{objective}
 \bra{\lambda_i}Q^{-}\ket{\lambda_i}=\bra{\lambda_i}(aQ^{-}\ket{\Psi_i}+cQ^{-}\ket{\theta})=|a|^2\bra{\Psi_i}Q^{-}\ket{\Psi_i}+|c|^2\bra{\theta}Q^{-}\ket{\theta},
\end{equation}
 where we have used $\bra{\Psi_i}Q^{-}Q^{-}\ket{\theta}=0$. 

First we show that $a,b,c$ can be chosen to be real. We can $c$ to be real without loss of generality as only $\abs{c}^2$ appears in Equations \ref{conditions} and \ref{objective}. The only place where $a,b$ appear as complex numbers is in the constraint $|a\bra{\Psi_i}Q^{-}\ket{\Psi_i}+b\bra{\Psi_i}Q^{+}\ket{\Psi_i}|>\sqrt{1-\etaa}$ (Equation \ref{objective}). Let $a=a_R+ia_I, b=b_R+ib_I$, where $a_R, a_I, b_R, b_I$ are real numbers. Then 
\begin{eqnarray*}
&&|a\bra{\Psi_i}Q^{-}\ket{\Psi_i}+b\bra{\Psi_i}Q^{+}\ket{\Psi_i}|^2 \\&&= (a_R\bra{\Psi_i}Q^{-}\ket{\Psi_i}+b_R\bra{\Psi_i}Q^{+}\ket{\Psi_i})^2+ (a_I\bra{\Psi_i}Q^{-}\ket{\Psi_i}+b_I\bra{\Psi_i}Q^{+}\ket{\Psi_i})^2\\&&= |a|^2\bra{\Psi_i}Q^{-}\ket{\Psi_i}^2+
 |b|^2\bra{\Psi_i}Q^{+}\ket{\Psi_i}^2 + 2(a_Rb_R+a_Ib_I)\bra{\Psi_i}Q^{-}\ket{\Psi_i}\bra{\Psi_i}Q^{+}\ket{\Psi_i}\\&&\leq
|a|^2\bra{\Psi_i}Q^{-}\ket{\Psi_i}^2+
 |b|^2\bra{\Psi_i}Q^{+}\ket{\Psi_i}^2 + 2(\sqrt{a_R^2+a_I^2}\sqrt{b_R^2+b_I^2})\bra{\Psi_i}Q^{-}\ket{\Psi_i}\bra{\Psi_i}Q^{+}\ket{\Psi_i}\\ &&= |a|^2\bra{\Psi_i}Q^{-}\ket{\Psi_i}^2+
 |b|^2\bra{\Psi_i}Q^{+}\ket{\Psi_i}^2 + 2|a||b|\bra{\Psi_i}Q^{-}\ket{\Psi_i}\bra{\Psi_i}Q^{+}\ket{\Psi_i}\\ &&= (|a|\bra{\Psi_i}Q^{-}\ket{\Psi_i}+|b|\bra{\Psi_i}Q^{+}\ket{\Psi_i})^2.
\end{eqnarray*}
Thus, changing $a,b$ to $|a|,|b|$ does not change the objective function (Equation \ref{objective}) and ensures that the constraints in Eq.~\eqref{conditions} are still satisfied. As a result, we can assume that $a$ and $b$ are both real, without loss of generality.

\bigskip
\begin{figure}[ht]
\centering
\begin{tikzpicture}[xscale=1.2,yscale=1.4]
\draw (3,0)--(3,4);
\draw (-1,2)--(7,2);
\draw[very thick] (1,2) to [out=90,in=90] (5,2); 
\draw[very thick] (5,2) to [out=270,in=270] (1,2);
\draw[very thick] (3,3.5) -- (5.5,2);
\draw[very thick] (0.5,2) -- (3,0.5);
\node at (4.5,1) {\textbf{E}};
\node at (5.5,2.3) {\textbf{L}};
\node at (2.6,0.5) {\textbf{L'}};
\node at (6.8,1.8) {$a$};
\node at (3.2,3.9) {$b$};
\node at (3,2.8) {$(b_1,0)$};
\node at (3,1.1) {$(-b_1,0)$};
\node at (3,3.6) {$(b_2,0)$};
\node at (3,0.2) {$(-b_2,0)$};
\node at (4.5,2) {$(0,a_1)$};
\node at (1.6,2) {$(0,-a_1)$};
\node at (6,2) {$(0,a_2)$};
\node at (0,2) {$(0,-a_2)$};
\end{tikzpicture}
\caption{Plot of the constraints}
 \label{fig:constraints}
\end{figure}
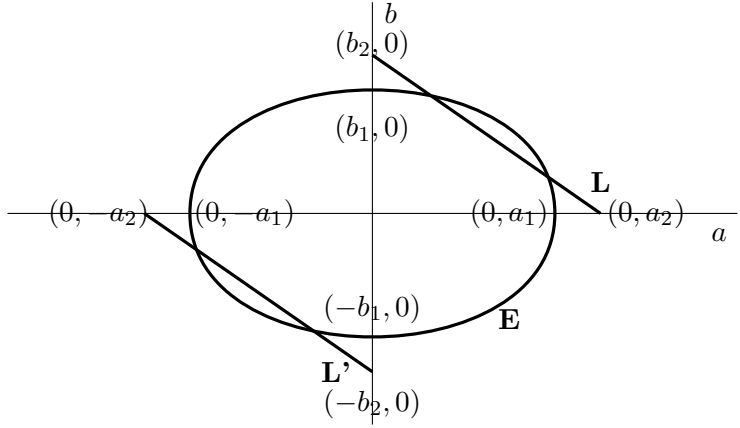
\bigskip

To find the optimal solution for Eq.~\eqref{objective}, we first fix $c$ and minimize $a^2$ with the constraints  $$a^2\bra{\Psi_i}Q^{-}\ket{\Psi_i}+b^2\bra{\Psi_i}Q^{+}\ket{\Psi_i}=1-c^2,\quad |a\bra{\Psi_i}Q^{-}\ket{\Psi_i}+b\bra{\Psi_i}Q^{+}\ket{\Psi_i}|>\sqrt{1-\etaa}.$$ 
We plot these constraints on $(a,b)$ plane in Figure~\ref{fig:constraints}. The ellipse $$E:~ a^2\bra{\Psi_i}Q^{-}\ket{\Psi_i}+b^2\bra{\Psi_i}Q^{+}\ket{\Psi_i}=1-c^2$$ intersects $a$-axis at $|a_1|=\sqrt{\frac{1-c^2}{\bra{\Psi_i}Q^{-}\ket{\Psi_i}}}$ and intersects $b$-axis at $|b_1|=\sqrt{\frac{1-c^2}{\bra{\Psi_i}Q^{+}\ket{\Psi_i}}}$. The lines $$L:~a\bra{\Psi_i}Q^{-}\ket{\Psi_i}+b\bra{\Psi_i}Q^{+}\ket{\Psi_i}=\sqrt{1-\etaa},\quad L':~a\bra{\Psi_i}Q^{-}\ket{\Psi_i}+b\bra{\Psi_i}Q^{+}\ket{\Psi_i}=-\sqrt{1-\etaa}$$ intersect $a$-axis at $|a_2|=\frac{\sqrt{1-\etaa}}{\bra{\Psi_i}Q^{-}\ket{\Psi_i}}$ and intersect $b$-axis at $|b_2|=\frac{\sqrt{1-\etaa}}{\bra{\Psi_i}Q^{+}\ket{\Psi_i}}$.  

First note that if $c^2>\etaa$, then there is no solution. For this, consider 
\begin{eqnarray*}
	&&1-\etaa < \br{a\bra{\Psi_i}Q^{-}\ket{\Psi_i}+b\bra{\Psi_i}Q^{+}\ket{\Psi_i}}^2\\
	&\leq& \br{\bra{\Psi_i}Q^{-}\ket{\Psi_i}+\bra{\Psi_i}Q^{+}\ket{\Psi_i}}\br{a^2\bra{\Psi_i}Q^{-}\ket{\Psi_i}+b^2\bra{\Psi_i}Q^{+}\ket{\Psi_i}}\\
	&=&\br{a^2\bra{\Psi_i}Q^{-}\ket{\Psi_i}+b^2\bra{\Psi_i}Q^{+}\ket{\Psi_i}} = 1-c^2.
\end{eqnarray*}
Thus, we assume that $c^2\leq \etaa$. Consider the first quadrant in Figure \ref{fig:constraints}. We can observe from the plot that $a=0$ is the minimum value of $a^2$ whenever the ellipse $E$ intersects the $b$-axis above the line $L$. This occurs when $$\sqrt{\frac{1-c^2}{\bra{\Psi_i}Q^{+}\ket{\Psi_i}}} > \frac{\sqrt{1-\etaa}}{\bra{\Psi_i}Q^{+}\ket{\Psi_i}} \rightarrow \bra{\Psi_i}Q^{+}\ket{\Psi_i} > \frac{1-\etaa}{1-c^2}.$$ But this is obvious, since the condition implies $\bra{\Psi_i}Q^{+}\ket{\Psi_i}>1-\etaa$, in which case there is a vector in $Q^{+}$ with high overlap with $\ket{\Psi_i}$ and hence the objective function is $0$.

So we assume that $\bra{\Psi_i}Q^{+}\ket{\Psi_i} < 1-\etaa$, in which case, for all $c$, the ellipse $E$ intersects the $b$-axis below the line $L$. To find the point of intersection, we simultaneously solve the equations for the line and the ellipse, that is $$a^2\bra{\Psi_i}Q^{-}\ket{\Psi_i}+b^2\bra{\Psi_i}Q^{+}\ket{\Psi_i}=1-c^2,\quad a\bra{\Psi_i}Q^{-}\ket{\Psi_i}+b\bra{\Psi_i}Q^{+}\ket{\Psi_i}=\sqrt{1-\etaa}.$$ The values of $a,b$ thus obtained are 
$$a=\sqrt{1-\etaa}-\sqrt{\frac{\bra{\Psi_i}Q^{+}\ket{\Psi_i}(\etaa-c^2)}{\bra{\Psi_i}Q^{-}\ket{\Psi_i}}}, \quad b=\sqrt{1-\etaa}+\sqrt{\frac{\bra{\Psi_i}Q^{-}\ket{\Psi_i}(\etaa-c^2)}{\bra{\Psi_i}Q^{+}\ket{\Psi_i}}}.$$
It is easy to verify that the solution satisfies the equations above. The other solution is with the signs reversed.  

Thus, we conclude that whenever $\bra{\Psi_i}Q^{+}\ket{\Psi_i} < 1-\etaa$, the minimum $|a|^2\bra{\Psi_i}Q^{-}\ket{\Psi_i}+|c|^2\bra{\theta}Q^{-}\ket{\theta}$ is 
$$\br{\sqrt{1-\etaa}-\sqrt{\frac{\bra{\Psi_i}Q^{+}\ket{\Psi_i}(\etaa-c^2)}{\bra{\Psi_i}Q^{-}\ket{\Psi_i}}}}^2\bra{\Psi_i}Q^{-}\ket{\Psi_i}+c^2\bra{\theta}Q^{-}\ket{\theta}.$$ 
This quantity is monotonically increasing with $c$. Hence the  quantity above is minimized when $c=0$. This justifies our intuition that the optimal vector lies in the plane $\{Q^{+}\ket{\Psi_i},Q^{-}\ket{\Psi_i}\}$. With this, we have found an overall minimum to be
$$\br{\sqrt{1-\etaa}-\sqrt{\frac{\bra{\Psi_i}Q^{+}\ket{\Psi_i}\etaa}{\bra{\Psi_i}Q^{-}\ket{\Psi_i}}}}^2\bra{\Psi_i}Q^{-}\ket{\Psi_i} =\br{\sqrt{(1-\etaa)\bra{\Psi_i}Q^{-}\ket{\Psi_i}}-\sqrt{\bra{\Psi_i}Q^{+}\ket{\Psi_i}\etaa}} ^2.$$ 
This proves the claim. 
\end{proof}

%

\section{Proof of Lemma \ref{cohlemma}}
\label{appen:cohlemma}
\begin{proof}

Fix an odd $k>1$. Let the messages prior to the $k$-th round be $(i_1,i_2\ldots i_{k-1})$. As defined in the protocol $\mathcal{P}$, the global quantum state before the $k$-th round is $\phi^{i_{k-1},i_{k-2}\ldots i_1}_{RBCAE_AE_B}$. Alice performs the measurement $$\{M^{1,i_{k-1},i_{k-2}\ldots i_2,i_1}_{ACE_A},M^{2,i_{k-1},i_{k-2}\ldots i_2,i_1}_{AXE_A}\ldots\}.$$ This implies
\begin{eqnarray}
\label{roundconvsplit}
\phi^{i_{k-1},i_{k-2}\ldots i_1}_{RBE_B} &=& \sum_{i_k} \Tr_{ACE_A}\br{M^{i_k,i_{k-1},i_{k-2}\ldots i_2,i_1}_{ACE_A}\phi^{i_{k-1},i_{k-2}\ldots i_1}_{RBCAE_BE_A}\br{M^{i_k,i_{k-1},i_{k-2}\ldots i_2,i_1}_{ACE_A}}^{\dagger}} \nonumber\\&=& \sum_{i_k} p_{i_k|i_{k-1},i_{k-2}\ldots i_2,i_1}\frac{\Tr_{ACE_A}\br{M^{i_k,i_{k-1},i_{k-2}\ldots i_2,i_1}_{ACE_A}\phi^{i_{k-1},i_{k-2}\ldots i_1}_{RBCAE_BE_A}\br{M^{i_k,i_{k-1},i_{k-2}\ldots i_2,i_1}_{ACE_A}}^{\dagger}}}{p_{i_k|i_{k-1},i_{k-2}\ldots i_2,i_1}}\nonumber\\&=& \sum_{i_k} p_{i_k|i_{k-1},i_{k-2}\ldots i_2,i_1}\phi^{i_k,i_{k-1},i_{k-2}\ldots i_2,i_1}_{RBE_B}.
\end{eqnarray}

A purification of $\phi^{i_{k-1},i_{k-2}\ldots i_1}_{RBE_B}$ on registers $RBCAE_BE_A$ is $\phi^{i_{k-1},i_{k-2}\ldots i_1}_{RBCAE_BE_A}$. Introduce a register $M_{k}$ (of sufficiently large dimension) and consider the pure state
\[\sum_{i_k}\sqrt{p_{i_k|i_{k-1},i_{k-2}\ldots i_2,i_1}}\ket{\phi^{i_k,i_{k-1},i_{k-2}\ldots i_2,i_1}}_{RBCAE_BE_A}\ket{i_k}_{M_k},\]
which purifies 
$$\sum_{i_k} p_{i_k|i_{k-1},i_{k-2}\ldots i_2,i_1}\phi^{i_k,i_{k-1},i_{k-2}\ldots i_2,i_1}_{RBE_B}$$ on register $RBCAE_BE_AM_k$.

By Uhlmann's theorem (Fact \ref{uhlmann}), there exists an isometry $U_{i_{k-1},i_{k-2}\ldots i_2,i_1}: \H_{ACE_A}\rightarrow \H_{ACE_AM_k}$ such that
\begin{equation}
\label{aliceunitary} 
U_{i_{k-1},i_{k-2}\ldots, i_1}\ket{\phi^{i_{k-1},i_{k-2}\ldots i_1}}_{RBCAE_BE_A} = \sum_{i_k}\sqrt{p_{i_k|i_{k-1},i_{k-2}\ldots, i_1}}\ket{\phi^{i_k,i_{k-1},i_{k-2}\ldots, i_1}}_{RBCAE_BE_A}\ket{i_k}_{M_k}.
\end{equation}
For $k=1$, introduce a register $M_1$ of sufficiently large dimension. Similar argument implies that there exists an isometry $U: \H_{ACE_A}\rightarrow \H_{ACE_AM_1}$ such that

\begin{equation}
\label{aliceunitary1} 
U\ket{\Psi}_{RBACE_BE_A} = \sum_{i_1}\sqrt{p_{i_1}}\ket{\phi^{i_1}}_{RBACE_BE_A}\ket{i_1}_{M_1}.
\end{equation}
For the case that $k$ is even, introduce register $M_k$ of sufficiently large dimension. Again by the similar argument, there exists an isometry $U_{i_{k-1},i_{k-2}\ldots i_2,i_1}: \H_{BE_B}\rightarrow \H_{BE_BM_k}$ such that 

\begin{equation}
\label{bobunitary} 
U_{i_{k-1},i_{k-2}\ldots, i_1}\ket{\phi^{i_{k-1},i_{k-2}\ldots i_1}}_{RBCAE_BE_A} = \sum_{i_k}\sqrt{p_{i_k|i_{k-1},i_{k-2}\ldots, i_1}}\ket{\phi^{i_k,i_{k-1},i_{k-2}\ldots, i_1}}_{RBCAE_BE_A}\ket{i_k}_{M_k}.
\end{equation}
Now, we recursively use Eq.~\eqref{aliceunitary}\eqref{aliceunitary1}\eqref{bobunitary}. Consider,
\begin{eqnarray*}
&&\ket{\Psi}_{RBCA}\ket{\theta}_{E_AE_B} = U^{\dagger}\sum_{i_1}\sqrt{p_{i_1}}\ket{\phi^{i_1}}_{RBCAE_BE_A}\ket{i_1}_{M_1} \\ &=& U^{\dagger}\sum_{i_1}\sqrt{p_{i_1}}U^{\dagger}_{i_1}\sum_{i_2}\sqrt{p_{i_2|i_1}}\ket{\phi^{i_2,i_1}}_{RBCAE_BE_A}\ket{i_2}_{M_2}\ket{i_1}_{M_1}\\ &=& U^{\dagger}\sum_{i_1,i_2}\sqrt{p_{i_1,i_2}}U_{i_1}^{\dagger}\ket{\phi^{i_2,i_1}}_{RBCAE_BE_A}\ket{i_2}_{M_2}\ket{i_1}_{M_1} \\&=& U^{\dagger}\sum_{i_1,i_2\ldots i_r}\sqrt{p_{i_1,i_2\ldots i_r}}U^{\dagger}_{i_1}U^{\dagger}_{i_2,i_1}\ldots U^{\dagger}_{i_r,i_{r-1}\ldots i_1}\ket{\tau^{i_r,i_{r-1}\ldots i_1}}_{RBCAB_0T_BE_A}\ket{i_r}_{M_r}\ldots\ket{i_1}_{M_1}.
\end{eqnarray*}
The last equality follows by recursion. This completes the proof.

\end{proof}

\section{Proof of Lemma \ref{goodcoh}}
\label{proof:goodcoh}
For the notational convenience in the proof below, we will sometimes represent the purified distance $P(\ketbra{v}, \ketbra{w})$ between two pure states $\ket{v}, \ket{w}$ as $\P(\ket{v}, \ket{w})$.
\begin{proof}
Let $\B$ be the set of tuples $(i_1,i_2\ldots i_r)$ for which $\F^2\br{\Psi_{RBC_0A},\tau^{i_r,i_{r-1}\ldots i_1}_{RBC_0A}}\leq 1-\eps$. Let $\G$ be the remaining set of tuples. From Corollary~\ref{cohequation} and the fact that $\Psi_{RBC_0A}$ is pure, it holds that
$$\sum_{i_1,i_2\ldots i_r}p_{i_1,i_2\ldots i_r}\F^2\br{\Psi_{RBC_0A},\tau^{i_r,i_{r-1}\ldots i_1}_{RBC_0A}}\geq 1-\eps^2.$$
Thus, $$(1-\eps)\sum_{(i_1,i_2\ldots i_r)\in\B}p_{i_1,i_2\ldots i_r}+ \sum_{(i_1,i_2\ldots i_r)\in\G}p_{i_1,i_2\ldots i_r} \geq 1-\eps^2,$$ which implies $\sum_{(i_1,i_2\ldots i_r)\in\B}p_{i_1,i_2\ldots i_r}\leq \eps$ and $\sum_{(i_1,i_2\ldots i_r)\in \G}p_{i_1,i_2\ldots i_r}\geq 1-\eps$. 

Define $p'_{i_1,i_2\ldots i_r}\defeq \frac{p_{i_1,i_2\ldots i_r}}{\sum_{{i_1,i_2\ldots i_r}\in \G} p_{i_1,i_2\ldots i_r}}$, if $(i_1,i_2\ldots i_r) \in \G$ and $p'_{i_1,i_2\ldots i_r}\defeq 0$ if $(i_1,i_2\ldots i_r)\in \B $. For all $(i_1,i_2\ldots i_r)\in \G$, $\F^2\br{\Psi_{RBC_0A},\tau^{i_r,i_{r-1}\ldots i_1}_{RBC_0A}}\geq 1-\eps$. Thus by Uhlmann's theorem (Fact \ref{uhlmann}),  there exists a pure state $\kappa^{i_r,i_{r-1}\ldots i_1}_{CE_AT_B} \in \cD(\H_{CE_AT_B})$ such that 
\begin{equation}
\label{goodproperty} 
\F^2\br{\Psi_{RBC_0A}\otimes\kappa^{i_r,i_{r-1}\ldots i_1}_{CE_AT_B},\tau^{i_r,i_{r-1}\ldots i_1}_{RBCAC_0T_BE_A}}\geq 1-\eps. 
\end{equation}
Consider,
\begin{eqnarray*}
&&\P\br{\sum_{i_1,i_2\ldots i_r}\sqrt{p_{i_1,i_2\ldots i_r}}\tau^{i_r,i_{r-1}\ldots i_1}_{RBCAC_0T_BE_A}\ket{i_r}_{M_r}\ldots\ket{i_1}_{M_1},\atop\sum_{i_1,i_2\ldots i_r}\sqrt{p'_{i_1,i_2\ldots i_r}}\tau^{i_r,i_{r-1}\ldots i_1}_{RBCAC_0T_BE_A}\ket{i_r}_{M_r}\ldots\ket{i_1}_{M_1}}\nonumber \\&=& \sqrt{1-\br{\sum_{i_1,i_2\ldots i_r}\sqrt{p_{i_1,i_2\ldots i_r}p'_{i_1,i_2\ldots i_r}}}^2} = \sqrt{1-\br{\sum_{i_1,i_2\ldots i_r\in \G}p_{i_1,i_2\ldots i_r}}}\leq \sqrt{\eps}
\end{eqnarray*}
 and 
\begin{eqnarray*}
&&\P\br{\sum_{i_1,i_2\ldots i_r}\sqrt{p'_{i_1,i_2\ldots i_r}}\tau^{i_r,i_{r-1}\ldots i_1}_{RBCAC_0T_BE_A}\ket{i_r}_{M_r}\ldots\ket{i_1}_{M_1},\atop \sum_{i_1,i_2\ldots i_r}\sqrt{p'_{i_1,i_2\ldots i_r}}\Psi_{RBC_0A}\otimes\kappa^{i_r,i_{r-1}\ldots i_1}_{CE_AT_B}\ket{i_r}_{M_r}\ldots\ket{i_1}_{M_1}} \\&=& \sqrt{1-\br{\sum_{i_1,i_2\ldots i_r}p'_{i_1,i_2\ldots i_r}\F\br{\tau^{i_r,i_{r-1}\ldots i_1}_{RBCAC_0T_BE_A},\Psi_{RBC_0A}\otimes\kappa^{i_r,i_{r-1}\ldots i_1}_{CE_AT_B}}}^2} \leq \sqrt{\eps}. \quad (\text{Equation \ref{goodproperty}})
\end{eqnarray*}
These together imply, using triangle inequality for purified distance (Fact \ref{fact:trianglepurified}),
\begin{eqnarray*}
\P\br{\sum_{i_1,i_2\ldots i_r}\sqrt{p_{i_1,i_2\ldots i_r}}\tau^{i_r,i_{r-1}\ldots i_1}_{RBCAC_0T_BE_A}\ket{i_r}_{M_r}\ldots\ket{i_1}_{M_1}\atop\sum_{i_1,i_2\ldots i_r}\sqrt{p'_{i_1,i_2\ldots i_r}}\Psi_{RBC_0A}\otimes\kappa^{i_r,i_{r-1}\ldots i_1}_{CE_AT_B}\ket{i_r}_{M_r}\ldots\ket{i_1}_{M_1}}\leq 2\sqrt{\eps}. 
\end{eqnarray*}
Thus, from corollary \ref{cohequation}, we have $$\P\br{\Psi_{RBCA}\otimes\theta_{E_AE_B},U^{\dagger}U_2^{\dagger}\ldots U_{r+1}^{\dagger}\sum_{i_1,i_2\ldots i_r}\sqrt{p'_{i_1,i_2\ldots i_r}}\Psi_{RBC_0A}\otimes\kappa^{i_r,i_{r-1}\ldots i_1}_{CE_AT_B}\ket{i_r}_{M_r}\ldots\ket{i_1}_{M_1}}\leq 2\sqrt{\eps}.$$
Furthermore, we have 
\begin{eqnarray*}
\sum_{i_1,i_2\ldots i_r} p'_{i_1,i_2\ldots i_r}\log(i_1\cdot i_2\ldots i_r) &\leq& \frac{1}{1-\eps}\sum_{{i_1,i_2\ldots i_r}\in \G}p_{i_1,i_2\ldots i_r}\log(i_1\cdot i_2\ldots i_r) \\ &\leq& \frac{1}{1-\eps}\sum_{i_1,i_2\ldots i_r}p_{i_1,i_2\ldots i_r}\log(i_1\cdot i_2\ldots i_r) =\frac{C}{1-\eps}.
 \end{eqnarray*}
This completes the proof.
\end{proof}

\section{Proof of Lemma \ref{exptoworst}}
\label{append:exptoworst}
\begin{proof}
From Lemma~\ref{convepr}, we have $$\P\br{\omega_{RBCA}\otimes\theta_{E_AE_B}, \bar{\nu}_{RBCAE_AE_B}}\leq \sqrt{\frac{8\eps}{e_d\cdot d}},$$ and $$\sum_{i_1,i_2\ldots i_r}p'_{i_1,i_2\ldots i_r}\log(i_1\cdot i_2\ldots i_r)\leq \frac{C}{1-\eps},$$ where 
$$\ket{\bar{\nu}}_{RBCAE_AE_B}=U^{\dagger}U_2^{\dagger}\ldots U_{r+1}^{\dagger}\sum_{i_1,i_2\ldots i_r}\sqrt{p'_{i_1,i_2\ldots i_r}}\omega_{RBC_0A}\otimes\kappa^{i_r,i_{r-1}\ldots i_1}_{CE_AT_B}\ket{i_r}_{M_r}\ldots\ket{i_1}_{M_1},$$ 
as defined in Lemma~\ref{convepr}.
Let $\B'$ be the set of tuples $(i_1,i_2\ldots i_r)$ which satisfy $i_1\cdot i_2\ldots\cdot i_r>2^{\frac{C}{(1-\eps)\mu}}$. Let $\G'$ be the set of rest of the tuples. Then $$\frac{C}{(1-\eps)} > \sum_{i_1,i_2\ldots i_r\in \B'}p'_{i_1,i_2\ldots i_r}\log(i_1\cdot i_2\ldots i_r) > \frac{C}{(1-\eps)\mu}\sum_{i_1,i_2\ldots i_r\in \B'}p'_{i_1,i_2\ldots i_r}.$$ This implies $\sum_{i_1,i_2\ldots i_r\in \B'}p'_{i_1,i_2\ldots i_r} < \mu$.
 Define a new probability distribution $q_{i_1,i_2\ldots i_r}\defeq \frac{p'_{i_1,i_2\ldots i_r}}{\sum_{(i_1,i_2\ldots i_r)\in \G'}p'_{i_1,i_2\ldots i_r}}$ for all $(i_1,i_2\ldots i_r)\in \G'$ and $q_{i_1,i_2\ldots i_r}=0$ for all $(i_1,i_2\ldots i_r)\in \B'$. 
Define $$\ket{\pi}_{RBCAE_AE_B}\defeq U^{\dagger}U_2^{\dagger}\ldots U_{r+1}^{\dagger}\sum_{i_1,i_2\ldots i_r\in \G'}\sqrt{q_{i_1,i_2\ldots i_r}}\omega_{RBC_0A}\otimes\kappa^{i_r,i_{r-1}\ldots i_1}_{CE_AT_B}\ket{i_r}_{M_r}\ldots\ket{i_1}_{M_1}.$$
Consider,
\begin{eqnarray*}
	&&\P\br{\pi_{RBCAE_AE_B}, \bar{\nu}_{RBCAE_AE_B}}\\
	&=&\sqrt{1-\br{\sum_{i_1,i_2\ldots i_r}\sqrt{p'_{i_1,i_2\ldots i_r}q_{i_1,i_2\ldots i_r}}}^2} = \sqrt{1-\sum_{(i_1,i_2\ldots i_r) \in \G'}p'_{i_1,i_2\ldots i_r}}\leq \sqrt{\mu}.
\end{eqnarray*}
Thus, the triangle inequality for purified distance (Fact \ref{fact:trianglepurified}) implies
\begin{eqnarray}
\label{eq:closegoodstate}
\P\br{\omega_{RBCA}\otimes\theta_{E_AE_B}, \pi_{RBCAE_AE_B}}\leq \sqrt{\frac{8\eps}{e_d\cdot d}}+\sqrt{\mu}.
\end{eqnarray}
Let $\T$ be the set of all tuples $(i_1,i_2\ldots i_k)$ (with $k\leq r$) that satisfy the following property: there exists a set of positive integers $\{i_{k+1},i_{k+2}\ldots i_r\}$ such that $(i_1,i_2\ldots i_k,i_{k+1}\ldots i_r)\in \G'$. Consider the following protocol $\mathcal{P'}$.
\bigskip
\begin{mdframed}
\bigskip

\textbf{Input:} A quantum state in registers $RBCAE_AE_B$.

\begin{itemize}
\item Alice applies the isometry $U:\H_{ACE_A}\rightarrow \H_{ACE_AM_1}$ (Definition \ref{shortunitaries}). She introduces a register $M'_1\equiv M_1$ in the state $\ket{0}_{M'_1}$ and performs the following unitary $W_1: \H_{M_1M'_1}\rightarrow \H_{M_1M'_1}$: $$W_1\ket{i}_{M_1}\ket{0}_{M'_1}=\ket{i}_{M_1}\ket{i}_{M'_1} \quad \text{if } (i)\in \T \quad, \quad W_1\ket{i}_{M_1}\ket{0}_{M'_1}=\ket{i}_{M_1}\ket{0}_{M'_1} \quad \text{if } (i)\notin \T.$$ 
She sends $M'_1$ to Bob.
\item Bob introduces a register $M'_2\equiv M_2$ in the state $\ket{0}_{M'_2}$. If he receives $\ket{0}_{M'_1}$ from Alice, he performs no operation. Else he applies the isometry $U_2: \H_{BE_BM'_1}\rightarrow \H_{BE_BM'_1M_2}$ and then performs the following unitary $W_2: \H_{M'_1M_2M'_2}\rightarrow \H_{M'_1M_2M'_2}$: $$W_1\ket{i}_{M'_1}\ket{j}_{M_2}\ket{0}_{M'_2}=\ket{i}_{M'_1}\ket{j}_{M_2}\ket{j}_{M'_2} \quad \text{if } (i,j)\in \T$$ and  $$W_1\ket{i}_{M'_1}\ket{j}_{M_2}\ket{0}_{M'_2}=\ket{i}_{M'_1}\ket{j}_{M_2}\ket{0}_{M'_2} \quad \text{if }(i,j)\notin \T.$$ 

He sends $M'_2$ to Alice. 

\item For every odd round $k>1$, Alice introduces a register $M'_k\equiv M_k$ in the state $\ket{0}_{M'_k}$. If she receives $\ket{0}_{M'_{k-1}}$ from Bob, she performs no further operation. Else, she applies the isometry $$U_k: \H_{ACE_AM_1M'_2M_3\ldots M'_{k-1}}\rightarrow \H_{ACE_AM_1M'_2M_3\ldots M'_{k-1}M_k}$$ and performs the following unitary $W_k: \H_{M_1M'_2\ldots M'_{k-1}M_kM'_k}\rightarrow \H_{M_1M'_2\ldots M'_{k-1}M_kM'_k}$: 
$$W_k\ket{i_1}_{M_1}\ket{i_2}_{M'_2}\ldots\ket{i_k}_{M_k}\ket{0}_{M'_k}=
\ket{i_1}_{M_1}\ket{i_2}_{M'_2}\ldots\ket{i_k}_{M_k}\ket{i_k}_{M'_k} \quad \text{if } (i_1,i_2\ldots i_k)\in \T$$ and  $$W_k\ket{i_1}_{M_1}\ket{i_2}_{M'_2}\ldots\ket{i_k}_{M_k}\ket{0}_{M'_k}=
\ket{i_1}_{M_1}\ket{i_2}_{M'_2}\ldots\ket{i_k}_{M_k}\ket{0}_{M'_k} \quad \text{if }(i_1,i_2\ldots i_k)\notin \T.$$ 

She sends $M'_k$ to Bob.

\item For every even round $k>2$, Bob introduces a register $M'_k\equiv M_k$ in the state $\ket{0}_{M'_k}$. If he receives $\ket{0}_{M'_{k-1}}$ from Alice, he performs no further operation.. Else, he applies the isometry $U_k: \H_{BE_BM'_1M_2M'_3\ldots M'_{k-1}}\rightarrow \H_{BE_BM'_1M_2M'_3\ldots M'_{k-1}M_k}$ and performs the following unitary $W_k: \H_{M'_1M_2\ldots M'_{k-1}M_kM'_k}\rightarrow \H_{M'_1M_2\ldots M'_{k-1}M_kM'_k}$: 
$$W_k\ket{i_1}_{M'_1}\ket{i_2}_{M_2}\ldots\ket{i_k}_{M_k}\ket{0}_{M'_k}=
\ket{i_1}_{M'_1}\ket{i_2}_{M_2}\ldots\ket{i_k}_{M_k}\ket{i_k}_{M'_k} \quad \text{if } (i_1,i_2\ldots i_k)\in \T$$ and  $$W_k\ket{i_1}_{M'_1}\ket{i_2}_{M_2}\ldots\ket{i_k}_{M_k}\ket{0}_{M'_k}=
\ket{i_1}_{M'_1}\ket{i_2}_{M_2}\ldots\ket{i_k}_{M_k}\ket{0}_{M'_k} \quad \text{if }(i_1,i_2\ldots i_k)\notin \T.$$ 

He sends $M'_k$ to Alice.

\item After round $r$, if Bob receives $\ket{0}_{M'_r}$ from Alice, he performs no further operation. Else he applies the unitary $U^b_{r+1}: \H_{BE_BM'_1M_2M'_3\ldots M'_r}\rightarrow \H_{BC_0T_BM'_1M_2M'_3\ldots M'_r}$. Alice applies the unitary $U^a_{r+1}: \H_{ACE_AM_1M'_2M_3\ldots M_r}\rightarrow \H_{ACE_AM_1M'_2M_3\ldots M_r}$. They trace out all of their registers except $A,B,C_0$. 
\end{itemize}
\bigskip
\end{mdframed}
\bigskip
Let $\E: \cL(\H_{RBCAE_AE_B})\rightarrow \cL(\H_{RBC_0A})$ be the quantum map generated by $\mathcal{P'}$. For any $k$, if any of the parties receive the state $\ket{0}_{M'_k}$, let this event be called \textit{abort}.

We show the following claim.
\begin{claim} It holds that
 $\E(\pi_{RBCAE_AE_B})=\omega_{RBC_0A}$
\end{claim}
\begin{proof}
We argue that the protocol never aborts when acting on $\pi_{RBCAE_AE_B}$. Consider the first round of the protocol. Define the projector $\Pi\defeq \sum_{i: (i)\notin \T}\ketbra{i}_{M_1}$. From Definition \ref{shortunitaries}, it is clear that the isometry $U^{\dagger}_2U^{\dagger}_3\ldots U^{\dagger}_{r+1}$ is of the form $\sum_i \ketbra{i}_{M_1}\otimes V_i$, for some set of isometries $\{V_i\}$ . Thus, from the definition of $\pi_{RBCAE_AE_B}$ (in which the summation is only over the tuples $(i_1,i_2\ldots i_r)\in \G'$), it holds that $$\Pi U\pi_{RBCAE_AE_B}=0.$$  This implies that Bob does not receive the state $\ket{0}_{M'_1}$ and hence he does not aborts.

Similar argument applies to other rounds, which implies that the protocol never aborts. Thus, the state at the end of the protocol is $$  \Tr_{CE_AT_B}(U_{r+1}U_{r}\ldots U_2U\pi_{RBCAE_AE_B}) = \omega_{RBC_0A}.$$
\end{proof}
Thus, from Equation \ref{eq:closegoodstate}, it holds that $$\P(\E(\omega_{RBCA}\otimes\theta_{E_AE_B}),\omega_{RBC_0A})\leq \sqrt{\frac{8\eps}{e_d\cdot d}}+\sqrt{\mu}.$$
Quantum communication cost of the protocol is at most $$\text{max}_{(i_1,i_2\ldots i_r)\in \G'}(\log((i_1+1)\cdot (i_2+1)\ldots (i_r+1)) \leq 2\cdot\text{max}_{(i_1,i_2\ldots i_r)\in \G'}(\log(i_1\cdot i_2\ldots i_r)\leq \frac{2C}{(1-\eps)\mu}.$$ This completes the proof.  
\end{proof}

\end{document}